\documentclass[a4paper,runningheads,envcountsect]{llncs}

\usepackage[utf8]{inputenc}
\usepackage[T1]{fontenc}
\usepackage[english]{babel}
\usepackage{amsmath, amssymb,mathtools}
\usepackage[margin=1in]{geometry}
\usepackage{microtype}
\usepackage[colorlinks, citecolor=blue]{hyperref}
\usepackage{IEEEtrantools}
\usepackage{array}
\usepackage{breakcites}
\usepackage{subcaption}

\usepackage{amsthm}

\usepackage{tikz}
\usetikzlibrary{automata,positioning,matrix}

\numberwithin{equation}{section}
\numberwithin{figure}{section}
 
\title{Different Strokes in Randomised Strategies:\texorpdfstring{\\}{ }Revisiting Kuhn's Theorem under Finite-Memory Assumptions\thanks{James C.~A.~Main is an F.R.S.-FNRS Research Fellow and Mickael Randour is an F.R.S.-FNRS Research Associate. Both authors are members of the TRAIL institute. This work has been supported by the Fonds de la Recherche Scientifique – FNRS under Grant n° T.0188.23 (PDR ControlleRS).}}
\titlerunning{Revisiting Kuhn's Theorem Under Finite-Memory Assumptions}

\author{{James C.~A.} Main \and Mickael Randour}

\institute{F.R.S.-FNRS \& UMONS -- Université de Mons, Belgium}

\tikzset{
  >=stealth,
  left sided/.style={
    draw=none,
    append after command={
      [shorten <= -0.5\pgflinewidth]
      (\tikzlastnode.north west) edge[dashed](\tikzlastnode.south west)
    }
  },
  two sided/.style={
    draw=none,
    append after command={
      [shorten <= -0.5\pgflinewidth]
      (\tikzlastnode.north west) edge[dashed](\tikzlastnode.south west)
      (\tikzlastnode.north east) edge[dashed](\tikzlastnode.south east)
    }
  },
  right sided/.style={
    draw=none,
    append after command={
      [shorten <= -0.5\pgflinewidth]
      (\tikzlastnode.north east) edge[dashed](\tikzlastnode.south east)
    }
  },
  every state/.style={circle, minimum size=1cm},
  every path/.style={thick},
  initial text=,
  node distance=1cm
}

\tikzstyle{stochasticc} = [fill, circle, minimum size=0.1cm, inner sep=0.05cm, outer sep=0cm]
\tikzstyle{stochastics} = [fill, rectangle, minimum size=0.1cm, inner sep=0.05cm, outer sep=0cm]

\newcommand{\rk}{\mathsf{rk}}

\newcommand{\init}{\mathsf{init}}

\newcommand{\IN}{\mathbb{N}}
\newcommand{\IQ}{\mathbb{Q}}

\newcommand{\proba}{\mathbb{P}}

\newcommand{\dist}[1]{\mathcal{D}(#1)}
\newcommand{\cyl}[1]{\mathsf{Cyl}(#1)}
\newcommand{\supp}[1]{\mathsf{supp}(#1)}

\newcommand{\nPlayer}{n}
\newcommand{\player}[1]{\mathcal{P}_{#1}}
\newcommand{\playerOne}{\player{1}}
\newcommand{\playerTwo}{\player{2}}
\newcommand{\playerI}{\player{i}}
\newcommand{\game}{\mathcal{G}}

\newcommand{\arenaState}{s}

\newcommand{\arenaActionSpace}{A}
\newcommand{\arenaAction}{a}

\newcommand{\concurStateSpace}{S}
\newcommand{\concurState}{s}
\newcommand{\concurActionSpace}{\bar{A}}
\newcommand{\concurActionSpacePl}[1]{A^{(#1)}}
\newcommand{\concurAction}{\bar{a}}
\newcommand{\concurActionPl}[1]{a^{(#1)}}
\newcommand{\concurActionOne}{\concurActionPl{1}}
\newcommand{\concurActionTwo}{\concurActionPl{2}}
\newcommand{\concurActionI}{\concurActionPl{i}}
\newcommand{\concurActionIAdv}{\concurActionPl{3-i}}

\newcommand{\concurActionAlt}{\bar{b}}
\newcommand{\concurActionPlAlt}[1]{b^{(#1)}}

\newcommand{\concurActionIAlt}{\concurActionPlAlt{i}}

\newcommand{\concurActionSpaceOne}{\concurActionSpacePl{1}}
\newcommand{\concurActionSpaceTwo}{\concurActionSpacePl{2}}
\newcommand{\concurActionSpaceI}{\concurActionSpacePl{i}}
\newcommand{\concurActionSpaceN}{\concurActionSpacePl{\nPlayer}}
\newcommand{\concurTrans}{\delta}
\newcommand{\concurTuple}{(\concurStateSpace, \concurActionSpaceOne, \concurActionSpaceTwo, \concurTrans)}

\newcommand{\nConcurTuple}{(\concurStateSpace, (\concurActionSpaceI)_{1\leq i \leq\nPlayer}, \concurTrans)}

\newcommand{\poGame}{\Gamma}
\newcommand{\poStateSpace}{S}
\newcommand{\poState}{s}
\newcommand{\poActionSpace}{\bar{A}}
\newcommand{\poActionSpacePl}[1]{A^{(#1)}}
\newcommand{\poAction}{\bar{a}}
\newcommand{\poActionPl}[1]{a^{(#1)}}
\newcommand{\poActionOne}{\poActionPl{1}}
\newcommand{\poActionTwo}{\poActionPl{2}}
\newcommand{\poActionI}{\poActionPl{i}}
\newcommand{\poActionIAdv}{\poActionPl{3-i}}
\newcommand{\poActionAlt}{\bar{b}}
\newcommand{\poActionPlAlt}[1]{b^{(#1)}}

\newcommand{\poActionIAlt}{\poActionPlAlt{i}}
\newcommand{\poActionIAdvAlt}{\poActionPlAlt{3-i}}
\newcommand{\poActionSpaceOne}{\poActionSpacePl{1}}
\newcommand{\poActionSpaceTwo}{\poActionSpacePl{2}}
\newcommand{\poActionSpaceI}{\poActionSpacePl{i}}

\newcommand{\poTuple}{(\poStateSpace, \poActionSpaceOne, \poActionSpaceTwo, \delta, \obsSpaceOne, \obsFunOne, \obsSpaceTwo, \obsFunTwo)}

\newcommand{\obsFun}[1]{\mathsf{Obs}_{#1}}
\newcommand{\obsSpace}[1]{\mathcal{Z}_{#1}}
\newcommand{\obsFunOne}{\obsFun{1}}
\newcommand{\obsSpaceOne}{\obsSpace{1}}
\newcommand{\obsFunTwo}{\obsFun{2}}
\newcommand{\obsSpaceTwo}{\obsSpace{2}}
\newcommand{\obsFunI}{\obsFun{i}}
\newcommand{\obsSpaceI}{\obsSpace{i}}

\newcommand{\mdpStateSpace}{S}
\newcommand{\mdpState}{s}
\newcommand{\mdpTrans}{\delta}
\newcommand{\mdpActionSpace}{A}
\newcommand{\mdpAction}{a}
\newcommand{\mdpTuple}{(\mdpStateSpace, \mdpActionSpace, \mdpTrans)}

\newcommand{\playSet}[1]{\mathsf{Plays}(#1)}
\newcommand{\play}{\pi}
\newcommand{\histSet}[1]{\mathsf{Hist}(#1)}
\newcommand{\hist}{h}
\newcommand{\prefHist}{w}

\newcommand{\paths}{\mathsf{Paths}}

\newcommand{\last}[1]{\mathsf{last}(#1)}

\newcommand{\safe}[1]{\mathsf{Safe}(#1)}
\newcommand{\reach}[1]{\mathsf{Reach}(#1)}

\newcommand{\target}{T}

\newcommand{\emptyword}{\varepsilon}

\newcommand{\mealy}{\mathcal{M}}
\newcommand{\mealyState}{m}
\newcommand{\mealyStateSpace}{M}
\newcommand{\mealyDistInit}{\mu_{\init}}
\newcommand{\mealyStateInit}{m_{\init}}
\newcommand{\mealyDist}[1]{\mu_{#1}}
\newcommand{\update}{\alpha_{\mathsf{up}}}
\newcommand{\nextmove}{\alpha_{\mathsf{nxt}}}
\newcommand{\mealyTuple}{(\mealyStateSpace, \mealyDistInit, \nextmove, \update)}
\newcommand{\mealyTupleInSt}{(\mealyStateSpace, \mealyStateInit, \nextmove, \update)}

\newcommand{\mealyAlt}{\mathcal{N}}
\newcommand{\mealyStateAlt}{n}
\newcommand{\mealyStateInitAlt}{n_\init}
\newcommand{\mealyStateSpaceAlt}{N}
\newcommand{\mealyDistInitAlt}{\nu_\init}
\newcommand{\mealyDistAlt}[1]{\nu_{#1}}
\newcommand{\updateAlt}{\beta_\mathsf{up}}
\newcommand{\nextmoveAlt}{\beta_\mathsf{nxt}}
\newcommand{\mealyTupleAlt}{(\mealyStateSpaceAlt, \mealyDistInitAlt, \nextmoveAlt, \updateAlt)}
\newcommand{\mealyTupleInStAlt}{(\mealyStateSpaceAlt, \mealyStateInitAlt, \nextmoveAlt, \updateAlt)}

\newcommand{\strat}[1]{\sigma_{#1}}
\newcommand{\stratAlt}[1]{\tau_{#1}}

\newcommand{\stratOne}{\strat{1}}
\newcommand{\stratTwo}{\strat{2}}
\newcommand{\stratI}{\strat{i}}
\newcommand{\stratAltOne}{\stratAlt{1}}
\newcommand{\stratAltTwo}{\stratAlt{2}}
\newcommand{\stratAltI}{\stratAlt{i}}

\newcommand{\stratClass}{\mathcal{C}}

\newcommand{\segmentpoint}[2]{x_{#1}^{#2}}

\newcommand{\segmentstrat}[2]{\sigma_{#1}^{#2}}

\makeatletter
\def\squarecorner#1{
    \pgf@x=\the\wd\pgfnodeparttextbox \pgfmathsetlength\pgf@xc{\pgfkeysvalueof{/pgf/inner xsep}}\advance\pgf@x by 2\pgf@xc \pgfmathsetlength\pgf@xb{\pgfkeysvalueof{/pgf/minimum width}}\ifdim\pgf@x<\pgf@xb \pgf@x=\pgf@xb \fi \pgf@y=\ht\pgfnodeparttextbox \advance\pgf@y by\dp\pgfnodeparttextbox \pgfmathsetlength\pgf@yc{\pgfkeysvalueof{/pgf/inner ysep}}\advance\pgf@y by 2\pgf@yc \pgfmathsetlength\pgf@yb{\pgfkeysvalueof{/pgf/minimum height}}\ifdim\pgf@y<\pgf@yb \pgf@y=\pgf@yb \fi \ifdim\pgf@x<\pgf@y \pgf@x=\pgf@y \else
        \pgf@y=\pgf@x \fi
\pgf@x=#1.5\pgf@x \advance\pgf@x by.5\wd\pgfnodeparttextbox \pgfmathsetlength\pgf@xa{\pgfkeysvalueof{/pgf/outer xsep}}\advance\pgf@x by#1\pgf@xa \pgf@y=#1.5\pgf@y \advance\pgf@y by-.5\dp\pgfnodeparttextbox \advance\pgf@y by.5\ht\pgfnodeparttextbox \pgfmathsetlength\pgf@ya{\pgfkeysvalueof{/pgf/outer ysep}}\advance\pgf@y by#1\pgf@ya }
\makeatother

\pgfdeclareshape{square}{
    \savedanchor\northeast{\squarecorner{}}
    \savedanchor\southwest{\squarecorner{-}}

    \foreach \x in {east,west} \foreach \y in {north,mid,base,south} {
        \inheritanchor[from=rectangle]{\y\space\x}
    }
    \foreach \x in {east,west,north,mid,base,south,center,text} {
        \inheritanchor[from=rectangle]{\x}
    }
    \inheritanchorborder[from=rectangle]
    \inheritbackgroundpath[from=rectangle]
}

\begin{document}
\maketitle
\begin{abstract}
Two-player (antagonistic) games on (possibly stochastic) graphs are a prevalent model in theoretical computer science, notably as a framework for reactive synthesis.

Optimal strategies may require randomisation when dealing with inherently probabilistic goals, balancing multiple objectives, or in contexts of partial information. There is no unique way to define randomised strategies. For instance, one can use so-called \textit{mixed} strategies or \textit{behavioural} ones. In the most general setting, these two classes do not share the same expressiveness. A seminal result in game theory --- \textit{Kuhn's theorem} --- asserts their equivalence in games of perfect recall. 

This result crucially relies on the possibility for strategies to use \textit{infinite memory}, i.e., unlimited knowledge of all past observations. However, computer systems are finite in practice. Hence it is pertinent to restrict our attention to \textit{finite-memory} strategies, defined as automata with outputs. Randomisation can be implemented in these in different ways: the \textit{initialisation}, \textit{outputs} or \textit{transitions} can be randomised or deterministic respectively. Depending on which aspects are randomised, the expressiveness of the corresponding class of finite-memory strategies differs.

In this work, we study two-player concurrent stochastic games and provide a complete taxonomy of the classes of finite-memory strategies obtained by varying which of the three aforementioned components are randomised.
Our taxonomy holds in games of perfect and imperfect information with perfect recall, and in games with more than two players.
We also provide an adapted taxonomy for games with imperfect recall.
\end{abstract}


\newcommand{\continueAct}{\mathsf{P}}
\newcommand{\checkAct}{\mathsf{C}}
\newcommand{\simpleGame}{\game_{a, b}}

\section{Introduction}
\smallskip\noindent\textbf{Games on graphs.} Games on (possibly stochastic) graphs have been studied for decades, both for their own interest (e.g.,~\cite{EM79,Con92,GZ05}) and for their value as a framework for \textit{reactive synthesis} (e.g.,~\cite{DBLP:conf/dagstuhl/2001automata,rECCS,DBLP:conf/lata/BrenguierCHPRRS16,DBLP:reference/mc/BloemCJ18}). The core problem is almost always to find \textit{optimal strategies} for the players: strategies that guarantee winning for Boolean winning conditions (e.g.,~\cite{DBLP:conf/focs/EmersonJ88,DBLP:journals/tcs/Zielonka98,DBLP:journals/corr/BruyereHR16,DBLP:conf/concur/BrihayeDOR19}), or strategies that achieve the best possible payoff in quantitative contexts (e.g.,~\cite{EM79,DBLP:journals/acta/BouyerMRLL18,DBLP:conf/concur/BruyereHRR19}). In multi-objective settings, one is interested in \textit{Pareto-optimal} strategies (e.g.,~\cite{DBLP:journals/acta/ChatterjeeRR14,DBLP:journals/iandc/VelnerC0HRR15,DBLP:journals/fmsd/RandourRS17,DBLP:conf/tacas/DelgrangeKQR20}), but the bottom line is the same: players are looking for strategies that guarantee the best possible results.

In reactive synthesis, we model the interaction between a system and its uncontrollable environment as a two-player antagonistic game, and we represent the specification to ensure as a winning objective. An optimal strategy for the system in this game then constitutes a formal blueprint for a \textit{controller} to implement in the real world~\cite{DBLP:reference/mc/BloemCJ18}.

\smallskip\noindent\textbf{Randomness in strategies.} In essence, a \textit{pure strategy} is simply a function mapping histories (i.e., the past and present of a play) to an action deterministically.

Optimal strategies may require \textit{randomisation} when dealing with inherently probabilistic goals, balancing multiple objectives, or in contexts of partial information: see, e.g.,~\cite{DBLP:conf/lics/ChatterjeeD12,DBLP:journals/fmsd/RandourRS17,DBLP:conf/icalp/BerthonRR17,DBLP:conf/tacas/DelgrangeKQR20}. There are different ways of randomising strategies. For instance, a \textit{mixed} strategy is essentially a probability distribution over a set of pure strategies. That is, the player randomly selects a pure strategy at the beginning of the game and then follows it for the entirety of the play without resorting to randomness ever again. By contrast, a \textit{behavioural} strategy randomly selects an action at each step: it thus maps histories to probability distributions over actions. 

\smallskip\noindent\textbf{Kuhn's theorem.} In full generality, these two definitions yield different classes of strategies (e.g.,~\cite{DBLP:journals/corr/abs-1006-1404} or~\cite[Chapter 11]{OR94}). Nonetheless, Kuhn's theorem~\cite{Aumann64} proves their equivalence under a mild hypothesis: in games of \textit{perfect recall}, for any mixed strategy there is an equivalent behavioural strategy and vice-versa. A game is said to be of perfect recall for a given player if said player never forgets their previous knowledge and the actions they have played (i.e., they can observe their own actions). Let us note that perfect recall and \textit{perfect information} are two different notions: perfect information is not required to have perfect recall.

Let us highlight that Kuhn's theorem crucially relies on two elements. First, mixed strategies can be distributions over an \textit{infinite} set of pure strategies. Second, strategies can use \textit{infinite memory}, i.e., they are able to remember the past completely, however long it might be. Indeed, consider a game in which a player can choose one of two actions in each round. One could define a (memoryless) behavioural strategy that selects one of the two actions by flipping a coin each round. This strategy generates infinitely many sequences of actions, therefore any equivalent mixed strategy needs the ability to randomise between infinitely many different sequences, and thus, infinitely many pure strategies. Moreover, infinitely many of these sequences require infinite memory to be generated (due to their non-regularity).

\smallskip\noindent\textbf{Finite-memory strategies.} From the point of view of reactive synthesis, infinite-memory strategies, along with randomised ones relying on infinite supports, are undesirable for implementation. This is why a plethora of recent advances has focused on \textit{finite-memory} strategies, usually represented as (a variation on) Mealy machines, i.e., finite automata with outputs. See, e.g.,~\cite{GZ05,DBLP:journals/acta/ChatterjeeRR14,DBLP:journals/iandc/BruyereFRR17,DBLP:conf/tacas/DelgrangeKQR20,DBLP:journals/lmcs/BouyerRORV22,DBLP:conf/concur/BouyerORV21}. Randomisation can be implemented in these finite-memory strategies in different ways: the \textit{initialisation}, \textit{outputs} or \textit{transitions} can be randomised or deterministic respectively.

Depending on which aspects are randomised, the expressiveness of the corresponding class of finite-memory strategies differs: in a nutshell, \textit{Kuhn's theorem crumbles when restricting ourselves to finite memory}. For instance, we show that some finite-memory strategies with only randomised outputs (i.e., the natural equivalent of behavioural strategies) cannot be emulated by finite-memory strategies with only randomised initialisation (i.e., the natural equivalent of mixed strategies) --- see Lemma~\ref{lemma:behavioural:mixed:strictness}. Similarly, it is known that some finite-memory strategies that are encoded by Mealy machines using randomisation in all three components admit no equivalent using randomisation only in outputs~\cite{DBLP:journals/tcs/AlfaroHK07,DBLP:journals/corr/abs-1006-1404}. 

\begin{figure*}[tbh]
  \begin{center}
    \begin{tikzpicture}
      \matrix (a) [matrix of nodes, align=center, text width=4cm, column sep=0.1cm, row sep=0.7cm]{
        & {DRR = RRR = RDR \\ (Thm.~\ref{theorem:RRR:DRR},~\ref{theorem:RRR:RDR})}  & \\
        DDR &  & RRD \\
        &  & DRD (behavioural) \\
        &  & RDD (mixed) \\
        & DDD (pure) & \\
      };
      \draw (a-1-2) -- (a-2-1) node[midway, above left] {Lem.~\ref{lemma:rrd_notin_ddr} (strictness)};
      \draw (a-1-2) -- (a-2-3) node[midway, above right] {Lem.~\ref{lemma:ddr_notin_rrd} (strictness)};
      \draw (a-2-3) -- (a-3-3) node[midway, right] {Lem.~\ref{lemma:rrd:behavioural:strictness} (strictness)};
      \draw (a-3-3) -- (a-4-3) node[midway, right] {Thm.~\ref{theorem:mixed:behavioural}, Lem.~\ref{lemma:behavioural:mixed:strictness}};
      \draw (a-2-1) -- (a-5-2) node[midway, below left] {Direct};
      \draw (a-4-3) -- (a-5-2) node[midway, below right] {Lem.~\ref{lemma:rrd_notin_ddr} (strictness)};
      \node at (-6, 3) (top) {};
      \node at (-6, -3) (bot) {};
      \draw[->] (bot) -- (top) node[rotate=90,midway,above] {Expressiveness};
    \end{tikzpicture}
    \caption{Lattice of strategy classes in terms of expressible probability distributions over plays against all strategies of the other player. In the three-letter acronyms, the letters, in order, refer to the initialisation, outputs and updates of the Mealy machines: D and R respectively denote deterministic and randomised components. Each line in the figure indicates that the class above is strictly more expressive than the class below.}\label{figure:lattice}
\end{center}
\end{figure*}

\smallskip\noindent\textbf{Our contributions.} We consider \textit{two-player zero-sum concurrent stochastic games of perfect information} (e.g.,~\cite{Sha53,MaitraSudderth}), encompassing two-player turn-based (deterministic) games and Markov decision processes as particular subcases. We establish a \textit{Kuhn-like taxonomy} of the classes of finite-memory strategies obtained by varying which of the three aforementioned Mealy machine components are randomised: we illustrate it in Figure~\ref{figure:lattice}, and describe it fully in Section~\ref{section:relationships}.

Let us highlight a few elements. Naturally, the least expressive model corresponds to pure strategies. In contrast to what happens with infinite memory, and as noted in the previous paragraph, we see that mixed strategies are strictly less expressive than behavioural ones. We also observe that allowing randomness both in initialisation and in outputs (RRD strategies) yields an even more expressive class --- and incomparable to what is obtained by allowing randomness in updates only. Finally, the most expressive class is obviously obtained when allowing randomness in all components; yet it may be dropped in initialisation or in outputs without reducing the expressiveness --- but not in both simultaneously.

To compare the expressiveness of strategy classes, we consider \textit{outcome-equivalence}, as defined in Section~\ref{section:preliminaries}. Intuitively, two strategies are outcome-equivalent if, against any strategy of the opponent, they yield identical probability distributions (i.e., they induce identical Markov chains). Hence we are agnostic with regard to the objective, winning condition, payoff function, or preference relation of the game, and with regard to how they are defined (e.g., colours on actions, states, transitions, etc).

Finally, let us note that in our setting of two-player concurrent stochastic games, the perfect recall hypothesis holds. Most importantly, we assume that actions are visible. Lifting this hypothesis drastically changes the relationships between the different models. While our main presentation considers two-player perfect-information games for the sake of simplicity, we show in Section~\ref{section:multiplayer} that \textit{our results hold in games with more than two players} and, in Section~\ref{section:imperfect:information}, that \textit{our results hold in games of imperfect information} too, assuming visible actions.
We provide an adapted taxonomy for \textit{games in which actions are not visible} in Section~\ref{section:imperfect:information}.

\smallskip\noindent\textbf{Related work.} We discuss several axes of research related to our work.

The first one deals with the \textit{various types of randomness} one can inject in strategies and their consequences. Obviously, Kuhn's theorem \cite{Aumann64} is a major inspiration, as well as the examples of differences between strategy models presented in~\cite{DBLP:journals/corr/abs-1006-1404}. On a different but related note,~\cite{DBLP:journals/iandc/Chatterjee0GH15} studies when randomness is not helpful in games nor strategies (as it can be simulated by other means or does not intervene).

A second direction focuses on trying to characterise the \textit{power of finite-memory strategies}, with or without randomness. One can notably cite~\cite{GZ05} for memoryless strategies, and~\cite{DBLP:conf/fsttcs/0001PR18,DBLP:journals/lmcs/BouyerRORV22},~\cite{DBLP:conf/concur/BouyerORV21}, and~\cite{DBLP:conf/stacs/BouyerRV22} for finite-memory ones in deterministic, stochastic, and infinite-arena games respectively.

The power of strategies also depends on the information they are allowed to register to update their memory: colours, as in the  papers of the previous paragraph, or the sequence of states~\cite{KopThesis,DBLP:journals/lmcs/0001HPW19,DBLP:conf/icalp/CasaresO23}, observations~\cite{DBLP:journals/jacm/BertrandGG17} or sequences of actions or labels~\cite{DBLP:journals/lmcs/0001HPW19}.

The last axis concentrates on the use of \textit{randomness as a means to simplify strategies} and/or reduce their memory requirements. Examples of this endeavour can be found in~\cite{DBLP:conf/qest/ChatterjeeAH04,DBLP:conf/hybrid/ChatterjeeHP08,DBLP:conf/stacs/Horn09,DBLP:journals/acta/ChatterjeeRR14,DBLP:conf/concur/MonmegePR20}. These are further motivations to understand randomised strategies even in contexts where randomness is not needed a priori to play optimally.

\smallskip\noindent\textbf{Outline.} Section~\ref{section:preliminaries} summarises all preliminary notions. In Section~\ref{section:relationships}, we present the taxonomy illustrated in Figure~\ref{figure:lattice} and comment on it. We divide its proofs into two sections: Section~\ref{section:inclusions} establishes the inclusions, and Section~\ref{section:strictness} establishes the separation of distinct strategy classes. Finally, Sections~\ref{section:multiplayer} and~\ref{section:imperfect:information} present how we transfer our results to the richer settings of multi-player games and of games of imperfect information respectively.
We conclude in Section~\ref{section:conclusion}.
Appendix~\ref{appendix:probability_memory} is a technical appendix dedicated to the details of an equation introduced in Section~\ref{section:preliminaries}.

A preliminary version of this work has been previously published as a conference paper~\cite{DBLP:conf/concur/MainR22}.
This version presents in detail the contributions of the conference paper with full proofs and extends the results of the conference paper by considering a broader class of games; only turn-based games are considered in~\cite{DBLP:conf/concur/MainR22}, whereas we consider \textit{concurrent games} here.
The separation of strategy classes presented in Section~\ref{section:strictness} has been enriched with examples derived from specifications, to complement the examples provided on a one-player games with one state and two actions (this is arguably the simplest possible setting in which we can consider non-trivial strategies).
Finally, the generalisation to games of imperfect information presented in Section~\ref{section:imperfect:information} has been extended to consider the case of games with imperfect recall.

\section{Preliminaries}\label{section:preliminaries}
\smallskip\noindent\textbf{Set-theoretic notation.}
We let $\IN$ and $\IQ$ denote the sets of natural and rational numbers respectively.
Given sets $A$ and $B'\subseteq B$, and a function $f\colon A\to B$, we let $f^{-1}(B') = \{a\in A\mid f(a)\in B'\}$ denote the inverse image of $B'$ by $f$.
For the inverse image of singleton sets, we write $f^{-1}(b)$ instead of $f^{-1}(\{b\})$ for any $b\in B$.

\smallskip\noindent\textbf{Probability.} Given any countable set $A$, we write $\dist{A}$ for the set of probability distributions over $A$, i.e., the set of functions $\mu\colon A\to [0, 1]$ such that $\sum_{a\in A}\mu(a)=1$.
Given such a probability distribution $\mu$, we let $\supp{\mu}=\{a\in A\mid \mu(a)> 0\}$ be the support of $\mu$.

Given a set $A$ and a $\sigma$-algebra $\mathcal{F}$ over $A$,
we denote by $\dist{A, \mathcal{F}}$ the set of probability distributions over
the measurable space $(A, \mathcal{F})$.

\smallskip\noindent\textbf{Games.} We consider two-player concurrent stochastic games of perfect information played on graphs.
We denote the two players by $\playerOne$ and $\playerTwo$.
At the start of a play, a pebble is placed on some initial state.
In each round, both players simultaneously select an action available in said state and the next state is chosen randomly following a distribution depending on the current state and the actions chosen by the players.
The game proceeds for an infinite number of rounds, yielding an infinite play.

Formally, a two-player \textit{concurrent stochastic game of perfect information}, or simply a \textit{game}, is a tuple $\game= \concurTuple$ where $\concurStateSpace$ is a non-empty finite set of states, $\concurActionSpaceOne$ and $\concurActionSpaceTwo$ are finite sets of actions for each player and $\concurTrans\colon\concurStateSpace\times\concurActionSpaceOne\times\concurActionSpaceTwo\to \dist{\concurStateSpace}$ is a (partial) probabilistic transition function.
We write $\concurActionSpace = \concurActionSpaceOne\times \concurActionSpaceTwo$ in the following.
Elements of $\concurActionSpace$ are denoted with a bar to emphasise that they are pairs of actions.
Given $\concurAction\in\concurActionSpace$, we adopt the convention that $\concurAction$ is given by the pair $(\concurActionOne, \concurActionTwo)$.

For any state $\concurState\in\concurStateSpace$, we let $\concurActionSpace(\concurState) = \{\concurAction\in\concurActionSpace\mid \concurTrans(\arenaState, \concurAction) \text{ is defined}\}$ and require that $\concurActionSpace(\concurState)$ is of the form $\concurActionSpaceOne(\concurState)\times \concurActionSpaceTwo(\concurState)$ for some subsets $\concurActionSpaceI(\concurState)$ of $\concurActionSpaceI$, i.e., the actions available to a player in a state are not constrained by the choices of the other.
We assume that for all $\concurState\in \concurStateSpace$, $\concurActionSpace(\concurState)$ is non-empty, i.e., there are no deadlocks in the game.

A \textit{play} of $\game$ is an infinite sequence $\concurState_0\concurAction_0\concurState_1\ldots\in (\concurStateSpace\concurActionSpace)^\omega$ such that for all $k\in\IN$, $\concurTrans(\concurState_k, \concurAction_k)(\concurState_{k+1})>0$.
A \textit{history} is a finite prefix of a play ending in a state.
Given a play $\play = \concurState_0\concurAction_0\concurState_1\concurAction_1\ldots$ and $k\in\IN$, we write $\play_{|k}$ for the history $\concurState_0\concurAction_0\ldots \concurAction_{k-1}\concurState_k$.
For any history $\hist =\concurState_0\concurAction_0\ldots \concurAction_{k-1}\concurState_k$, we let $\last{\hist} = \concurState_k$.
We write $\playSet{\game}$ to denote the set of plays of $\game$, $\histSet{\game}$ to denote the set of histories of $\game$.
Given some initial state $\concurState_\init\in\concurStateSpace$, we write $\histSet{\game, \concurState_\init}$ for the set of histories starting in state $\concurState_\init$.

There exist several classes of games that have been studied in their own right.
A game is \textit{turn-based} if at each round, only one player can influence the next transition.
In other words, $\game = \concurTuple$ is turn-based if for all states $\concurState\in\concurStateSpace$, there exists $i\in\{1, 2\}$ such that $|\concurActionSpaceI(\concurState)|=1$ (in which case $\player{3-i}$ controls $\concurState$).
Turn-based games are traditionally described via a partition of the state space into states controlled by $\playerOne$ and states controlled by $\playerTwo$.
A game is \textit{deterministic} if its transitions are not subject to randomness; a game $\game= \concurTuple$ is a deterministic game if for all $s\in S$ and $\concurAction\in \concurActionSpace(\concurState)$, $\concurTrans(\concurState, \concurAction)$ is a Dirac distribution.

An interesting subclass of turn-based games is that of one-player games.
A game is a \textit{one-player game} if only one player controls all transitions.
A game $\game = \concurTuple$ is a one-player game if there exists $i\in\{1, 2\}$ such that for all $\concurState\in\concurStateSpace$, $|\concurActionSpaceI(\concurState)| = 1$.
A one-player game in the sense above is the equivalent of a \textit{Markov decision process} (MDP) in our context, and will be referred to as such.
When dealing with MDPs, we lighten notation and drop information related to the inactive player.
We view MDPs as tuples $\mdpTuple$ where $\mdpStateSpace$ is a finite set of states, $\mdpActionSpace$ is a finite set of actions and $\mdpTrans\colon\mdpStateSpace\times\mdpActionSpace\to\dist{\mdpStateSpace}$ is the transition function.
Notions defined for two-player concurrent games can be directly adapted to MDPs, e.g., a play is a sequence in $(\mdpStateSpace\mdpActionSpace)^\omega$ instead of a sequence in $(\concurStateSpace\concurActionSpace)^\omega$.

We fix a game $\game = \concurTuple$ for the remainder of the section.

\smallskip\noindent\textbf{Strategies and outcomes.} A strategy is a function that describes how a player should act based on a history.
Players need not act in a deterministic fashion: they can use randomisation  to select an action.
Formally, a (behavioural) \textit{strategy} of $\playerI$ is a function $\stratI\colon \histSet{\game}\to\dist{\concurActionSpaceI}$ such that for all histories $\hist\in\histSet{\game}$, $\supp{\stratI(\hist)}\subseteq\concurActionSpaceI(\last{\hist})$.
In other words, a strategy assigns to any history a distribution over the actions available to $\playerI$ in this state.

When both players fix a strategy and an initial state is decided, we obtain a purely stochastic process, i.e., a Markov chain.
Let us recall the relevant $\sigma$-algebra for the definition of probabilities over plays.
For any history $\hist\in\histSet{\game}$, we define
$\cyl{\hist} = \{\play\in\playSet{\game}\mid \hist\text{ is a prefix of } \play\}$,
the \textit{cylinder} of $\hist$, consisting of plays that extend $\hist$.
Let us denote by $\mathcal{F}_{\game}$ the $\sigma$-algebra generated by all cylinder sets.

Let $\stratOne$ and $\stratTwo$ be strategies of $\playerOne$ and $\playerTwo$ respectively and $\concurState_\init\in\concurStateSpace$ be an initial state.
We define the probability measure (over $(\playSet{\game}, \mathcal{F}_\game)$)
induced by playing $\stratOne$ and $\stratTwo$ from $\concurState_\init$ in $\game$, written $\proba^{\stratOne, \stratTwo}_{\concurState_\init}$, as follows.
For any history $\hist=\concurState_0\concurAction_0\ldots \concurState_n\in\histSet{\game, \concurState_\init}$, the probability $\proba^{\stratOne, \stratTwo}_{\concurState_\init}(\cyl{h})$ assigned to $\cyl{\hist}$ is given by the product
\[
  \prod_{k=0}^{n-1}
  \stratOne(\concurState_0\concurAction_0\ldots \concurState_k)(\concurActionOne_k)\cdot
  \stratTwo(\concurState_0\concurAction_0\ldots \concurState_k)(\concurActionTwo_k)\cdot
  \concurTrans(\concurState_k, \concurAction_k)(\concurState_{k+1}).\]
For any history $\hist\in\histSet{\game}\setminus \histSet{\game, \concurState_\init}$, we set $\proba^{\stratOne, \stratTwo}_{\concurState_\init}(\cyl{\hist}) = 0$. By Carath\'{e}odory's extension theorem~\cite[Theorem A.1.3]{Dur19}, the measure described above can be extended in a unique fashion to $(\playSet{\game}, \mathcal{F}_\game)$.
For MDPs, we drop the strategy of the absent player in the notation of this distribution and write $\proba^{\stratOne}_{\concurState_\init}$.

Let $\stratI$ be a strategy of $\playerI$.
A play or play prefix $\concurState_0\concurAction_0\concurState_1\ldots$ is said to be \textit{consistent} with $\stratI$ if for all action indices $k$, it holds that $\stratI(\concurState_0\concurAction_0\ldots \concurState_k)(\concurActionI_k)> 0$.\footnote{We use the terminology of consistency not only for plays and histories, but also for prefixes of plays that end with an action pair.}

\smallskip\noindent\textbf{Outcome-equivalence of strategies.} In later sections, we study the expressiveness of finite-memory strategy models depending on the type of randomisation allowed.
Two strategies may yield the same outcomes despite being different: the actions suggested by a strategy in an inconsistent history can be changed without affecting which probability distributions are induced by the strategy.
Therefore, instead of using the equality of strategies as a measure of equivalence, we consider some weaker notion of equivalence, referred to as outcome-equivalence.

We say that two strategies $\stratOne$ and $\stratAltOne$ of $\playerOne$ are \textit{outcome-equivalent} if for any strategy $\stratTwo$ of $\playerTwo$ and for any initial state $\concurState_\init$, the probability distributions $\proba^{\stratOne, \stratTwo}_{\concurState_\init}$ and $\proba^{\stratAltOne, \stratTwo}_{\concurState_\init}$ coincide.

We now provide the criterion used in our proofs to establish the outcome-equivalence of strategies.
This criterion does not invoke the probability distributions induced by strategies directly.
The idea is that when comparing two strategies of a player, we need only be concerned with the suggestions these strategies provide in histories that are consistent with them.
In other words, any deviation in unreachable histories does not affect the outcome. Hence, one could reformulate outcome-equivalence as having to suggest the same distributions over actions in histories that are consistent with (one of) the strategies. In the sequel, we prove that this reformulation is indeed equivalent to the definition of outcome-equivalence. We rely on this reformulation to prove the outcome-equivalence of two strategies.

\begin{lemma}[Strategic criterion for outcome-equivalence]\label{lemma:prelim:consistent_equivalence}
  Let $\stratI$ and $\stratAltI$ be two strategies of $\playerI$.
  These two strategies are outcome-equivalent if and only if for all histories $\hist\in\histSet{\game}$, $\hist$ consistent with $\stratI$ implies $\stratI(\hist) = \stratAltI(\hist)$.
\end{lemma}

\begin{proof}
  To aid with notation, we assume that $i = 1$; the proof of the other case is done by exchanging the players below.
  First, we assume that $\stratOne$ and $\stratAltOne$ are outcome-equivalent.
  Let $\hist\in\histSet{\game}$ be a history that is consistent with $\stratOne$.
  Let $\concurState_\init$ denote the first state of $\hist$ and let $\stratTwo$ be a strategy of $\playerTwo$ consistent with $\hist$.
  Let $\concurActionOne\in\concurActionSpaceOne(\concurState)$ and $\concurActionTwo\in\supp{\stratTwo(\hist)}$, and write $\concurAction = (\concurActionOne, \concurActionTwo)$.
  Let $\concurState\in\supp{\delta(\last{\hist}), \concurAction)}$.
  By definition of the probability of a cylinder set and consistency of $\hist$ with both $\stratOne$ and $\stratTwo$, we have
  \[\stratOne(\hist)(\concurActionOne) =
    \frac{\proba_{\concurState_\init}^{\stratOne, \stratTwo}(\cyl{\hist\concurAction\concurState})}{\proba_{\concurState_\init}^{\stratOne, \stratTwo}(\cyl{\hist})\cdot \stratTwo(\hist)(\concurActionTwo)\cdot\concurTrans(\last{\hist}, \concurAction)(\concurState)}.\]
  Furthermore, the outcome-equivalence of $\stratOne$ and $\stratAltOne$ implies that $\proba_{\concurState_\init}^{\stratAltOne, \stratTwo}(\cyl{\hist}) = \proba_{\concurState_\init}^{\stratOne, \stratTwo}(\cyl{\hist})>0$.
  Therefore, we have
    \[\stratAltOne(\hist)(\concurActionOne) =
    \frac{\proba_{\concurState_\init}^{\stratAltOne, \stratTwo}(\cyl{\hist\concurAction\concurState})}{\proba_{\concurState_\init}^{\stratAltOne, \stratTwo}(\cyl{\hist})\cdot \stratTwo(\hist)(\concurActionTwo) \cdot\concurTrans(\last{\hist}, \concurAction)(\concurState)}.\]
  It follows from the equations above and the outcome-equivalence of $\stratOne$ and $\stratAltOne$ that $\stratOne(\hist)(\concurActionOne) = \stratAltOne(\hist)(\concurActionOne)$.
  We have shown that $\stratOne(\hist) = \stratAltOne(\hist)$, which ends the proof of the first direction.
  
  Let us now assume that $\stratOne$ and $\stratAltOne$ coincide over histories consistent with $\stratOne$.
  Let $\stratTwo$ be a strategy of $\playerTwo$ and $\concurState_\init\in S$ be an initial state.
  It suffices to study the probability of cylinder sets.
  Let $\hist\in\histSet{\game}$ be a history starting in $\concurState_\init$.
  If $\hist$ is consistent with $\stratOne$, then all prefixes of $\hist$ also are, therefore the definition of the probability of a cylinder ensures that $\proba_{\concurState_\init}^{\stratOne, \stratTwo}(\cyl{\hist}) = \proba_{\concurState_\init}^{\stratAltOne, \stratTwo}(\cyl{\hist})$.
  Otherwise, if $\hist$ is not consistent with $\stratOne$, then $\hist$ is necessarily of the form $\hist'\concurAction\hist''$ with $\hist'$ consistent with $\stratOne$ and $\stratOne(\hist')(\concurActionOne) = 0$.
  It follows that $\stratAltOne(h')(\concurActionOne) = 0$, thus $\proba_{\concurState_\init}^{\stratOne, \stratTwo}(h) = \proba_{\concurState_\init}^{\stratAltOne, \stratTwo}(h) = 0$.
  This shows that $\stratOne$ and $\stratAltOne$ are outcome-equivalent, ending the proof.
\end{proof}

\smallskip\noindent\textbf{Subclasses of  strategies.}
A strategy is called \textit{pure} if it does not use randomisation; a pure strategy of $\playerI$ can be viewed as a function $\histSet{\game}\to \concurActionSpaceI$.
A strategy that only uses information on the current state of the play is called \textit{memoryless}: a strategy $\stratI$ of $\playerI$ is memoryless if for all histories $\hist, \hist'\in\histSet{\game}$, $\last{\hist}=\last{\hist'}$ implies $\stratI(\hist) =\stratI(\hist')$.
Memoryless strategies can be viewed as functions $\concurStateSpace\to \dist{\concurActionSpaceI}$.
Strategies that are both memoryless and pure can be viewed as functions $\concurStateSpace\to\concurActionSpaceI$.

A strategy is said to be \textit{finite-memory} (FM) if it can be encoded by a Mealy machine, i.e., an automaton with outputs along its edges.
We can include randomisation in the initialisation, outputs and updates (i.e., transitions) of the Mealy machine.
Formally, a \textit{stochastic Mealy machine} of $\playerI$ is a tuple $\mealy = \mealyTuple$, where $\mealyStateSpace$ is a finite set of memory states, $\mealyDistInit\in\dist{M}$ is an initial distribution, $\nextmove\colon \mealyStateSpace\times \concurStateSpace\to \dist{\concurActionSpaceI}$ is the (stochastic) next-move function and $\update\colon \mealyStateSpace\times\concurStateSpace\times\concurActionSpaceI\to \dist{\mealyStateSpace}$ is the (stochastic) update function.

Before we explain how to define the strategy induced by a Mealy machine, let us first describe how these machines work.
Fix a Mealy machine $\mealy = \mealyTuple$.
Let $\concurState_0\in \concurStateSpace$.
At the start of a play, an initial memory state $\mealyState_0$ is selected randomly following $\mealyDistInit$.
Then, at each step of the play, the action $\concurActionI_k$ of $\playerI$ is chosen following the distribution $\nextmove(\mealyState_k, \concurState_k)$, whereas the action $\concurActionIAdv$ is independently chosen according to the strategy of $\player{3-i}$.
The memory state $\mealyState_{k+1}$  is then randomly updated following the distribution $\update(\mealyState_k, \concurState_k, \concurAction_k)$ and the game state $\concurState_{k+1}$ is chosen following the distribution $\concurTrans(\concurState_k, \concurAction_k)$, both choices being made independently.

Let us now explain how a strategy can be derived from a Mealy
machine. As explained previously, when in a certain memory state $\mealyState\in\mealyStateSpace$ and game state $\concurState\in\concurStateSpace$, the probability of an action $\concurActionI\in \concurActionSpaceI(\concurState)$ being chosen is given by $\nextmove(\mealyState,\concurState)(\concurActionI)$.
Therefore, the probability of choosing the action $\concurActionI\in\concurActionSpaceI$ after some history $\hist = \prefHist\concurState$ (where $\prefHist\in (\concurStateSpace\concurActionSpace)^*$ and $\concurState=\last{\hist}$) is given by the sum, for each memory state $\mealyState\in\mealyStateSpace$, of the probability that $\mealyState$ was reached after $\prefHist$ has taken place (i.e., after $\mealy$ processes $\prefHist$), multiplied by $\nextmove(\mealyState, \concurState)(\concurActionI)$.

To provide a formal definition of the strategy induced by $\mealy$, we first describe the distribution over memory states of $\mealy$ after elements of $(\concurStateSpace\concurActionSpace)^*$ take place (under the strategy induced by $\mealy$).
We formally define this distribution inductively.
Details for the derivation of the inductive formula, which rely on conditional probabilities, are deferred to Appendix~\ref{appendix:probability_memory}.

The distribution $\mealyDist{\emptyword}$ over memory states after the empty word $\emptyword$ (i.e., nothing) has taken place is by definition $\mealyDistInit$.
Assume inductively that we know the distribution $\mealyDist{\prefHist}$ for $\prefHist= \concurState_0\concurAction_0\ldots \concurState_{k-1}\concurAction_{k-1}$.
We explain how to derive $\mealyDist{\prefHist\concurState_k\concurAction_k}$ from $\mealyDist{\prefHist}$ for any state $\concurState_k\in\supp{\concurTrans(\concurState_{k-1}, \concurAction_{k-1})}$ and for any pair of actions $\concurAction_k\in \concurActionSpace(\concurState_k)$.

In general, the choice of an action by $\playerI$ conditions what the predecessor memory states could be.
First, we note that if $\nextmove(\mealyState', \concurState_k)(\concurActionI_k) = 0$ holds for all memory states $\mealyState'\in \supp{\mealyDist{\prefHist}}$, then the action $\concurActionI_k$ is actually never chosen.
We leave this case undefined (the related conditional probabilities are ill-defined) and assume $\concurActionI_k\in\supp{\nextmove(\mealyState', \concurState_k)}$ for some $\mealyState'\in\supp{\mealyDist{\prefHist}}$.
The equation for $\mealyDist{\prefHist\concurState_k\concurAction_k}$ uses the likelihood of being in a memory state knowing that the action $\concurActionI_k$ was chosen, and not $\mealyDist{\prefHist}$ directly.
We have, for any memory state $\mealyState\in\mealyStateSpace$,
\[\mealyDist{w\concurState_k\concurAction_k}(\mealyState) = \frac{\sum_{\mealyState'\in\mealyStateSpace} \mealyDist{\prefHist}(\mealyState')\cdot \update(\mealyState', \concurState_k, \concurAction_k)(\mealyState)\cdot \nextmove(\mealyState', \concurState_k)(\concurActionI_k)}{\sum_{\mealyState'\in\mealyStateSpace} \mealyDist{\prefHist}(\mealyState')\cdot \nextmove(\mealyState', \concurState_k)(\concurActionI_k)}.\]
We remark that this quotient is not well-defined whenever for all $\mealyState'\in\supp{\mealyDist{\prefHist}}$, $\nextmove(\mealyState', \concurState_k)(\concurActionI_k) = 0$, further justifying the distinction above.

Using these distributions, we formally define the (partial) strategy $\stratI^{\mealy}$ induced by the Mealy machine $\mealy = \mealyTuple$ as the strategy $\stratI^{\mealy}\colon \histSet{\game}\to\dist{\concurActionSpaceI}$ such that for all histories $\hist = \prefHist\concurState$, for all actions $\concurActionI\in \concurActionSpaceI(\concurState)$,
\[\stratI^{\mealy}(\hist)(\concurActionI) = \sum_{\mealyState\in\mealyStateSpace}\mealyDist{\prefHist}(\mealyState)\cdot \nextmove(\mealyState,\concurState)(\concurActionI).\]
This strategy is only partially defined because distributions $\mealyDist{\prefHist}$ are not defined for all $\prefHist\in(\concurStateSpace\concurActionSpace)^*$.
All histories for which $\stratI^{\mealy}$ is undefined can be shown to be of the form $\hist\concurAction\hist'$ such that $\stratI^{\mealy}$ is defined for $\hist$ and $\stratI^{\mealy}(\hist)(\concurActionI) = 0$.
Therefore, no matter how the partial definition of $\stratI^{\mealy}$ given above is extended, it does not influence the induced probability distribution over plays involving this strategy.

\smallskip\noindent\textbf{Classifying finite-memory strategies.} In the sequel, we investigate the
relationships between different classes of finite-memory strategies with respect to
expressive power. We classify finite-memory strategies following the type of
stochastic Mealy machines that can induce them.
We introduce a concise notation for each class: we use three-letter acronyms of the form XXX with X $\in \{$D,R$\}$, where the letters, in order, refer to the initialisation, outputs and updates of the Mealy machines, with D and R respectively denoting deterministic and randomised components.
For instance, we will write RRD to denote the class of Mealy machines
that have randomised initialisation and outputs, but deterministic updates.
We also apply this terminology to FM strategies: we will say that an FM
strategy is in the class XXX --- i.e., it is an XXX strategy --- if it is induced by an XXX Mealy machine.

Moreover, in the remainder of the paper, we will abusively identify Mealy machines and
their induced FM strategies. For instance, we will say
that $\mathcal{M}$ is an XXX strategy to mean that $\mathcal{M}$ is an XXX
Mealy machine (thus inducing an XXX strategy). 
As a by-product of this identification, we apply the terminology introduced previously
for strategies to Mealy machines, without explicitly referring to the
strategy they induce.
For instance, we may say a history is consistent with some Mealy machine,
or that two Mealy machines are outcome-equivalent.
Let us note however that we will not use a Mealy machine in lieu of its induced
strategy whenever we are interested in the strategy itself as a function.
This choice lightens notations; the strategy induced by a Mealy machine need not
be introduced unless it is required as a function.

We close this section by commenting on some of the classes, and discuss previous appearances in the literature, under different names.
Pure strategies use no randomisation: hence, the class DDD corresponds to pure FM strategies, which can be represented by Mealy machines
that do not rely on randomisation.

Strategies in the class DRD have been referred to as \textit{behavioural} FM
strategies in~\cite{DBLP:journals/corr/abs-1006-1404}. The name comes from the randomised
outputs, reminiscent of behavioural strategies that output a distribution over
actions after a history. We note that stochastic Mealy machines that induce DRD
strategies are such that their distributions over memory states are Dirac due to the
deterministic initialisation and updates.

Similarly, RDD strategies have been referred to as \textit{mixed} FM
strategies~\cite{DBLP:journals/corr/abs-1006-1404}.
The general definition of a mixed strategy is a distribution over pure strategies:
under a mixed strategy, a player randomly selects a pure strategy at the start of a
play and plays according to it for the whole play. RDD strategies are similar in the
way that the random initialisation can be viewed as randomly selecting some DDD strategy
(i.e., a pure FM strategy) among a \textit{finite} selection of such strategies.

The elements of RRR, the broadest class of FM strategies, have been referred
to as general FM strategies~\cite{DBLP:journals/corr/abs-1006-1404} and stochastic-update
FM strategies~\cite{DBLP:journals/corr/abs-1104-3489,DBLP:journals/lmcs/ChatterjeeKK17}.
The latter name highlights the random nature of updates and insists on the
difference with models that rely on deterministic updates, more common in the literature.

\section{Taxonomy of finite-memory strategies}\label{section:relationships}
In this section, we comment on the relationships between the classes of finite-memory strategies in terms of expressiveness.
We say that a class $\stratClass_1$ of FM strategies is no less expressive than a class $\stratClass_2$ if for all games $\game$, for all FM strategies $\mealy\in\stratClass_2$ in $\game$, one can find some FM strategy $\mealy'\in\stratClass_1$ of $\game$ such that $\mealy$ and $\mealy'$ are outcome-equivalent strategies.
For the sake of brevity, we will say that $\stratClass_2$ is included in $\stratClass_1$, and write $\stratClass_2\subseteq\stratClass_1$.

Figure~\ref{figure:lattice} summarises our results.
Each line representing an inclusion is decorated with a reference to the relevant results.
The strictness
results hold in one-player deterministic games.
In particular, there are no collapses in the diagram in the turn-based setting, which subsumes two-player deterministic turn-based games and Markov decision processes.

Some inclusions follow directly from some classes having more randomisation power than others: a deterministic component can be emulated using Dirac distributions.
For instance, the inclusion $\mathrm{DRD} \subseteq \mathrm{RRD}$ follows from the fact that RRD Mealy machines have both randomised initialisation and outputs whereas DRD ones only have randomised outputs.
The inclusions $\mathrm{RDD} \subseteq \mathrm{DRD}$, $\mathrm{RRR} \subseteq \mathrm{DRR}$ and $\mathrm{RRR} \subseteq \mathrm{RDR}$, which do not follow from such arguments, are covered in Section~\ref{section:inclusions}.

Pure strategies are strictly less expressive than any other class of FM strategies; pure strategies cannot induce any non-Dirac distributions on plays in deterministic one-player games.
Other arguments for the separation of classes of strategies are provided in Section~\ref{section:strictness}.

We close this section by comparing our results with Kuhn's theorem.
Kuhn's theorem asserts that the classes of behavioural strategies and mixed strategies in games of perfect recall share the same expressiveness.
Games of perfect recall have two traits: players never forget the sequence of histories controlled by them that have taken place and they can see their own actions.
In particular, stochastic games of perfect information are a special case of games of perfect recall.
Recall that mixed strategies are distributions over pure strategies.
We comment briefly on the techniques used in the proof of Kuhn's theorem, and compare them with the finite-memory setting.
Let us fix a game $\game=\concurTuple$.

On the one hand, the emulation of mixed strategies with behavioural strategies is performed as follows.
Let $p_i$ be a mixed strategy  of $\playerI$, i.e., a distribution over pure strategies of $\game$.
An outcome-equivalent behavioural strategy $\stratI$ is constructed such that, for all histories $\hist\in\histSet{\game}$ and actions $\concurActionI\in \concurActionSpaceI(\last{\hist})$, the probability $\stratI(\hist)(\concurActionI)$ is defined as
\[\frac {p_i(\{\stratAltI \text{ pure strategy}\mid \stratAltI \text{ consistent with } \hist\text{ and }\stratAltI(\hist)=\concurActionI\})}
  {p_i(\{\stratAltI \text{ pure strategy}\mid \stratAltI \text{ consistent with } \hist\})}.\]
In the finite-memory case, similar ideas can be used to show that $\mathrm{RDD}\subseteq\mathrm{DRD}$.
In the proof of Theorem~\ref{theorem:mixed:behavioural}, from some RDD strategy (i.e., a so-called mixed FM strategy), we construct a DRD strategy (i.e., a so-called behavioural FM strategy) that keeps track of the finitely many pure FM strategies that the RDD strategy mixes and that are consistent with the current history.
An adaption of the quotient above is used in the next-move function of the DRD strategy.

On the other hand, the emulation of behavioural strategies by mixed strategies exploits the fact that mixed strategies may randomise over \textit{infinite} sets.
In a finite-memory setting, the same techniques cannot be applied.
As a consequence, the class of RDD strategies is strictly included in the class of DRD strategies.
In a certain sense, one could say that Kuhn's theorem only partially holds in the case of FM strategies.

\section{Non-trivial inclusions}\label{section:inclusions}

This section covers the non-trivial inclusions that are asserted in the lattice of Figure~\ref{figure:lattice}.
The structure of this section is as follows.
Section~\ref{subsection:mixed_included_behavioural} covers the inclusion $\mathrm{RDD}\subseteq\mathrm{DRD}$.
The inclusion $\mathrm{RRR}\subseteq\mathrm{DRR}$ is presented in Section~\ref{subsection:RRR_vs_DRR}.
Finally, we close this section by proving the inclusion $\mathrm{RRR}\subseteq\mathrm{RDR}$ in Section~\ref{subsection:RRR_vs_RDR}.

\subsection{Simulating RDD strategies with DRD ones}\label{subsection:mixed_included_behavioural}
In this section, we focus on the classes RDD and DRD.
We prove that for all RDD strategies in any game, one can find some outcome-equivalent DRD strategy (Theorem~\ref{theorem:mixed:behavioural}).
Let us note that the converse inclusion is not true, and this discussion is relegated to Section~\ref{subsection:mixed_different_behavioural}.
The construction provided in the proof of Theorem~\ref{theorem:mixed:behavioural} yields a DRD strategy that has a state space of size exponential in the size of the state space of the original RDD strategy.
We complement Theorem~\ref{theorem:mixed:behavioural} by proving that there are some RDD strategies for which this exponential blow-up in the number of states is necessary for any outcome-equivalent DRD strategy (Lemma~\ref{lemma:mixed:behavioural:bound}).
We show that this blow-up is unavoidable in both deterministic turn-based two-player games and MDPs.

Let $\game = \concurTuple$ be a game.
Fix an RDD strategy $\mealy = \mealyTuple$ of $\playerI$.
Let us sketch how to emulate $\mealy$ with a DRD strategy $\mealyAlt = \mealyTupleInStAlt$ built with a subset construction-like approach.
The memory states of $\mealyAlt$ are functions $f\colon\supp{\mealyDistInit}\to \mealyStateSpace\cup\{\bot\}$.
A memory state $f$ is interpreted as follows.
For all initial memory states $\mealyState_0\in\supp{\mealyDistInit}$, we have $f(\mealyState_0) = \bot$ if the history seen up to now is not consistent with the pure FM strategy $(\mealyStateSpace, \mealyState_0, \nextmove, \update)$, and otherwise $f(\mealyState_0)$ is the memory state reached in the same pure FM strategy after processing the current history.
Updates are naturally derived from these semantics.

Using this state space and update scheme, we can compute the likelihood of each memory state of the mixed FM strategy $\mealy$ after some sequence $\prefHist\in (\concurStateSpace\concurActionSpace)^*$ has taken place.
Indeed, we keep track of each initial memory state from which it was possible to be consistent with $\prefHist$, and, for each such initial memory state $\mealyState_0$, the memory state reached after $\prefHist$ was processed starting in $\mealyState_0$.
Therefore, this likelihood can be inferred from $\mealyDistInit$; the probability of $\mealy$ being in $\mealyState\in\mealyStateSpace$ after $\prefHist$ has been processed is given by the (normalised) sum of the probability of each initial memory state $\mealyState_0\in\supp{\mealyDistInit}$ such that $f(\mealyState_0) = \mealyState$.

The definition of the next-move function of $\mealyAlt$ is directly based on the distribution over states of $\mealy$ described in the previous paragraph, and ensures that the two strategies select actions with the same probabilities at any given state.
For any action $\concurActionI\in\concurActionSpaceI(\concurState)$, the probability of $\concurActionI$ being chosen in game state $\concurState$ and in memory state $f$ is determined by the probability of $\mealy$ being in some memory state $\mealyState$ such that $\nextmove(\mealyState,\concurState) = \concurActionI$, where this probability is inferred from $f$.

Intuitively, we postpone the initial randomisation and instead randomise at each step in an attempt of replicating the initial distribution in the long run.
In the sequel, we formalise the DRD strategy outlined above and prove its outcome-equivalence with the RDD strategy it is based on.

\begin{theorem}\label{theorem:mixed:behavioural}
  Let $\game = \concurTuple$ be a game.
  Let $\mealy = \mealyTuple$ be an RDD strategy of $\playerI$.
  There exists a DRD strategy $\mealyAlt = \mealyTupleInStAlt$ such that $\mealyAlt$ and $\mealy$ are outcome-equivalent.
\end{theorem}
\begin{proof}
  We formalise the strategy described above.
  Let us write $\mealyStateSpace_0$ for the support of the initial distribution $\mealyDistInit$ of $\mealy$.
  We define the set of memory states $\mealyStateSpaceAlt$ to be the set of functions $\mealyStateSpace_0\to\mealyStateSpace\cup\{\bot\}$.
  The initial memory state of $\mealyAlt$ is given by the identity function $\mealyStateInitAlt\colon \mealyState_0\mapsto \mealyState_0$ over $\mealyStateSpace_0$.
  The update function $\updateAlt$ is as follows.
  For any $f\in\mealyStateSpaceAlt$, any $\concurState\in\concurStateSpace$ and $\concurAction\in \concurActionSpace(\concurState)$, we let $\updateAlt(f, \concurState, \concurAction)$ be the function $f'$ such that for all $\mealyState_0\in \mealyStateSpace_0$, we have
  \[f'(\mealyState_0) = \begin{cases}
      \update(f(\mealyState_0), \concurState, \concurAction) & \text{ if } f(\mealyState_0)\in \mealyStateSpace\text{ and } \nextmove(f(\mealyState_0), \concurState) = \concurActionI \\
      \bot & \text{ otherwise.}
    \end{cases}\]
  Whenever we perform an update of the memory, we refine our knowledge on what the initial memory state could have been according to the actions selected by $\playerI$ prior to the update.
  This refinement proceeds by mapping to $\bot$ any initial memory states $\mealyState_0$ such that the played action would not have been selected in the memory state $f(\mealyState_0)\in M$, effectively removing $\mealyState_0$ from the set of initial memory states from which we could have started.

  The next-move function $\nextmoveAlt$ is defined as follows: for any memory state $f\in\mealyStateSpaceAlt$ and $\concurState\in\concurStateSpace$, we let $\nextmoveAlt(f, \concurState)$ be arbitrary if $f$ maps $\bot$ to all memory states, and otherwise $\nextmoveAlt(f, \concurState)$ is the distribution over $\concurActionSpaceI$ such that, for all $\concurActionI\in\concurActionSpaceI(\concurState)$, we have
  \[\nextmoveAlt(f, \concurState)(\concurActionI) = \sum_{\substack{\mealyState_0\in \mealyStateSpace_0 \\ \nextmove(f(\mealyState_0), \concurState) =\concurActionI}}\frac{\mealyDistInit(\mealyState_0)}{\sum_{\mealyState_0'\in f^{-1}(\mealyStateSpace)}\mealyDistInit(\mealyState_0')}.\]

  We note that the memory state $f\in\mealyStateSpaceAlt$ mapping $\bot$ to all initial memory states is only reached whenever a history inconsistent with $\mealy$ has taken place under $\mealy$.
  Thanks to Lemma~\ref{lemma:prelim:consistent_equivalence}, we need not take in account histories inconsistent with $\mealy$ to establish the outcome-equivalence of $\mealy$ and $\mealyAlt$.
  This explains why the next-move function is left arbitrary in that case.

  We now show that $\mealy$ and $\mealyAlt$ are outcome-equivalent via Lemma~\ref{lemma:prelim:consistent_equivalence}.
  To this end, we first show a relation, for each $\prefHist\in (\concurStateSpace\concurActionSpace)^*$ consistent with $\mealy$, between the distribution $\mealyDist{\prefHist}\in\dist{\mealyStateSpace}$ over the memory states of $\mealy$ after processing $\prefHist$ and the function $f_\prefHist$ reached after $\mealyAlt$ reads $\prefHist$ (recall that for a DRD strategy, the distribution over its states after processing $\prefHist$ is a Dirac distribution).
  Formally, this relation is as follows: for any $\prefHist\in (\concurStateSpace\concurActionSpace)^*$ consistent with $\mealy$ and any memory state $\mealyState\in\mealyStateSpace$, we have
  \begin{equation}\label{equation:mixed:behavioural:1}
    \mealyDist{\prefHist}(\mealyState) = \frac{\sum_{\mealyState_0\in f_\prefHist^{-1}(\mealyState)}\mealyDistInit(\mealyState_0)}{\sum_{\mealyState_0\in f_{\prefHist}^{-1}(\mealyStateSpace)}\mealyDistInit(\mealyState_0)}.
  \end{equation}
In the above, $f_{\prefHist}^{-1}(\mealyStateSpace)$ is the set of potential initial $\mealyState_0\in \mealyStateSpace_0$ of $\mealy$ that are compatible with $\prefHist$ taking place.
This equation intuitively expresses that $\mealyAlt$ accurately keeps track of the current distribution over memory states of $\mealy$ along a play. A corollary of the above is that whenever we follow histories consistent with $\mealy$, we are assured to never reach the memory state of $\mealyAlt$ that assigns $\bot$ to all states in $\mealyStateSpace_0$.

  We prove Equation~\eqref{equation:mixed:behavioural:1} with an inductive argument.
  The case of $\prefHist=\varepsilon$ is trivial: by definition $\mealyDist{\varepsilon} = \mealyDistInit$ and $f_\varepsilon$ is the identity function over $\mealyStateSpace_0$.
  Now, let us assume that Equation~\eqref{equation:mixed:behavioural:1} holds for $\prefHist'\in (\concurStateSpace\concurActionSpace)^*$ consistent with $\mealy$, and let us prove it for $\prefHist = \prefHist'\concurState\concurAction$ consistent with $\mealy$.

  When writing relations between $\mealyDist{w'}$ and $\mealyDist{w}$ in the remainder of the proof, we adopt notation slightly different to Section~\ref{section:preliminaries}.
  In this case, the update function $\update$ and next-move $\nextmove$ of $\mealy$ are deterministic.
  Thus, instead considering sums weighted by Dirac distributions, we only sum over relevant states for clarity.

  First, we remark that it may be the case that $f_{\prefHist}^{-1}(\mealyStateSpace) \neq f_{\prefHist'}^{-1}(\mealyStateSpace)$.
  In light of this, we must take care not to have $f_{\prefHist}^{-1}(\mealyStateSpace)=\emptyset$, in which case the denominator of the right-hand side of Equation~\eqref{equation:mixed:behavioural:1} evaluates to zero.
  From the definition of $\updateAlt$, it follows that $f_{\prefHist}^{-1}(\mealyStateSpace)$ is formed of the memory elements $\mealyState_0\in f_{\prefHist'}^{-1}(\mealyStateSpace)$ such that $\nextmove(f_{\prefHist'}(\mealyState_0), \concurState) =\concurActionI$.
  We know that $\prefHist = \prefHist'\concurState\concurAction$ is consistent with $\mealy$.
  This implies there is some $\mealyState\in\mealyStateSpace$ such that $\nextmove(\mealyState,\concurState) = \concurActionI$ and $\mealyDist{\prefHist'}(\mealyState)> 0$.
  From the inductive hypothesis (Equation~\eqref{equation:mixed:behavioural:1} with $\prefHist'$), we obtain that there is some $\mealyState_0\in f_{\prefHist'}^{-1}(\mealyStateSpace)$ such that $f_{\prefHist'}(\mealyState_0) = \mealyState$, otherwise the right-hand side of the equation would evaluate to zero.
  The equality $f_{\prefHist'}(\mealyState_0) = \mealyState$ implies $\mealyState_0\in f_{\prefHist}^{-1}(\mealyStateSpace)$, thus we have shown that $\mealyStateSpace_0(\mealyState)$ is non-empty.

  Now that we have shown that Equation~\eqref{equation:mixed:behavioural:1} is well-defined for $\prefHist$, we move on to its proof.
  Let us write $\nextmove(\cdot, \concurState)^{-1}(\concurActionI)$ for the set $\{\mealyState\in\mealyStateSpace\mid \nextmove(\mealyState, \concurState) = \concurActionI\}$.
  By definition, we have
  \[\mealyDist{\prefHist}(\mealyState) =
    \frac{\sum_{\substack{m'\in \nextmove(\cdot, \concurState)^{-1}(\concurActionI)\\ \update(\mealyState', \concurState, \concurAction) =\mealyState}}\mealyDist{\prefHist'}(\mealyState')}{\sum_{\mealyState'\in \nextmove(\cdot, \concurState)^{-1}(\concurActionI)}\mealyDist{\prefHist'}(\mealyState')}.\]

  For the numerator, we obtain from the inductive hypothesis that
  \begin{align*}
    \sum_{\substack{m'\in \nextmove(\cdot,\concurState)^{-1}(\concurActionI)\\ \update(m',\concurState, \concurAction) = m}}\mealyDist{\prefHist'}(m') &=  \sum_{\substack{m'\in \nextmove(\cdot,\concurState)^{-1}(\concurActionI)\\ \update(m',\concurState, \concurAction) = m}} \sum_{\mealyState_0 \in f_{\prefHist'}^{-1}(m')}\frac{\mealyDistInit(\mealyState_0)}{\sum_{\mealyState_0'\in f_{\prefHist'}^{-1}(\mealyStateSpace)}\mealyDistInit(\mealyState_0')} \\
   &= \sum_{\mealyState_0 \in f_{\prefHist}^{-1}(\mealyState)}\frac{\mealyDistInit(\mealyState_0)}{\sum_{\mealyState_0'\in f_{\prefHist'}^{-1}(\mealyStateSpace)}\mealyDistInit(\mealyState_0')}.
  \end{align*}
  To derive the simple sum from the double sum, we rely on the fact that $f_\prefHist(\mealyState_0) = \mealyState$ holds if and only if $\update(f_{\prefHist'}(\mealyState_0), \concurState, \concurAction) = \mealyState$ and $\nextmove(f_{\prefHist'}(\mealyState_0), \concurState)=\concurActionI$, by definition of $\updateAlt$.

  For the denominator, we obtain from the inductive hypothesis,
  \begin{align*}
    \sum_{m'\in \nextmove(\cdot, \concurState)^{-1}(\concurActionI)}
    \mealyDist{\prefHist'}(\mealyState')
    & =  \sum_{\mealyState'\in \nextmove(\cdot, \concurState)^{-1}(\concurActionI)}
      \sum_{\mealyState_0 \in f_{\prefHist'}^{-1}(\mealyState')}
      \frac{\mealyDistInit(\mealyState_0)}{\sum_{\mealyState_0'\in f_{\prefHist'}^{-1}(\mealyStateSpace)}\mealyDistInit(\mealyState_0')} \\
    & = \sum_{\mealyState_0 \in f_{\prefHist}^{-1}(\mealyStateSpace)}
      \frac{\mealyDistInit(\mealyState_0)}{\sum_{\mealyState_0'\in f_{\prefHist'}^{-1}(\mealyStateSpace)}\mealyDistInit(\mealyState_0')}.
  \end{align*}
  The last equality is a consequence of the definition of $\updateAlt$: recall that $f_{\prefHist}^{-1}(\mealyStateSpace)$ consists of the elements $\mealyState_0$ of $\mealyStateSpace_0(\prefHist')$ such that $\nextmove(f_{w'}(\mealyState_0), \concurState) =\concurActionI$. By combining the two equations above, we immediately obtain Equation~\eqref{equation:mixed:behavioural:1}, ending the inductive argument.

  We now establish the outcome-equivalence of $\mealy$ and $\mealyAlt$.
  Let $\hist = \prefHist\concurState\in \histSet{\game}$ be a history of $\game$ consistent with $\mealy$.
  Let $\concurActionI\in \concurActionSpaceI(\concurState)$ be an action enabled in $\concurState$.
  The probability of $\concurActionI$ being played after $\hist$ under $\mealy$ is given by the weighted sum $\sum_{\mealyState\in \nextmove(\cdot, \concurState)^{-1}(\concurActionI)}\mealyDist{\prefHist}(\mealyState)$.
  Under $\mealyAlt$, the probability of $\concurActionI$ being played is $\nextmoveAlt(f_{\prefHist}, \concurState)(\concurActionI)$.
  It follows from Equation~\eqref{equation:mixed:behavioural:1} that these two probabilities coincide.
  We have shown the outcome-equivalence of strategies $\mealy$ and $\mealyAlt$, ending the proof.
\end{proof}

The construction of a DRD strategy provided in the proof of Theorem~\ref{theorem:mixed:behavioural} leads to an exponential blow-up of the memory state space.
For an RDD strategy $\mealy = \mealyTuple$, we have constructed an outcome-equivalent DRD strategy with a state space consisting of functions $\supp{\mealyDistInit}\to \mealyStateSpace\cup\{\bot\}$, therefore with a state space of size $(|\mealyStateSpace|+1)^{|\supp{\mealyDistInit}|}$.
In the upcoming lemma, we state that an exponential blow-up in the number of initial memory states cannot be avoided in general, even in the turn-based setting.

\begin{lemma}\label{lemma:mixed:behavioural:bound}
  Let $k\in \IN_0$. There exist a two-player turn-based deterministic game (respectively an MDP) $\game_k$ with $k+2$ states, $4k+2$ transitions, $k+2$ actions, and an RDD strategy $\mealy_k$ of $\playerOne$ with $k$ states such that any outcome-equivalent DRD strategy must have at least $2^k-1$ states.
\end{lemma}

\begin{proof}
  We construct a two-player turn-based deterministic game $\game_k= (\concurStateSpace_k, \concurActionSpaceOne_k, \concurActionSpaceTwo_k, \concurTrans_k)$ as follows.
  We let $\concurStateSpace_k = \{\concurState_j\mid 1\leq j\leq k\}\cup\{t, \concurState^\star\}$.
  The sets of actions, common to the two players, are $\arenaActionSpace_k\coloneqq \concurActionSpaceOne_k = \concurActionSpaceTwo_k =\{a_i\mid 1\leq i\leq k\}\cup \{b, \bot\}$.
  All states besides $t$ are controlled by $\playerOne$ in the following sense.
  For all $1\leq j\leq k$, we let $\concurActionSpaceOne_k(s_j) = \{a_j, b\}$ and $\concurActionSpaceTwo_k(s_j) = \{\bot\}$.
  Next, we let $\concurActionSpaceOne_k(s^\star) = \{a_j\mid 1\leq j\leq k\}$ and $\concurActionSpaceTwo_k(s^\star) = \{\bot\}$.
  Finally, for state $t$, we have $\concurActionSpaceOne_k(t) = \{\bot\}$ and $\concurActionSpaceOne_k(t) = \arenaActionSpace_k\setminus\{\bot\}$.
  
  We define the deterministic transition function $\concurTrans_k$ as a function $S_k\times \concurActionSpace_k\to S_k$ (instead of dealing with Dirac distributions over successor states).
  For each $j\in\{1, \ldots, k\}$, all transitions from $s_j$ move back to $t$, i.e., $\concurTrans_k(s_j, a_j, \bot) = \concurTrans_k(s_j, b, \bot) = t$.
  In state $t$, we set for all $j\in\{1, \ldots, k\}$, $\concurTrans_k(t, \bot, a_j) = s_j$ and $\concurTrans_k(t, \bot, b) = s^\star$.
  In state $s^\star$, for all $j\in\{1, \ldots, k\}$, the action $a_j$ labels a self-loop, i.e., we have $\concurTrans_k(s^\star, a_j, \bot) = s^\star$.
  We illustrate the game $\game_3$ in Figure~\ref{figure:game:mixed:behavioural:bound}.
  We omit $\bot$ actions from edge labels to lighten the figure.

  \begin{figure}
    \begin{center}
      \begin{tikzpicture}
        \begin{scope}
          \node[draw, square, minimum size=1cm] (t)  {$t$};
          \node[state, left=of t] (s2)  {$s_2$};
          \node[state, above=of t] (s1)  {$s_1$};
          \node[state, below=of t] (s3)  {$s_3$};
          \node[state, right=of t] (star)  {$s^\star$};
\path[->] (t) edgenode[align=center, left] {$a_1$} (s1);
          \path[->] (s1) edge[bend right=37]node[align=center, left] {$a_1$} (t);
          \path[->] (s1) edge[bend left=37]node[align=center, right] {$b$} (t);
          \path[->] (t) edgenode[align=center, above] {$a_2$} (s2);
          \path[->] (s2) edge[bend right=37]node[align=center, below] {$a_2$} (t);
          \path[->] (s2) edge[bend left=37]node[align=center, above] {$b$} (t);
          \path[->] (t) edgenode[align=center, left] {$a_3$} (s3);
          \path[->] (s3) edge[bend right=37]node[align=center, right] {$a_3$} (t);
          \path[->] (s3) edge[bend left=37]node[align=center, left] {$b$} (t);
          \path[->] (t) edgenode[align=center, above] {$b$} (star);
          \path[->] (star) edge[loop above]node[align=center, above] {$a_1$} (star);
          \path[->] (star) edge[loop right]node[align=center, right] {$a_2$} (star);
          \path[->] (star) edge[loop below]node[align=center, below] {$a_3$} (star);
        \end{scope}
      \end{tikzpicture}
      \caption{The game $\game_3$ from the proof of Lemma~\ref{lemma:mixed:behavioural:bound}. Circles and squares respectively represent states controlled by $\playerOne$ and $\playerTwo$.}\label{figure:game:mixed:behavioural:bound}
    \end{center}
  \end{figure}
  
  We define an RDD strategy $\mealy_k = (\mealyStateSpace, \mealyDistInit, \nextmove, \update)$ of $\playerOne$ as follows.
  We let $\mealyStateSpace = \{1, \ldots, k\}$, and $\mealyDistInit$ is taken to be the uniform distribution over $\mealyStateSpace$.
  The memory update function is taken to be trivial: we set $\update(\mealyState, \concurState, \concurAction) = \mealyState$ for all $\mealyState\in\mealyStateSpace$, $\concurState\in\concurStateSpace_k$ and $\concurAction\in\concurActionSpace_k$.
  For each memory state $\mealyState\in\mealyStateSpace$, we let $\nextmove(\mealyState, \concurState_m) = \nextmove(\mealyState, s^\star) = \arenaAction_m$ and, for all $j\neq \mealyState$, we let $\nextmove(\mealyState, \concurState_j) = b$, and we let $\nextmove(\mealyState, t)=\bot$.
  In $\mealy$, once the initial state is decided, it no longer changes.
  In the memory state $\mealyState\in\mealyStateSpace$, the strategy prescribes action $\arenaAction_\mealyState$ in the states $\concurState_\mealyState$ and $\concurState^\star$, and in states $\concurState_j$ with $j\neq\mealyState$, the strategy prescribes action $b$.

  We now establish that all DRD strategies that are outcome-equivalent to $\mealy$ must have at least $2^k-1$ memory states.
  Let $\mealyAlt = \mealyTupleInStAlt$ be one such FM strategy.
  We give a lower bound on $|\mealyStateSpaceAlt|$ by showing that there must be at least $2^k-1$ distinct distributions of the form $\nextmoveAlt(\cdot, \concurState^\star)$.
  
  Let $E = \{j_1, \ldots, j_\ell\} \subsetneq \mealyStateSpace$ be a proper subset of $\mealyStateSpace$.
  Consider the history ($\bot$ actions are omitted and parentheses are provided to improve readability) $\hist_E = (t\, \arenaAction_{j_1}\,\concurState_{j_1}\,b)(t\,\arenaAction_{j_2}\,\concurState_{j_2}\,b)\ldots (t\,\arenaAction_{j_\ell}\,\concurState_{j_\ell}\,b)t\,b\,\concurState^\star$.
  Let $\mealyState\in E$.
  We see that along the history $\hist_E$, the action $b$ is used in state $\concurState_\mealyState$.
  Therefore, $\hist_E$ is not consistent with the pure FM strategy $(\mealyStateSpace, \mealyState, \nextmove, \update)$ derived from $\mealy$ by setting its initial state to $\mealyState$.
  Similarly, we see that for $m\notin E$, the history $\hist_E$ is consistent with the pure FM strategy $(\mealyStateSpace, \mealyState, \nextmove, \update)$.
  Thus, the set of actions that can be played after $\hist_E$ when following $\mealy_n$ is exactly the set $\{\arenaAction_\mealyState\mid \mealyState\in\mealyStateSpace\setminus E\}\neq\emptyset$.
  Due to the deterministic initialisation and updates of DRD strategies, there must be some $\mealyStateAlt_E\in \mealyStateSpaceAlt$ such that $\supp{\nextmoveAlt(\mealyStateAlt_E, s^\star)} = \{a_m\mid m\in M\setminus E\}$.
  Necessarily, we must have $\supp{\nextmoveAlt(\mealyStateAlt_E, s^\star)}\neq \supp{\nextmoveAlt(\mealyStateAlt_{E'}, s^\star)}$ whenever $E\neq E'$, hence $\mealyStateAlt_E\neq \mealyStateAlt_{E'}$. 
  Consequently, we must have at least one memory state in $\mealyAlt$ per proper subset of $\mealyStateSpace$, i.e., $|\mealyStateSpaceAlt|\geq 2^k-1$.

  The proof of the existence of a suitable MDP remains.
  We explain how to adapt the deterministic game $\game_k$.
  To change $\game_k$ to a suitable MDP $\game_k'$, keep the same state space and remove all actions of $\playerTwo$.
  All transitions are left unchanged except the transitions from state $t$, which are altered as follows.
  When using $\bot$ in $t$, we let there be a uniform probability of reaching a state other than $t$ in $\game_k'$.
  The only (formal) change to be made to $\mealy_k$ to obtain a suitable RDD strategy $\mealy_k'$ of $\game_k'$ is to remove the actions of $\playerTwo$ from updates.

  By performing these changes, we can reuse the argument above for the two-player case to conclude that any DRD strategy that is outcome-equivalent to $\mealy_k'$ in $\game_k'$ requires at least $2^k-1$ memory states.
  This concludes our explanation of how to adapt the game and strategy above to the context of MDPs.
\end{proof}

\subsection{Simulating RRR strategies with DRR ones}\label{subsection:RRR_vs_DRR}
In this section, we establish that DRR strategies are as expressive as RRR strategies, i.e., randomness in the initialisation can be removed.
We outline the ideas behind the construction of a DRR strategy that is outcome-equivalent to a given RRR strategy.
The general idea is to simulate the behaviour of the RRR strategy at the start of the play using a new initial memory state and then move back into the RRR strategy we simulate.

We substitute the random selection of an initial memory element in two stages.
To ensure the first action is selected in the same way under both the supplied strategy and the strategy we construct, we rely on randomised outputs.
The probability of selecting an action $\concurActionI$ in a given state $\concurState$ of the game in our new initial memory state is given as the sum of selecting action $\concurActionI$ in state $\concurState$ in each memory state $\mealyState$ weighed by the initial probability of $\mealyState$.

We then leverage the stochastic updates to simulate that we had been using the original RRR strategy from the start.
To achieve this, we base the update function of the constructed Mealy machine on the equations for the update of the distribution over memory states after a some sequence in $\prefHist\in(\concurStateSpace\concurActionSpace)^*$  takes place (denoted by $\mealyDist{\prefHist}$ in Section~\ref{section:preliminaries}).

We now state our expressiveness result and formalise the construction outlined
above.
\begin{theorem}\label{theorem:RRR:DRR}
  Let $\game = \concurTuple$ be a game.
  Let $\mealy = \mealyTuple$ be an RRR strategy of $\playerI$.
  There exists a DRR strategy $\mealyAlt = \mealyTupleInStAlt$ such that $\mealyAlt$ and $\mealy$ are outcome-equivalent, and such that $|\mealyStateSpaceAlt| = |\mealyStateSpace|+1$.
\end{theorem}

\begin{proof}
  Let us define $\mealyAlt = \mealyTupleInStAlt$ as follows.
  Let $\mealyStateInitAlt$ be such that $\mealyStateInitAlt\notin M$. We set $\mealyStateSpaceAlt = \mealyStateSpace\cup\{\mealyStateInitAlt\}$.
  We let $\updateAlt$ and $\nextmoveAlt$ coincide with $\update$ and $\nextmove$ over $\mealyStateSpace\times\concurStateSpace\times\concurActionSpace$ and $\mealyStateSpace\times\concurStateSpace$ respectively (for the update function, we identify distributions over $\mealyStateSpace$ to distributions over $\mealyStateSpaceAlt$ that assign probability zero to $\mealyStateInitAlt$).
  It remains to define these two functions over $\{\mealyStateInitAlt\}\times \concurStateSpace\times\concurActionSpace$ and $\{\mealyStateInitAlt\}\times \concurStateSpace$ respectively.

  First, we complete the definition of the memory update function $\updateAlt$.
  Let $\concurState\in\concurStateSpace$ and $\concurAction\in\concurActionSpace$.
  We let $\updateAlt(\mealyStateInitAlt, \concurState,\concurAction)(\mealyStateInitAlt)=0$.
  We assume that there exists some $\mealyState_0\in\mealyStateSpace$ such that $\mealyDistInit(\mealyState_0)>0$  and $\nextmove(\mealyState_0, \concurState)(\concurActionI)> 0$ (i.e., the action $\concurActionI$ has a positive probability of being played in $\concurState$ at the start of a play under the strategy $\mealy$).
  We set, for all $\mealyState\in\mealyStateSpace$,
  \[\updateAlt(\mealyStateInitAlt,\concurState,\concurAction)(\mealyState) =
    \frac{\sum_{\mealyState' \in\mealyStateSpace}\mealyDistInit(\mealyState')\cdot\update(\mealyState',\concurState, \concurAction)(\mealyState)\cdot\nextmove(\mealyState',\concurState)(\concurActionI)}{\sum_{\mealyState' \in\mealyStateSpace}\mealyDistInit(\mealyState')\cdot\nextmove(\mealyState',\concurState)(\concurActionI)}.\]
  Whenever we have $\nextmove(\mealyState_0,\concurState)(\concurActionI)= 0$ for all $\mealyState_0\in\mealyStateSpace$ such that $\mealyDistInit(\mealyState_0)>0$, we let $\updateAlt(\mealyStateInitAlt,\concurState,\concurAction)$ be arbitrary.
  
  For the next-move function $\nextmoveAlt$, we define, for all states $\concurState\in\concurStateSpace$ and actions $\concurActionI\in\concurActionSpaceI(\concurState)$,
  \[\nextmoveAlt(\mealyStateInitAlt, \concurState)(\concurActionI) = \sum_{\mealyState\in\mealyStateSpace}\mealyDistInit(\mealyState)\cdot \nextmove(\mealyState, \concurState)(\concurActionI).\]

  It remains to prove that $\mealy$ and $\mealyAlt$ are outcome-equivalent.
  By Lemma~\ref{lemma:prelim:consistent_equivalence}, it suffices to show that both strategies suggest the same distributions over actions along histories consistent with $\mealy$.
  We provide a proof in two steps.
  First, we consider histories with a single state.
  Second, we show that the distributions over memory states coincide in both Mealy machines after any $\prefHist\in \concurStateSpace\concurActionSpace$ that is consistent with $\mealy$ takes place.
  We conclude from this and the construction of $\mealyAlt$ that $\mealy$ and $\mealyAlt$ map all histories that are consistent with $\mealy$ and have more than one state to the same distribution over actions of $\playerI$, ending the proof.

  We show the first claim above.
  Let $\concurState\in\concurStateSpace$ and $\concurActionI\in\concurActionSpaceI(\concurState)$.
  On the one hand, the probability of the action $\concurActionI$ being played after the history $\concurState$ under $\mealy$ is given by
  \[\sum_{\mealyState\in\mealyStateSpace}\mealyDistInit(\mealyState)\cdot \nextmove(\mealyState, \concurState)(\concurActionI).\]
  On the other hand, the probability of this same action $\concurActionI$ being played after the history $\concurState$ under $\mealyAlt$ is given by $\nextmoveAlt(\mealyStateInitAlt, \concurState)(\concurActionI)$.
  These two probabilities coincide by construction.

  Second, let $\prefHist=\concurState\concurAction\in\concurStateSpace\concurActionSpace$ be consistent with $\mealy$.
  Let $\mealyDist{\prefHist}$ and $\mealyDistAlt{\prefHist}$ denote the distribution over memory states after $\prefHist$ takes place under $\mealy$ and $\mealyAlt$ respectively.
  Fix some $\mealyState\in\mealyStateSpace$, and let us prove that $\mealyDist{\prefHist}(\mealyState) = \mealyDistAlt{\prefHist}(\mealyState)$.
  On the one hand, we have
  \begin{align*}
    \mealyDist{\prefHist}(\mealyState)
    & = \frac{\sum_{\mealyState' \in \mealyStateSpace}\mealyDistInit(\mealyState')\cdot \update(\mealyState', \concurState, \concurAction)(\mealyState)\cdot \nextmove(\mealyState', \concurState)(\concurActionI)}{\sum_{\mealyState' \in \mealyStateSpace}\mealyDistInit(\mealyState')\cdot \nextmove(\mealyState', \concurState)(\concurActionI)} \\
    & = \updateAlt(\mealyStateInitAlt, \concurState, \concurActionI)(\mealyState),\end{align*}
  and on the other hand, we have (because $\mealyStateInitAlt$ is the sole initial state of $\mealyAlt$), 
  \[\mealyDistAlt{\prefHist}(\mealyState) =
    \frac{\updateAlt(\mealyStateInitAlt, \concurState, \concurAction)(\mealyState)\cdot \nextmoveAlt(\mealyStateInitAlt, \concurState)(\concurActionI)}{\nextmoveAlt(\mealyStateInitAlt, \concurState)(\concurActionI)}
    = \updateAlt(\mealyStateInitAlt, \concurState, \concurAction)(\mealyState).\]

  We have shown that $\mealyDist{\prefHist} = \mealyDistAlt{\prefHist}$.
  Furthermore, because $\nextmove$ and $\nextmoveAlt$ agree over $\mealyStateSpace\times \concurStateSpace$, and that $\update$ and $\updateAlt$ agree over $\mealyStateSpace\times\concurStateSpace\times\concurActionSpace$, this equality generalises to all $\prefHist\in(\concurStateSpace\concurActionSpace)^+$ that are consistent with $\mealy$.
  It follows that for any history $\hist\in (\concurStateSpace\concurActionSpace)^+\concurStateSpace$ that is consistent with $\mealy$, the images of $\hist$ by the strategies induced by $\mealy$ and $\mealyAlt$ match.
  We conclude that $\mealy$ and $\mealyAlt$ are outcome-equivalent by Lemma~\ref{lemma:prelim:consistent_equivalence}.
\end{proof}

\subsection{Simulating RRR strategies with RDR ones}\label{subsection:RRR_vs_RDR}

We are concerned in this section with the simulation of RRR strategies by RDR strategies, i.e., with substituting randomised outputs with deterministic outputs.
The idea behind the removal of randomisation in outputs is to simulate said randomisation by means of both stochastic initialisation and updates.
These are used to preemptively perform the random selection of an action, simultaneously with the selection of an initial or successor memory state.

Let $\game = \concurTuple$ be a game and let $\mealy = \mealyTuple$ be an RRR strategy of $\playerI$.
We construct an RDR strategy $\mealyAlt = \mealyTupleAlt$ that is outcome-equivalent to $\mealy$ and such that $|\mealyStateSpaceAlt|\leq |\mealyStateSpace|\cdot |\concurStateSpace|\cdot |\concurActionSpaceI|$.
The state space of $\mealyAlt$ consists of pairs $(\mealyState, \stratI)$ where $\mealyState\in\mealyStateSpace$ and $\stratI\colon\concurStateSpace\to\concurActionSpaceI$ is a pure memoryless strategy of $\playerI$.
To achieve our bound on the size of $\mealyStateSpaceAlt$, we cannot consider all pure memoryless strategies of $\playerI$, as there are exponentially many.
We illustrate how we select pure memoryless strategies to achieve the aforementioned bound through the following example.
We apply the upcoming construction on a DRD strategy (which is a special case of RRR strategies) with a single memory state (i.e., a memoryless randomised strategy).

\begin{example}\label{example:RRR:RDR}
  We consider a game $\game = \concurTuple$ where $\concurStateSpace = \{\concurState_1, \concurState_2, \concurState_3\}$, $\concurActionSpaceOne = \{\arenaAction_1, \arenaAction_2, \arenaAction_3\}$ and all actions are enabled in all states.
  We need not specify $\concurActionSpaceTwo$ and $\concurTrans$ for this example.
  For our construction, we fix an order on the actions of $\game$: $\arenaAction_1<\arenaAction_2<\arenaAction_3$.
  
  Let $\mealy = (\{\mealyState\}, \mealyState, \nextmove, \update)$ be the DRD strategy such that $\nextmove(\mealyState, \concurState_1)$  and $\nextmove(\mealyState, \concurState_2)$ are uniform distributions over $\{\arenaAction_1, \arenaAction_2\}$ and $\concurActionSpaceOne$ respectively,
  and $\nextmove(\mealyState, \concurState_3)$ is defined by $\nextmove(\mealyState, \concurState_3)(\arenaAction_1) = \frac{1}{3}$, $\nextmove(\mealyState, \concurState_3)(\arenaAction_2) = \frac{1}{6}$ and
  $\nextmove(\mealyState, \concurState_3)(\arenaAction_3) = \frac{1}{2}$.

  Figure~\ref{figure:RRR:RDR} illustrates the probability of each action being chosen in each state as the length of a segment.
  Let us write $0 = x_1 < x_2 < x_3 < x_4 < x_5 = 1$ for all of the endpoints of the segments appearing in the illustration.
  For each index $k\in\{1, \ldots, 4\}$, we define a pure memoryless strategy $\strat{k}$ that assigns to each state the action lying in the segment above it in the figure.
  For instance, $\strat{2}$ is such that $\strat{2}(\concurState_1) = \arenaAction_1$ and $\strat{2}(\concurState_2) = \strat{2}(\concurState_3) = \arenaAction_2$.
  Furthermore, for all $k\in\{1, \ldots, 4\}$, the length $x_{k+1}-x_k$ of its corresponding interval denotes the probability of the strategy being chosen during stochastic updates.
    \begin{figure}[htb]
    \begin{center}
      \begingroup
      \def\width{0.1\textwidth}
      \def\widthtwo{0.2\textwidth}
      \def\widththree{0.3\textwidth}
      \def\widthhalf{0.05\textwidth}
      \def\widthandhalf{0.15\textwidth}
      \begin{tikzpicture}[outer sep=0pt, node distance=0pt, minimum height=0.7cm, minimum width=\width, align=center]
        \node[rectangle, draw] (s1) {$\concurState_1$};
        \node[rectangle, draw, minimum width=\widththree, right=of s1] (a1s1) {$a_1$};
        \node[rectangle, draw, minimum width=\widththree, right=of a1s1] (a2s1) {$a_2$};
        \node[rectangle, draw, below=of s1] (s2) {$\concurState_2$};
        \node[rectangle, draw, minimum width=\widthtwo, right=of s2] (a1s2) {$\arenaAction_1$};
        \node[rectangle, draw, minimum width=\widthtwo, right=of a1s2] (a2s2) {$\arenaAction_2$};
        \node[rectangle, draw, minimum width=\widthtwo, right=of a2s2] (a3s2) {$\arenaAction_3$};
        \node[rectangle, draw, below=of s2] (s3) {$\concurState_3$};
        \node[rectangle, draw, minimum width=\widthtwo, right=of s3] (a1s3) {$\arenaAction_1$};
        \node[rectangle, draw, minimum width=\width, right=of a1s3] (a2s3) {$\arenaAction_2$};
        \node[rectangle, draw, minimum width=\widththree, right=of a2s3] (a3s3) {$\arenaAction_3$};
        \node[rectangle, below=of s3] (sig) {$\strat{k}$};
        \node[left sided, minimum width=\widthtwo, right=of sig] (sig1) {$\strat{1}$};
        \node[left sided, minimum width=\width, right=of sig1] (sig2) {$\strat{2}$};
        \node[left sided, minimum width=\width, right=of sig2] (sig3) {$\strat{3}$};
        \node[two sided, minimum width=\widthtwo, right=of sig3] (sig4) {$\strat{4}$};

        \node[rectangle, minimum width=\width, below = of sig] (s0) {};
        \node[rectangle, minimum width=\widthtwo, right = of s0] (proba1) {};
        \node[left of=proba1, node distance=\width] {$x_1=0$};
        \node[rectangle, minimum width=\width, right=of proba1] (proba2) {};
        \node[left of=proba2, node distance=\widthhalf] {$x_2=\frac{1}{3}$};
        \node[rectangle, minimum width=\width, right=of proba2] (proba3) {};
        \node[left of=proba3, node distance=\widthhalf] {$x_3=\frac{1}{2}$};
        \node[rectangle, minimum width=\widthtwo, right=of proba3] (proba4) {};
        \node[left of=proba4, node distance=\width] {$x_4=\frac{2}{3}$};
        \node[right of=proba4, node distance=\width] {$x_5=1$};

        \node[right sided, minimum width=\width, above= of a2s3] {};
        \node[right sided, minimum width=\width, above= of sig3] {};
        \node[right sided, minimum width=\width, right= of a1s1] {};
        \node[left sided, minimum width=\width, left= of a2s1] {};

      \end{tikzpicture}
      \endgroup
      \caption{Representation of cumulative probability of actions under strategy $\mealy$ and derived memoryless strategies.}\label{figure:RRR:RDR}
    \end{center}
  \end{figure}
  
  We construct an RDR strategy $\mealyAlt = \mealyTupleAlt$ that is outcome-equivalent to $\mealy$ in the following way.
  We let $\mealyStateSpaceAlt = \{\mealyState\}\times\{\strat{1}, \strat{2}, \strat{3}, \strat{4}\}$.
  The initial distribution is given by $\mealyDistInitAlt(\mealyState, \strat{k}) = x_{k+1} - x_{k}$, i.e., the probability of $\strat{k}$ in the illustration.
  We set, for any $j, k\in\{1, \ldots,4\}$, $\concurState\in\concurStateSpace$ and $\concurActionOne\in\concurActionSpaceOne$, $\updateAlt((\mealyState, \strat{k}),\concurState,\concurActionOne)((\mealyState, \strat{j})) = x_{j+1} - x_j$.
  Finally, we let $\nextmoveAlt((\mealyState, \strat{k}), \concurState) = \strat{k}(\concurState)$ for all $k\in\{1, \ldots, 4\}$ and $\concurState\in\concurStateSpace$. 

  The argument for the outcome-equivalence of $\mealyAlt$ and $\mealy$ is the following; for any state $\concurState\in\concurStateSpace$, the probability of moving into a memory state $(\mealyState, \strat{k})$ such that $\strat{k}(\concurState) = \arenaAction$ is by construction the probability $\nextmove(\mealyState, \concurState)$.
  \hfill$\lhd$
\end{example}

In the previous example, we had a unique memory state $\mealyState$ and we defined some memoryless strategies from the next-move function partially evaluated in this state (i.e., from $\nextmove(\mealyState, \cdot)$).
In general, each memory state may have a different partially evaluated next-move function, and therefore we must define some memoryless strategies for each individual memory state.
For each memory state, we can bound the number of derived memoryless strategies by $|\concurStateSpace|\cdot |\concurActionSpaceI|$; we look at cumulative probabilities over actions (of which there are at most $|\concurActionSpaceI|$) for each state.
This explains our announced bound on $|\mealyStateSpaceAlt|$.

Furthermore, in general, the memory update function is not trivial.
Generalising the construction above can be done in a straightforward manner to handle updates.
Intuitively, the probability to move to some memory state of the form $(\mealyState, \stratI)$ is given by the probability of moving into $\mealyState$ multiplied by the probability of $\sigma$ (in the sense of Figure~\ref{figure:RRR:RDR}).

We now formally state our result in the general setting and provide its proof.
The Mealy machine we construct has updates that do not depend on the actions of the player who owns it; this property is useful when we study games of imperfect information in Section~\ref{section:imperfect:information}.

\begin{theorem}\label{theorem:RRR:RDR}
  Let $\game = \concurTuple$ be a game.
  Let $\mealy = \mealyTuple$ be an RRR strategy of $\playerI$.
  There exists an RDR strategy $\mealyAlt = \mealyTupleAlt$ such that $\mealyAlt$ and $\mealy$ are outcome-equivalent, and such that $|\mealyStateSpaceAlt|\leq |\mealyStateSpace|\cdot (|\concurStateSpace|\cdot (|\concurActionSpaceI| - 1) + 1)$.
  Furthermore, the updates of $\mealyAlt$ do not depend on the actions of $\playerI$.
\end{theorem}

\begin{proof}
  Let us fix a linear order on the set of actions $\concurActionSpaceI$, denoted by <.
  Fix some $\mealyState\in\mealyStateSpace$.
  We let $\segmentpoint{1}{\mealyState} < \ldots < \segmentpoint{\ell(\mealyState)}{\mealyState}$ denote the elements of the set
  \[\left\{\sum_{\concurActionIAlt < \concurActionI} \nextmove(\mealyState, \concurState)(\concurActionIAlt) \mid \concurState\in \concurStateSpace,\, \concurActionI\in\concurActionSpaceI\right\}\]
  that are strictly inferior to $1$, and let $\segmentpoint{\ell(\mealyState)+1}{\mealyState} = 1$.
  These $\segmentpoint{j}{\mealyState}$ represent the cumulative probability provided by $\nextmove(\mealyState, \cdot)$ over actions of $\playerI$ taken in order, for each state of $\game$.
  For each $j \in\{1, \ldots, \ell(\mealyState)\}$, we define a memoryless strategy $\segmentstrat{j}{\mealyState}\colon \concurStateSpace\to\concurActionSpaceI$ as follows: we have $\segmentstrat{j}{\mealyState}(\concurState) = \concurActionI$ if
  $\sum_{\concurActionIAlt < \concurActionI} \nextmove(\mealyState,\concurState)(\concurActionIAlt) \leq \segmentpoint{j}{\mealyState} <
  \sum_{\concurActionIAlt \leq \concurActionI} \nextmove(\mealyState, \concurState)(\concurActionIAlt)$.
  In other words, for any state $\concurState\in\concurStateSpace$, we have $\segmentstrat{j}{\mealyState}(\concurState) = \concurActionI$ whenever $\segmentpoint{j}{\mealyState}$ is at least the cumulative probability of actions strictly inferior to $\concurActionI$ in $\nextmove(\mealyState,\concurState)$ and at most the cumulative probability of actions up to action $\concurActionI$ included.
  Refer to Figure~\ref{figure:RRR:RDR} of Example~\ref{example:RRR:RDR} for an explicit illustration.
  We refer to $\segmentpoint{j+1}{\mealyState} - \segmentpoint{j}{\mealyState}$  as the probability of $\segmentstrat{j}{\mealyState}$ in the sequel.

  Let $\mealyState\in\mealyStateSpace$, $\concurState\in\concurStateSpace$ and $\concurActionI\in \concurActionSpaceI(\concurState)$.
  We show that we can relate  $\nextmove(\mealyState,\concurState)(\concurActionI)$ and the sum of the probabilities of each $\segmentstrat{j}{\mealyState}$ such that $\segmentstrat{j}{\mealyState}(\concurState)=\concurActionI$ as follows.
  First, we introduce some notation.
  Let $I(\mealyState, \concurState, \concurActionI)$ denote the set of indices $j$ such that $\segmentstrat{j}{\mealyState}(\concurState)= \concurActionI$, i.e., the indices such that the $j$th strategy related to $m$ prescribes action $\concurActionI$ in $\concurState$.
  It holds that
  \begin{equation}\label{equation:RRR:RDR:1}
    \sum_{j\in I(\mealyState, \concurState, \concurActionI)}
    (\segmentpoint{j+1}{\mealyState} - \segmentpoint{j}{\mealyState}) = \nextmove(\mealyState, \concurState)(\concurActionI).
  \end{equation}
  Let $\concurState\in\concurStateSpace$ and $\concurActionI\in \concurActionSpaceI(\concurState)$.
  Equation~\eqref{equation:RRR:RDR:1} can be proven as follows.
  First, note that all indices $j$ appearing in the sum are consecutive by construction.
  Therefore, the sum above is telescoping and is equal to $\segmentpoint{j^++1}{\mealyState} - \segmentpoint{j^-}{\mealyState}$, where $j^+$ and $j^-$ denote the largest and smallest indices in the sum respectively.
  By construction, we have $\segmentpoint{j^-}{\mealyState} = \sum_{\concurActionIAlt <\concurActionI}\nextmove(\mealyState, \concurState)(\concurActionIAlt)$ and
  $\segmentpoint{j^++1}{\mealyState} = \sum_{\concurActionIAlt \leq \concurActionI}\nextmove(\mealyState, \concurState)(\concurActionIAlt)$.
  We conclude that $\segmentpoint{j^++1}{\mealyState} - \segmentpoint{j^-}{\mealyState} = \nextmove(\mealyState, \concurState)(\concurActionI)$, proving Equation~\eqref{equation:RRR:RDR:1}.
  This equation is used to establish the outcome-equivalence of $\mealy$ with the strategy defined below.

  We  now define an RDR strategy $\mealyAlt = \mealyTupleAlt$.
  We define
  \[\mealyStateSpaceAlt = \{(\mealyState, \segmentstrat{j}{\mealyState}) \mid \mealyState\in\mealyStateSpace,\, 1\leq j \leq \ell(\mealyState)\}.\]
  The initial distribution and update function of $\mealyAlt$ are derived from those of $\mealy$ multiplied with the probability of the memoryless strategy that appears in the second component of the memory state of $\mealyAlt$ into which we move.
  The initial distribution $\mealyDistInitAlt$ is defined as
  $\mealyDistInitAlt((\mealyState, \segmentstrat{j}{\mealyState})) =
  \mealyDistInit(\mealyState)\cdot (\segmentpoint{j+1}{\mealyState} - \segmentpoint{j}{\mealyState})$
  for all $(\mealyState, \segmentstrat{j}{\mealyState})\in \mealyStateSpaceAlt$. The update function is defined as
  $\updateAlt((\mealyState, \segmentstrat{j}{\mealyState}), \concurState, \concurAction)((\mealyState', \segmentstrat{k}{\mealyState'})) =
  \update(\mealyState, \concurState, \concurActionAlt)(\mealyState')\cdot (\segmentpoint{k+1}{\mealyState'} - \segmentpoint{k}{\mealyState'})$, where $\concurActionAlt = (\segmentstrat{j}{\mealyState}(\concurState), \concurActionTwo)$ if $i=1$ (respectively $\concurActionAlt = (\concurActionOne, \segmentstrat{j}{\mealyState}(\concurState))$ if $i=2$),
  for all $(\mealyState, \segmentstrat{j}{\mealyState}), (\mealyState', \segmentstrat{k}{\mealyState'})\in \mealyStateSpaceAlt$, $\concurState\in\concurStateSpace$ and $\concurAction\in\concurActionSpace$.
  We remark that this update function does not depend on the action of $\playerI$ given as input.
  Finally, the deterministic next-move function of $\mealyAlt$ is defined as
  $\nextmoveAlt((\mealyState, \segmentstrat{j}{\mealyState}), \concurState) = \segmentstrat{j}{\mealyState}(\concurState)$ for all
  $(\mealyState, \segmentstrat{j}{\mealyState})\in \mealyStateSpaceAlt$ and all $\concurState\in\concurStateSpace$.

  We now prove the outcome-equivalence of $\mealy$ and $\mealyAlt$.
  For any $\prefHist\in (\concurStateSpace\concurActionSpace)^*$, let $\mealyDist{\prefHist}$ (resp.~$\mealyDistAlt{\prefHist}$) denote the distribution over $\mealyStateSpace$ (resp.~$\mealyStateSpaceAlt$) after $\prefHist$ has occurred under strategy $\mealy$ (resp.~$\mealyAlt$).
  It follows from Lemma~\ref{lemma:prelim:consistent_equivalence} and the definition of strategies derived from FM strategies that it suffices to establish, for all histories $\hist = \prefHist\concurState$ consistent with $\mealy$, that the following holds:
  \begin{equation}\label{equation:RRR:RDR:2}
    \sum_{\mealyState\in\mealyStateSpace} \mealyDist{\prefHist}(\mealyState)\cdot \nextmove(\mealyState, \concurState)(\concurActionI) =
    \sum_{\mealyState\in\mealyStateSpace}\sum_{j\in I(\mealyState, \concurState, \concurActionI)} \mealyDistAlt{\prefHist}((\mealyState, \segmentstrat{j}{\mealyState})).
  \end{equation}

  To prove Equation~\eqref{equation:RRR:RDR:2}, we rely on the following property: for any $\prefHist\in (\concurStateSpace\concurActionSpace)^*$ consistent with $\mealy$, $\mealyDist{\prefHist}(\mealyState)$ is proportional to $\mealyDistAlt{\prefHist}((\mealyState, \segmentstrat{j}{\mealyState}))$.
  To be precise, for any $\prefHist\in (\concurStateSpace\concurActionSpace)^*$ consistent with $\mealy$, we have
  \begin{equation}\label{equation:RRR:RDR:3}
    \mealyDistAlt{\prefHist}((\mealyState, \segmentstrat{j}{\mealyState})) = (\segmentpoint{j+1}{\mealyState}-\segmentpoint{j}{\mealyState})\cdot\mealyDist{\prefHist}(\mealyState).
  \end{equation}
  To show Equation~\eqref{equation:RRR:RDR:3}, we proceed by induction.
  Consider the empty word $\prefHist= \varepsilon$.
  Because $\mealyDistInit=\mealyDist{\varepsilon}$ and $\mealyDistInitAlt =\mealyDistAlt{\varepsilon}$, Equation~\eqref{equation:RRR:RDR:3} follows from the definition of $\mealyDistInitAlt$.
  Let us now assume inductively that for $\prefHist'\in (\concurStateSpace\concurActionSpace)^*$ consistent with $\mealy$, we have Equation~\eqref{equation:RRR:RDR:3} and let us prove it for $\prefHist = \prefHist'\concurState\concurAction$ consistent with $\mealy$.
  Fix $(\mealyState, \segmentstrat{j}{\mealyState})\in \mealyStateSpaceAlt$.

  To invoke the inductive relation between $\mealyDistAlt{\prefHist}$ and $\mealyDistAlt{\prefHist'}$, we must have that $\prefHist$ is consistent with $\mealyAlt$.
  There exist $\mealyState'\in\supp{\mealyDist{\prefHist'}}$ such that $\nextmove(\mealyState', \concurState)(\concurActionI)>0$ and $k\in I(\mealyState', \concurState, \concurActionI)$ (this set is non-empty due to $\nextmove(\mealyState', \concurState)(\concurActionI)>0$).
  By the induction hypothesis, we obtain $\mealyDistAlt{\prefHist'}((\mealyState', \segmentstrat{k}{\mealyState'})) > 0$, which is sufficient to conclude that $\prefHist$ is consistent with $\mealyAlt$.  
  We thus obtain, from the equation relating $\mealyDistAlt{\prefHist}$ and $\mealyDistAlt{\prefHist'}$,
  \[\mealyDistAlt{\prefHist}((\mealyState, \segmentstrat{j}{\mealyState})) =
    \frac{\sum_{\mealyState'\in M}\sum_{k \in I(\mealyState', \concurState, \concurActionI)}
      \mealyDistAlt{\prefHist'}((\mealyState', \segmentstrat{k}{\mealyState'})) \cdot
      \updateAlt((\mealyState', \segmentstrat{k}{\mealyState'}), \concurState, \concurAction )((\mealyState, \segmentstrat{j}{\mealyState}))}
    {\sum_{\mealyState'\in M}\sum_{k \in I(\mealyState', \concurState, \concurActionI)}
      \mealyDistAlt{\prefHist'}((\mealyState', \segmentstrat{k}{\mealyState'}))}.\]
  The numerator of the above can be rewritten as follows, by successively using the
  definition of $\updateAlt$ followed by the inductive hypothesis and
  Equation~\eqref{equation:RRR:RDR:1}:
  \begin{align*}
    \sum_{\mealyState'\in M}
    &\sum_{k\in I(\mealyState', \concurState, \concurActionI)}
      \mealyDistAlt{\prefHist'}((\mealyState', \segmentstrat{k}{\mealyState'})) \cdot
      \update(\mealyState', \concurState, \concurAction)(\mealyState) \cdot
      (\segmentpoint{j+1}{\mealyState}-\segmentpoint{j}{\mealyState}) \\
    = & (\segmentpoint{j+1}{\mealyState}-\segmentpoint{j}{\mealyState})\cdot
        \sum_{\mealyState'\in M} \left(\update(\mealyState', \concurState, \concurAction)(\mealyState) \cdot
        \mealyDist{\prefHist'}(\mealyState') \cdot
        \sum_{k\in I(\mealyState', \concurState, \concurActionI)}
        (\segmentpoint{k+1}{\mealyState'}-\segmentpoint{k}{\mealyState'})\right) \\
    = & (\segmentpoint{j+1}{\mealyState}-\segmentpoint{j}{\mealyState})\cdot
        \sum_{\mealyState'\in M} \update(\mealyState', \concurState,\concurAction)(\mealyState) \cdot
        \mealyDist{\prefHist'}(\mealyState')\cdot
        \nextmove(\mealyState', \concurState)(\concurActionI).
  \end{align*}
  Following the same reasoning, the denominator can be rewritten as
  \[\sum_{\mealyState'\in M} \mealyDist{\prefHist'}(\mealyState') \cdot \nextmove(\mealyState', \concurState)(\concurActionI).\]
  By combining the equations above and the formula for the update of $\mealyDist{\prefHist}$, we obtain
  $\mealyDistAlt{\prefHist}((\mealyState, \segmentstrat{j}{\mealyState})) = (\segmentpoint{j+1}{\mealyState} - \segmentpoint{j}{\mealyState})\cdot \mealyDist{\prefHist}(\mealyState)$, ending the proof of Equation~\eqref{equation:RRR:RDR:3}.

  We show that Equation~\eqref{equation:RRR:RDR:3} implies Equation~\eqref{equation:RRR:RDR:2}, which will prove that $\mealy$ and $\mealyAlt$ are outcome-equivalent.
  Let $\hist = \prefHist\concurState\in\histSet{\game}$ be a history consistent with $\mealy$.
  Let $\concurActionI\in \concurActionSpaceI(\concurState)$.
  The probability that the action $\concurActionI$ is chosen after history $\hist$ under $\mealy$ is given by $\sum_{\mealyState\in\mealyStateSpace}\mealyDist{\prefHist}(\mealyState)\cdot\nextmove(\mealyState, \concurState)(\concurActionI)$.
  The probability that $\concurActionI$ is selected after $\hist$ under $\mealyAlt$, on the other hand, is given by
  \begin{align*}
    \sum_{\mealyState\in\mealyStateSpace}\sum_{j\in I(\mealyState, \concurState, \concurActionI)}\mealyDistAlt{\prefHist}((\mealyState,\segmentstrat{j}{\mealyState}))
    &= \sum_{\mealyState\in\mealyStateSpace}\left(\mealyDist{\prefHist}(\mealyState) \cdot \sum_{j\in I(\mealyState, \concurState, \concurActionI)}
      (\segmentpoint{j+1}{\mealyState}-\segmentpoint{j}{\mealyState}) \right) \\
    &= \sum_{\mealyState\in\mealyStateSpace}\mealyDist{\prefHist}(\mealyState) \cdot\nextmove(\mealyState, \concurState)(\concurActionI).
  \end{align*}
  In the above, the first equation is obtained from Equation~\eqref{equation:RRR:RDR:3} and the second equation follows from Equation~\eqref{equation:RRR:RDR:1}.
  This concludes the argument for the outcome-equivalence of our two FM strategies.

  To end the proof of this theorem, we prove the upper bound on $|\mealyStateSpaceAlt|$ given in the statement of the result.
  For any memory state $\mealyState\in\mealyStateSpace$, $\ell(\mealyState)$ is bounded by $|\concurStateSpace|\cdot (|\concurActionSpaceI|-1) + 1$: by definition of the numbers $\segmentpoint{j}{\mealyState}$, we see that we must have $\ell(\mealyState)\leq |\concurStateSpace|\cdot |\concurActionSpaceI|$.
  To obtain the aforementioned bound, observe that for all $\concurState\in\concurStateSpace$, we have $\sum_{\concurActionIAlt<\min\concurActionSpaceI}\nextmove(\mealyState, \concurState)(\concurActionIAlt) = 0$, i.e., $0$ admits (at least) $|\concurStateSpace|$ different writings in the set of the $\segmentpoint{j}{\mealyState}$s, hence $\ell(\mealyState) \leq |\concurStateSpace|\cdot |\concurActionSpaceI| - (|\concurStateSpace| - 1) = |\concurStateSpace|\cdot (|\concurActionSpaceI| - 1) + 1$.
  Therefore, we have at most $|\concurStateSpace|\cdot (|\concurActionSpaceI| - 1) + 1$ pairs of the form $(\mealyState, \segmentstrat{j}{\mealyState})$ per memory state $\mealyState\in\mealyStateSpace$.
  It follows that $|\mealyStateSpaceAlt|\leq |\mealyStateSpace|\cdot (|\concurStateSpace|\cdot (|\concurActionSpaceI| - 1) + 1)$.
\end{proof}

\begin{remark}
  The choice of the order on the set of actions fixed at the start of the previous proof influences the size of the constructed strategy.
  It is not necessary to use the same ordering of actions for all memory states.
  The order is used to define all memoryless strategies of the form $\segmentstrat{j}{\mealyState}$, which do not interact with strategies associated to other memory states.
  For this reason, it is possible to use different orderings on actions depending on the memory state $\mealyState$ that is considered.
  \hfill$\lhd$
\end{remark}

\begin{remark}
  The upper bound on the number of memory states given in the statement of Theorem~\ref{theorem:RRR:RDR} can be slightly improved in a turn-based setting.
  In general, we can replace the term $|\concurStateSpace|$ in the bound by the number of states that $\playerI$ controls (more precisely, by the number of $\playerI$-controlled states with at least two enabled actions).
  \hfill$\lhd$
\end{remark}

\section{Separating classes of strategies}\label{section:strictness}
We now discuss the separation of strategies given in the lattice of Figure~\ref{figure:lattice}, and in particular we consider the strictness of inclusions.
All separation results hold in one-player deterministic games with a single state and two actions.
This is one of the simplest possible settings to show that strategy classes are distinct.
Indeed, in a game with a single state and a single action, the only strategy is to always play the unique action, and therefore all strategy classes collapse into one.
For the entirety of this section, we let $\simpleGame$ denote the game depicted in Figure~\ref{figure:game:mixed:behavioural}, and we provide strategies of $\simpleGame$ to show that strategy classes differ.
We complement the separating strategies from $\simpleGame$ with problem instances from the literature for which strategies from some class suffice whereas strategies from the compared class do not.

\begin{figure}
  \begin{center}
    \begin{tikzpicture}
      \node[state] (s0)  {$s$};
\path[->] (s0) edge[loop left]node[align=center, left] {$a$} (s0);
      \path[->] (s0) edge[loop right]node[align=center, right] {$b$} (s0);
    \end{tikzpicture}
    \caption{The MDP $\simpleGame$ with a single state and two actions.}\label{figure:game:mixed:behavioural}
  \end{center}
\end{figure}

We illustrate FM strategies witnessing non-inclusions asserted in the lattice of Figure~\ref{figure:lattice} in Figures~\ref{figure:strictness:1} and~\ref{figure:strictness:2}.
The Mealy machines are interpreted as follows.
Edges that exit memory states read a game state (omitted in these figures due to
$s$ being the sole involved game state) and split into edges labelled by an action and a probability of this action being played, e.g., for $c\in \{a, b\}$ and $p\in [0, 1]$, the notation $c\mid p$ indicates that the probability of playing action $c$ in the current memory state is $p$.
In Figure~\ref{figure:strictness:2}, the edges are further split after the choice of an action for randomised updates.
The edge labels following this second split represent the probabilities of stochastic updates.
This second split is omitted whenever an update is deterministic.

\begin{figure}
  \captionsetup[subfigure]{justification=centering}
  \centering
    \begin{subfigure}[b]{0.3\textwidth}
    \centering
    \begin{tikzpicture}
      \begin{scope}[initial text={$\frac{1}{2}$},]
        \node[state, initial below] (m0) {$m_1$};
        \node[stochastics, node distance=0.5cm, above = of m0] (m0s) {};
        \node[state, initial below, right = of m0] (m1) {$m_2$};
        \node[stochastics, node distance=0.5cm, above = of m1] (m1s) {};
\path[-] (m0) edge[bend left, in=120](m0s);
        \path[->] (m0s) edge[bend left, out=60]node[right] {$a\mid 1$} (m0);
        \path[-] (m1) edge[bend left, in=120](m1s);
        \path[->] (m1s) edge[bend left, out=60]node[right] {$b\mid 1$} (m1);
      \end{scope}
    \end{tikzpicture}
    \caption{DDR $\nsubseteq$ RDD.}
    \label{figure:DDD:strict:RDD}
  \end{subfigure}
  \hfill
  \begin{subfigure}[b]{0.3\textwidth}
    \centering
    \begin{tikzpicture}
      \node[state, initial below, initial text={\phantom{$\frac{1}{2}$}}] (m0) {$m$};
      \node[stochastics, node distance=0.5cm, above = of m0] (m0s) {};
\path[-] (m0) edge (m0s);
      \path[->] (m0s) edge[bend left, out=60]node[right] {$a\mid \frac{1}{2}$} (m0);
      \path[->] (m0s) edge[bend right, out=300]node[left] {$b\mid \frac{1}{2}$} (m0);
    \end{tikzpicture}
    \caption{RDD $\subsetneq$ DRD.}
    \label{figure:RDD:strict:DRD}
  \end{subfigure}
  \hfill
  \begin{subfigure}[b]{0.3\textwidth}
    \centering
    \begin{tikzpicture}
      \begin{scope}[initial text={$\frac{1}{2}$}]
        \node[state, initial below] (m0) {$m_1$};
        \node[stochastics, node distance=0.5cm, above = of m0] (m0s) {};
        \node[state, initial below, right = of m0] (m1) {$m_2$};
        \node[stochastics, node distance=0.5cm, above = of m1] (m1s) {};
\path[-] (m0) edge (m0s);
        \path[->] (m0s) edge[bend left, out=60]node[right] {$a\mid \frac{1}{2}$} (m0);
        \path[->] (m0s) edge[bend right, out=300]node[left] {$b\mid \frac{1}{2}$} (m0);
        \path[-] (m1) edge[bend left, in=120](m1s);
        \path[->] (m1s) edge[bend left, out=60]node[right] {$b\mid 1$} (m1);
      \end{scope}
    \end{tikzpicture}
    \caption{DRD $\subsetneq$ RRD.}
    \label{figure:DRD:strict:RRD}
  \end{subfigure}
  \caption{Depictions of Mealy machines witnessing the strictness of three
    inclusions asserted in Figure~\ref{figure:lattice}. 
    For the sake of readability, we do not label transitions by $s$ as it is the
    sole state the Mealy machines can read in $\simpleGame$, and the only state with a choice in the games of Figure~\ref{figure:strictness:games}.}
  \label{figure:strictness:1}
\end{figure}

The rest of the section is structured as follows.
We discuss the strict inclusion of DDD in RDD and show that RDD is not included in DDR in Section~\ref{subsection:pure_different_mixed}.
Section~\ref{subsection:mixed_different_behavioural} complements the previous Section~\ref{subsection:mixed_included_behavioural} and presents a DRD strategy that has no outcome-equivalent RDD counterpart.
The strict inclusion of the class DRD in the class of RRD strategies is covered in Section~\ref{subsection:behavioural_different_RRD}.
Finally, we prove that DDR is not included in RRD in Section~\ref{subsection:DDR_different_RRD}, which implies that DDR is incomparable to the strategy classes RDD, DRD and RRD.

\subsection{DDD strategies are strictly less expressive than RDD ones}\label{subsection:pure_different_mixed}
Pure FM strategies are less powerful than RDD strategies.
The latter class of strategies can induce non-Dirac distributions over the plays of $\simpleGame$, whereas the former cannot.
We illustrate a strategy that has no outcome-equivalent DDD strategy in Figure~\ref{figure:DDD:strict:RDD}.
Furthermore, there is no DDR strategy that is outcome equivalent to the strategy depicted in Figure~\ref{figure:DDD:strict:RDD}: DDR strategies lack the ability to provide a randomised action at the first step of a game.
We obtain the following result.

\begin{lemma}\label{lemma:rrd_notin_ddr}
  There exists an RDD strategy of $\playerOne$ in $\simpleGame$ such that there is no outcome-equivalent DDR strategy (in particular, there is no outcome-equivalent DDD strategy).
\end{lemma}

We now provide a setting and example from~\cite{DBLP:journals/lmcs/EtessamiKVY08,DBLP:journals/fmsd/RandourRS17} in which RDD strategies can satisfy a specification that pure strategies cannot.
We consider MDPs with several reachability objectives with absorbing targets.

Let $\game = \mdpTuple$ be an MDP and let $k\geq 1$.
Given $\target\subseteq\mdpStateSpace$, we let $\reach{\target}$ denote the set of plays in which a state of $\target$ occurs; this set of plays is the \textit{reachability objective} for target $\target$.
For all $1\leq j\leq k$, we let $\target_j\subseteq\mdpStateSpace$ be a set of \textit{absorbing states}, i.e., for all $\mdpState\in\target_j$ and all $\mdpAction\in\mdpActionSpace(\mdpState)$, $\mdpTrans(\mdpState, \mdpAction)(\mdpState)=1$.
Given a vector $q = (q_j)_{1\leq j\leq k}\in ([0, 1]\cap\IQ)^k$ and an initial state $\mdpState_\init\in\mdpStateSpace$, we consider the problem of determining the existence of a strategy $\stratOne$ such that for all $1\leq j\leq k$, $\proba^{\stratOne}_{\mdpState_\init}(\reach{\target_j})\geq q_j$.
If there exists such a strategy, we say that $q$ is \textit{achievable} (from $\mdpState_\init$).
We give an instance of the problem that illustrates that pure strategies do not suffice below.

\begin{figure}
  \centering
    \begin{subfigure}[b]{0.34\textwidth}
    \centering
    \begin{tikzpicture}[node distance=0.7cm]
      \node[state] (s) {$s$};
      \node[state, left = of s, accepting] (t1) {$t_1$};
      \node[state, right = of s, accepting, dashed] (t2) {$t_2$};
\path[->] (s) edge node[above] {$a$} (t1);
      \path[->] (s) edge node[above] {$b$} (t2);
    \end{tikzpicture}
    \vspace{1,5cm}
    \caption{An MDP with two reachability objectives.}
    \label{figure:strictness:game1}
  \end{subfigure}
  \hfill
  \begin{subfigure}[b]{0.3\textwidth}
    \centering
    \begin{tikzpicture}
      \node[state] (s) {$s$};
      \node[state, right = of s, accepting] (t1) {$t$};
\path[->] (s) edge node[align=center,above] {$(a, a)$\\$(b, b)$} (t1);
      \path[->] (s) edge[loop below] node[align=center,below] {$(a, b)$\\$(b, a)$} (s);
    \end{tikzpicture}
    \caption{A concurrent reachability game.}
    \label{figure:strictness:game2}
  \end{subfigure}
  \hfill
  \begin{subfigure}[b]{0.34\textwidth}
    \centering
    \begin{tikzpicture}
      \node[state] (s) {$s$};
      \node[node distance=1.5cm, right = of s] (mid) {};
      \node[state, node distance=0.1cm, above = of mid, fill=black!10] (bot) {$u$};
      \node[state, node distance=0.1cm, below = of mid] (s1) {$s'$};
\path[->] (s) edge node[align=center,above] {$(a, a)$\\$(b, b)$} (bot);
      \path[->] (s) edge node[align=center,below] {$(a, b)$} (s1);
      \path[->] (s) edge[loop below] node[align=center,below] {$(b, a)$} (s);
    \end{tikzpicture}
    \caption{A concurrent safety game (snowball game~\cite{DBLP:journals/tcs/AlfaroHK07}).}
    \label{figure:strictness:game3}
  \end{subfigure}
  \caption{Games we use to further illustrate the separation of classes of strategies via example specifications.
    States depicted without outgoing transitions have outgoing self-loops that are omitted to lighten figures.}
  \label{figure:strictness:games}
\end{figure}

\begin{example}
  Consider the MDP depicted in Figure~\ref{figure:strictness:game1} and let $\arenaState$ be the initial state.
  We consider the two targets $\target_1 = \{t_1\}$ and $\target_2 = \{t_2\}$ and the vector $q = (\frac{1}{2}, \frac{1}{2})$.
  It is clear that no pure strategy witnesses the achievability of $q$; any pure strategy achieves the vector $(1, 0)$ or $(0, 1)$ if it chooses action $a$ or $b$ in $\arenaState$ respectively.
  However, there is an RDD strategy that witnesses the achievability of $q$; any extension of the strategy depicted in Figure~\ref{figure:DDD:strict:RDD} that accounts for the new game states $t_1$ and $t_2$ achieves $q$.
  \hfill$\lhd$
\end{example}

It turns out that RDD strategies suffice to witness that a vector is achievable no matter the considered instance of the problem.
We provide a short proof of this statement below.

\begin{lemma}
  Let $q$ be an achievable vector in the MDP $\game$ with respect to the reachability objectives $\reach{\target_1}$, \ldots, $\reach{\target_k}$ for absorbing targets $\target_1$, \ldots, $\target_k\subseteq\mdpStateSpace$.
  There exists an RDD strategy witnessing the achievability of $q$.
\end{lemma}
\begin{proof}
  It is shown in~\cite{DBLP:journals/lmcs/EtessamiKVY08} that the set of achievable vectors is a polyhedral set.
  Furthermore, the vertices of this set are achievable by pure memoryless strategies.
  It follows that any achievable vector is dominated by a convex combination of vectors achievable by pure memoryless strategies.
  We conclude that RDD strategies suffice to achieve $q$. 
\end{proof}

\subsection{RDD strategies are strictly less expressive than DRD ones}\label{subsection:mixed_different_behavioural}

The goal of this section is to show that there exists a DRD strategy that cannot be emulated by any RDD strategy.
Let us first explain some intuition behind this statement.
Intuitively, an RDD strategy can only randomise once at the start between a finite number of pure FM (DDD) strategies.
After this initial randomisation, the sequence of actions prescribed by the RDD strategy is fixed relative to the play in progress.
Any DRD strategy that chooses an action randomly at each step, such as the strategy depicted in Figure~\ref{figure:RDD:strict:DRD}, i.e., the strategy playing actions $a$ and $b$ with uniform probability at each step in $\simpleGame$, cannot be reproduced by an RDD strategy.
Indeed, this randomisation generates an infinite number of patterns of actions.
These patterns cannot all be captured by an RDD strategy due to the fact that its initial randomisation is over a finite set.

\begin{lemma}\label{lemma:behavioural:mixed:strictness}
  There exists a DRD strategy of $\playerOne$ in $\simpleGame$ such that there is no outcome-equivalent RDD strategy.
\end{lemma}

\begin{proof}
  Let $\stratOne\colon \{s\}\to \dist{\{a, b\}}$ be the memoryless strategy in $\simpleGame$ induced by the Mealy machine depicted in Figure~\ref{figure:RDD:strict:DRD}.
  The distribution $\stratOne(s)$ is the uniform distribution over $\{a, b\}$.
  The strategy $\stratOne$ induces a probability distribution over plays of $\simpleGame$ such that all plays have a probability of zero.
  Indeed, let $\play$ be a play of $\simpleGame$.
  One can view the singleton $\{\play\}$ as the decreasing
  intersection $\bigcap_{k\in\IN}\cyl{\play_{|k}}$.
  Hence, the probability of $\{\play\}$ is the limit of the probability of $\cyl{\play_{|k}}$ when $k$ goes to infinity.
  One can easily show that the probability under $\stratOne$ of $\cyl{\play_{|k}}$ is $\frac{1}{2^k}$.
  It follows that the probability of $\{\play\}$ is zero.

  We now establish that there is no outcome-equivalent RDD strategy.
  First, let us recall that any RDD strategy can be presented as a distribution over a finite number of pure FM strategies.
  Given that there are no probabilities on the transitions of $\simpleGame$, for any pure strategy $\stratOne^{\mathit{pure}}$, there is a single outcome under $\stratOne^{\mathit{pure}}$.
  We can infer that, for any RDD strategy of $\simpleGame$, there must be at least one play that has a non-zero probability, and therefore this strategy cannot be outcome-equivalent to $\stratOne$, ending the proof.
\end{proof}

We present a setting in which RDD strategies do not suffice, whereas DRD strategies suffice.
We study concurrent reachability games.
Let $\game=\concurTuple$ be a game, $\concurState_\init\in\concurStateSpace$ be an initial state and $\target\subseteq\concurStateSpace$ be a set of target vertices.
We consider the reachability objective $\reach{\target}$ again.
In a concurrent zero-sum reachability game, the goal of $\playerOne$ is to maximise the worst-case probability of $\reach{\target}$.
Formally, we say that a strategy $\stratOne$ of $\playerOne$ ensures the threshold $q\in[0, 1]$ from $\concurState_\init$ if  $\inf_{\stratTwo}\proba^{\stratOne, \stratTwo}_{\concurState_\init}(\reach{\target})\geq q$, where $\stratTwo$ ranges over strategies of $\playerTwo$.
The goal of $\playerOne$ is to ensure the greatest possible threshold.

The supremum of the thresholds that can be ensured from $\concurState_\init$ is called the \textit{value} of $\concurState_\init$.
A strategy is \textit{optimal} from $\concurState_\init$ if it ensures the value of $\concurState_\init$.
If there exists an optimal strategy from a state $\concurState_\init$ of value $1$, we say that $\playerOne$ wins almost-surely from $\concurState_\init$.

We illustrate in the following example that RDD strategies may be unable to ensure thresholds that DRD strategies can in concurrent reachability games.

\begin{example}
  Consider the game depicted in Figure~\ref{figure:strictness:game2} and let $s$ be the initial state.
  Let $\target = \{t\}$ be the target.

  We first claim that there are no RDD strategies of $\playerOne$ that win almost-surely from $s$.
  We fix an RDD Mealy machine $\mealy = \mealyTuple$ of $\playerOne$ and let $\stratOne^\mealy$ denote the strategy it induces.
  For all $\mealyStateInit\in\supp{\mealyDistInit}$, we consider the pure FM strategy $\stratOne^{\mealyStateInit}$ induced by $\mealyTupleInSt$.
  We fix $\mealyStateInit\in\supp{\mealyDistInit}$ and a pure strategy $\stratTwo$ of $\playerTwo$ such that for all histories $\hist$ ending in $s$, $\stratTwo(\hist)\neq\stratOne^{\mealyStateInit}(\hist)$.
  It follows that $\proba^{\stratOne^{\mealyStateInit}, \stratTwo}_{\concurState}(\reach{\target}) = 0$.
  This implies that $\mealy$ is not almost-surely winning from $s$ because, by the law of total probability, we have
  \[\proba^{\stratOne^\mealy, \stratTwo}_{\concurState}(\reach{\target}) =
    \sum_{\mealyState\in\mealyStateSpace}\mealyDistInit(\mealyState)\cdot
    \proba^{\stratOne^\mealyState, \stratTwo}_{\concurState}(\reach{\target}).\]

  On the other hand, the memoryless randomised strategy depicted in Figure~\ref{figure:RDD:strict:DRD} is almost-surely winning.
  At each round prior to a visit of $t$, no matter the choices of $\playerTwo$, this strategy ensures a probability of $\frac{1}{2}$ of matching the action of $\playerTwo$.
  It follows that this strategy is almost-surely winning. \hfill$\lhd$
\end{example}

In full generality, there need not exist optimal strategies in concurrent reachability games~\cite{DBLP:journals/tcs/AlfaroHK07}.
Nonetheless, memoryless randomised strategies (which are a restricted class of DRD strategies) can be used to ensure any possible threshold in these games.
In particular, if there exists an optimal strategy, there always exists one that is memoryless.
We summarise these results in the following theorem.
\begin{theorem}[\cite{DBLP:journals/tcs/AlfaroHK07,gog23}]
  In all concurrent reachability games, if a threshold $q$ can be ensured by $\playerOne$, then there exists a memoryless strategy that ensures $q$.
\end{theorem}

\subsection{DRD strategies are strictly less expressive than RRD ones}
\label{subsection:behavioural_different_RRD}

In this section, we show that there exists an RRD strategy that has no outcome-equivalent DRD strategy.
The example we provide is based on existing results for concurrent safety games, i.e., games where the goal is the complement of a reachability objective.
Given a game $\game=\concurTuple$, we let $\target\subseteq\concurStateSpace$ be a set of states and let $\safe{\target}$ be the \textit{safety objective}, which is the set of plays that do not traverse $\target$.
A strategy $\stratOne$ of $\playerOne$ in $\game$ is said to be \textit{positively winning} for the safety objective $\safe{\target}$ from an initial state $\concurState_\init$ if for all strategies $\stratTwo$ of $\playerTwo$, $\proba^{\stratOne, \stratTwo}_{\concurState_\init}(\safe{\target}) > 0$.

Consider the game depicted in Figure~\ref{figure:strictness:game3} with the safety objective $\safe{\{u\}}$ and consider $\concurState$ to be its initial state.
It is shown in~\cite{DBLP:journals/tcs/AlfaroHK07} that $\playerOne$ does not have a positively winning DRD strategy in this game.
The authors of~\cite{DBLP:journals/corr/abs-1006-1404} show however there exists a positively winning RRD strategy.
The Mealy machine of Figure~\ref{figure:DRD:strict:RRD} matches their positively winning RRD Mealy machine.

The main idea underlying the strategy induced by this Mealy machine is the following.
It attempts the action $a$ at all steps with a positive probability due to memory state $m_1$.
It also has a positive probability of never playing $a$ due to memory state $m_2$.
Therefore, $a$ is played after a history $s(bs)^k$ with a probability that decreases to zero as $k$ increases, as otherwise $a$ would eventually occur almost-surely.

This behaviour cannot be achieved with a DRD strategy.
The distribution over memory states of a DRD strategy following a history is a Dirac distribution due to the deterministic initialisation and deterministic updates.
It follows that DRD strategies suggest actions with probabilities given directly by the next-move function, i.e., the image of a DRD strategy is finite.
It follows that there is no DRD strategy that is outcome-equivalent to the strategy depicted in Figure~\ref{figure:DRD:strict:RRD}.
We formalise this argument in the proof of the following lemma.

\begin{lemma}\label{lemma:rrd:behavioural:strictness}
  There exists an RRD strategy of $\playerOne$ in $\simpleGame$ such that there is no outcome-equivalent DRD strategy.
\end{lemma}

\begin{proof}
  We consider the RRD strategy $\stratOne$ induced by the Mealy machine $\mealy = (\mealyStateSpace, m_1, \nextmove, \update)$ depicted in Figure~\ref{figure:DRD:strict:RRD}.
  For any $\prefHist\in (\{s\}\{a, b\})^*$, let $\mealyDist{\prefHist}$ denote the distribution over $\mealyStateSpace$ after $\prefHist$ as taken place under $\mealy$.
  It can be shown by induction that for any $k\in\IN$, $\mealyDist{(sb)^k}(m_1) = 1 - \mealyDist{(sb)^k}(m_2) = \frac{1}{2^k+1}$ and for any $\prefHist\in (\{s\}\{a, b\})^*$ with at least one occurrence of $a$, $\mealyDist{w}(m_1) = 1$.
  It follows that for any $k\in\IN$, $\stratOne((sb)^ks)(a) = \frac{1}{2(2^{k}+1)}$ and $\stratOne((sb)^ks)(b) = \frac{2^{k+1}+1}{2(2^{k}+1)}$, and for any history $\hist$ containing an occurrence of $a$, $\stratOne(\hist)(a) = \stratOne(\hist)(b) = \frac{1}{2}$.
  We obtain that $\stratOne$ plays the action $a$ with positive probabilities that can be arbitrarily small and that all histories of $\simpleGame$ are consistent with $\stratOne$.
  
  We now show that no DRD strategy is outcome-equivalent to $\stratOne$.
  Let $\mealyAlt= \mealyTupleInStAlt$ denote a DRD strategy and let $\stratAltOne$ denote its induced strategy.
  By Lemma~\ref{lemma:prelim:consistent_equivalence}, $\stratAltOne$ is outcome-equivalent to $\stratOne$ if and only if both strategies are equal, as all histories are consistent with $\stratOne$.
  For all $\hist\in\histSet{\simpleGame}$, due to the deterministic initialisation and updates of $\mealyAlt$, we have $\stratAltOne(\hist) = \nextmoveAlt(\mealyStateAlt, \last{\hist})$ for some $\mealyStateAlt\in\mealyStateSpaceAlt$.
  In particular, $\stratAltOne$ cannot play the action $a$ with arbitrarily small positive probabilities as it can only assign finitely many distributions to histories.
  We conclude that $\stratAltOne\neq\stratOne$, which ends the proof.
\end{proof}

We return to positively winning strategies in concurrent safety games.
It is argued in~\cite{DBLP:journals/corr/abs-1006-1404} that RRR strategies are sufficient to win positively in any concurrent safety game.
We build on their argument to show that RRD strategies suffice to win positively in any concurrent safety game.

Each state in a concurrent safety game can be assigned a rank.
States of highest rank are those from which $\playerTwo$ wins almost-surely for their dual reachability objective.
States of minimal rank, if they are not simultaneously of maximal rank, are those from which $\playerOne$ can surely enforce the safety objective no matter the strategy of $\playerTwo$, i.e., $\playerOne$ has a (memoryless) strategy such that all plays consistent with this strategy that start from a state of minimal rank satisfy the safety objective.

Let $\concurState\in\concurStateSpace$ be a state that is positively winning.
There exists an action of $\playerOne$, which we will call a \textit{sound action}, and a set $\concurActionSpaceTwo_\star(\concurState)\subseteq\concurActionSpaceTwo(\concurState)$ of actions of $\playerTwo$ such that the sound action surely prevents moving to states of higher rank against all actions in $\concurActionSpaceTwo_\star(\concurState)$.
Furthermore, for actions of $\playerTwo$ outside of $\concurActionSpaceTwo_\star(\concurState)$, there is an action of $\playerOne$ that moves to a state of strictly lower rank with positive probability.
For instance, in the snowball game depicted in Figure~\ref{figure:strictness:game3}, for state $\concurState$, the action $b$ is a sound action for $\concurState$ with respect to $\concurActionSpaceTwo_\star(\concurState) = \{a\}$.

The property we require on our strategy to win positively is to use a strategy much like that of Figure~\ref{figure:DRD:strict:RRD}.
On the one hand, it must have a positive probability of only using sound actions from any point: this way, the safety objective is ensured whenever $\playerTwo$ only uses actions in the sets of the form $\concurActionSpaceTwo_\star(\concurState)$ in the remainder of the play.
On the other hand, to account for the possibility of $\playerTwo$ taking an action outside of $\concurActionSpaceTwo_\star(\concurState)$ in state $\concurState$, all actions should have a positive probability of occurring in all rounds, so a vertex of lower rank can be reached with positive probability in this case.

Because the state space is finite, one of two cases occurs.
If $\playerTwo$ only resorts to actions compatible with sound actions from some point on, then the safety objective is satisfied with positive probability because sound actions are guaranteed to be always played from some point on with positive probability.
Otherwise, states of minimal ranks are reached with positive probability, from which $\playerOne$ can surely avoid $\target$. 

The idea of the RRR strategy proposed in~\cite{DBLP:journals/corr/abs-1006-1404} to obtain the behaviour described above is to rely on pairs of memory states.
In a pair, one memory state only proposes sound actions and the other memory state suggests all actions uniformly at random.
When initialising the Mealy machine and each time there is a change in the rank of states, to ensure the resulting strategy has the property above, a stochastic memory update is used to give a uniform probability over such a pair of states.

We show that it suffices to randomise once at the start, for each rank (besides the maximum and minimum one), whether only sound actions should be suggested or whether we should play uniformly at random.
This allows us to avoid stochastic updates and obtain an RRD strategy.
\begin{theorem}\label{theorem:safety positive:RRD}
  Let $\game=\concurTuple$ be a game and $\target\subseteq\concurStateSpace$ be a set of states.
  There exists an RRD strategy $\mealy$ such that, for all $\concurState_\init\in\concurStateSpace$, if there exists a positively winning strategy from $\concurState_\init$ for the objective $\safe{\target}$, then $\mealy$ is positively winning from $\concurState_\init$.
\end{theorem}
\begin{proof}
  We assume that there exists at least some state from which $\playerOne$ wins positively, otherwise the result is immediate.
  We use properties of~\cite[Algorithm 3]{DBLP:journals/tcs/AlfaroHK07}, which computes the set of almost-surely winning states in a concurrent reachability game, i.e., the complement of the set of positively winning states for the player with a safety objective.
  Each iteration of this algorithm computes two sets of states that are positively winning for $\playerOne$ and (essentially) removes them from the state space.
  Therefore, it yields a non-increasing sequence $\concurStateSpace=U_0\supseteq U_1 \ldots \supseteq U_{k}$ of sets of states ($k+2$ being double the number of iterations of the algorithm) such that $\concurStateSpace\setminus U_{k}$ is the set of positively winning states for $\playerOne$.
  In particular, note that $\target\subseteq U_{k}$.
  Let, for all $\concurState\in\concurStateSpace$, $\rk(\concurState)$ be the greatest $j$ such that $\concurState\in U_j$.

  The sequence of sets $(U_j)_{1\leq j\leq k}$ has the following property.
  For all states $\concurState\in\concurStateSpace$ such that $\rk(\concurState) < k$, there exists a sound action $\concurActionOne_{\mathsf{sd}}(s)\in\concurActionSpaceOne(\concurState)$ and a subset
  $\concurActionSpaceTwo_\star(\concurState)\subseteq\concurActionSpaceTwo(\concurState)$ such that
  (i) for all $\concurActionTwo\in\concurActionSpaceTwo_\star(\concurState)$ and all $\concurState'\in\supp{\concurTrans(\concurState,\concurActionOne_{\mathsf{sd}}(s), \concurActionTwo)}$, $\rk(\concurState')\leq\rk(\concurState)$,
  and (ii) for all $\concurActionTwo\in\concurActionSpaceTwo(\concurState)\setminus\concurActionSpaceTwo_\star(\concurState)$, there exists an action $\concurActionOne\in\concurActionSpaceOne(\concurState)$ and a state $\concurState'\in\supp{\concurTrans(\concurState,\concurActionOne, \concurActionTwo)}$ such that $\rk(\concurState') < \rk(\concurState)$.
  These conditions follow from the structure of the algorithm.
  In particular, the pure memoryless strategy of $\playerOne$ that only plays sound actions, when played from states of rank $0$, is such that all of its outcomes satisfy $\safe{\target}$ (i.e., states of rank $0$ are surely winning for $\playerOne$).
  
  We now define an RRD strategy.
  Let $\mealy = \mealyTuple$ such that $\mealyStateSpace = \{\mathsf{sd}, \mathsf{un}\}^{k-1}$ ($\mathsf{sd}$ and $\mathsf{un}$ respectively stand for sound and uniform).
  We let $\mealyDistInit$ be a uniform distribution over $\mealyStateSpace$.
  Let $\mealyState = (\mealyState_j)_{1\leq j\leq k-1}\in\mealyStateSpace$ and $\concurState\in\concurStateSpace$.
  If $\rk(\concurState)=k$, we let $\nextmove(\mealyState, \concurState)$ be arbitrary.
  Otherwise, if $\rk(\concurState) = 0$ or $\mealyState_{\rk(\concurState)} = \mathsf{sd}$, we let $\nextmove(\mealyState, \concurState)$ be a Dirac distribution on $\concurActionOne_{\mathsf{sd}}(\concurState)$.
  Otherwise (if $0< \rk(\concurState) < k$ and $\mealyState_{\rk(\concurState)} = \mathsf{un}$), we let $\nextmove(\mealyState, \concurState)$ be a uniform distribution over $\concurActionSpaceOne(\concurState)$.
  The deterministic memory updates are trivial: for all $\mealyState\in\mealyStateSpace$, $\concurState\in\concurStateSpace$ and $\concurAction\in\concurActionSpace(\concurState)$, we let $\update(\mealyState, \concurState, \concurAction) = \mealyState$.
  Given $\prefHist\in(\concurStateSpace\concurActionSpace)^*$, we let $\mealyDist{\prefHist}$ denote the distribution over memory states of $\mealy$ after $\prefHist$ has taken place.
  For $\mealyState\in\mealyStateSpace$, we let $\stratOne^\mealyState$ be the strategy induced by the Mealy machine obtained by fixing the initial state of $\mealy$ to $\mealyState$.

  We now prove that $\mealy$ induces a positively winning strategy from any state from which $\playerOne$ has a positively winning strategy.
  Let $\concurState_0$ be such a state and let $\stratTwo$ be an arbitrary strategy of $\playerTwo$.
  We use  an inductive argument on histories, starting with the history $\hist_0 = \concurState_0$.
  At step $j$ of the induction, we assume that we have some history $\hist_j = \prefHist_j\concurState_j$ consistent with $\stratTwo$ such that $\rk(\concurState_j) < k- j$ and $\supp{\mealyDist{\prefHist_j}} = \{\mathsf{sd}, \mathsf{un}\}^{\rk(\concurState_j)}\times \mealyStateSpace_j$ for some $\mealyStateSpace_j\subseteq \{\mathsf{sd}, \mathsf{un}\}^{k - \rk(\concurState_j)}$ (this last hypothesis implies that $\hist_j$ is consistent with $\mealy$, otherwise $\mealyDist{\prefHist_j}$ would not be defined).
  This induction hypothesis is clearly satisfied at step $0$ of the induction (positively winning states have rank at most $k-1$).

  We consider two cases.
  First, we assume that, for all extensions $\prefHist_j\hist$ of $\hist_j$, if they are consistent with $\stratTwo$ and only sound actions are used by $\playerOne$ in the suffix $\hist$, then $\supp{\stratTwo(\prefHist_j\hist)}\subseteq\concurActionSpaceTwo_\star(\last{\hist})$.
  We remark that if $\rk(\concurState_j) = 0$, we are necessarily in this case.
  We claim that for all extensions $\prefHist_j\hist$ of $\hist_j$ consistent with $\stratTwo$ in which only sound $\playerOne$ actions occur in $\hist$, it holds that all states in $\hist$ have rank at most $\rk(\concurState_j)$.
  This follows by a straightforward induction using the definition of sound actions and actions in sets $\concurActionSpaceTwo_\star(\concurState')$ (informally, the rank of states cannot increase at each step in this setting).

  By the induction hypothesis, there exists some $\mealyState\in\supp{\mealyDist{\prefHist_j}}$ such that $\mealyState_\ell=\mathsf{sd}$ for all $\ell\leq\rk(\concurState_j)$.
  In particular, $\hist_j$ is consistent with $\stratOne^\mealyState$ due to the definition of updates in $\mealy$.
  It follows from the above that all extensions of $\hist_j$ that are consistent with both $\stratOne^\mealyState$ and $\stratTwo$ satisfy $\safe{\target}$ (because all targets have rank $k$).
  Therefore, only a subset of $\cyl{\hist_j}$ of $\proba^{\stratOne^\mealyState, \stratTwo}_{\concurState}$-measure zero is not included in $\safe{\target}$.
  Therefore, $\proba^{\stratOne^\mealyState, \stratTwo}_{\concurState}(\safe{\target}) \geq \proba^{\stratOne^\mealyState, \stratTwo}_{\concurState}(\cyl{\hist_j}) > 0$.
  We conclude that $\proba^{\stratOne, \stratTwo}_{\concurState}(\safe{\target})>0$ as $\proba^{\stratOne^\mealyState, \stratTwo}_{\concurState}(\safe{\target})$ is the conditional probability of $\safe{\target}$ with respect to $\proba^{\stratOne, \stratTwo}_{\concurState}$ assuming that the initial memory state is $\mealyState$.

  Next, assume that there exists a history $\prefHist_j\hist$ extending $\hist_j$ that is consistent with $\stratTwo$, in which only sound actions are used by $\playerOne$ in the suffix $\hist$ and such that $\supp{\stratTwo(\prefHist_j\hist)}\nsubseteq\concurActionSpaceTwo_\star(\last{\hist})$.
  We assume that $\prefHist_j\hist$ is the shortest such extension of $\hist_j$.
  We fix $\concurActionTwo\in\supp{\stratTwo}(\prefHist_j\hist)\setminus\concurActionSpaceTwo_\star(\last{\hist})$, and $\concurActionOne\in\concurActionSpaceOne(\last{\hist})$ and $\concurState_{j+1}\in\supp{\concurTrans(\last{\hist}, \concurActionOne, \concurActionTwo)}$ such that $\rk(\concurState_{j+1}) < \rk(\last{\hist})$.
  We let $\concurAction = (\concurActionOne, \concurActionTwo)$.
  
  We define $\hist_{j+1} = \prefHist_j\hist\concurAction\concurState_{j+1}$ and show that it satisfies the induction hypothesis above.
  First, by construction, $\hist_{j+1}$ is consistent with $\stratTwo$.
  Second, it holds that $\rk(\last{\hist})\leq\rk(\concurState_j)$.
  This can be shown by the same argument as in the first case, as only sound actions occur in $\hist$ and all $\playerTwo$ actions taken in any state $\concurState$ in $\hist$ are in $\concurActionSpaceTwo_\star(\concurState)$.
  It follows that $\rk(\concurState_{j+1}) < \rk(\concurState_{j})$, implying that $\rk(\concurState_{j+1})< k - (j+1)$.
  Third, it can be shown by a straightforward induction that $\supp{\mealyDist{\prefHist}} = \supp{\mealyDist{\prefHist_j}}$ for $\prefHist$ such that $\prefHist_j\hist = \prefHist\last{\hist}$.
  The omitted inductive argument is based on the fact that all $\playerOne$ actions are sound in $\hist$, are taken in states of rank at most $\rk(\concurState_j)$ and $\supp{\mealyDist{\prefHist_j}} = \{\mathsf{sd}, \mathsf{un}\}^{\rk(\concurState_j)}\times \mealyStateSpace_j$.
  Finally, it holds that $\supp{\mealyDist{\prefHist_j\hist\concurAction}} = \{\mealyState\in\supp{\mealyDist{\prefHist_j}} \mid \mealyState_{\rk(\last{\hist})}=\mathsf{un}\}$ if $\concurActionOne\neq\concurActionOne_{\mathsf{sd}}(\last{\hist})$ and $\supp{\mealyDist{\prefHist_j\hist\concurAction}} = \supp{\mealyDist{\prefHist_j}}$ otherwise.
  By the inductive hypothesis, we obtain that
  \[\supp{\mealyDist{\prefHist_j\hist\concurAction}} =
    \{\mathsf{sd}, \mathsf{un}\}^{\rk(\last{\hist})-1}\times I\times \{\mathsf{sd}, \mathsf{un}\}^{\rk(\concurState_j) - \rk(\last{\hist})}\times \mealyStateSpace_j,\]
  where $I = \{\mathsf{un}\}$ in the first case, and $I = \{\mathsf{sd}, \mathsf{un}\}$ otherwise.
  This shows that we can continue the inductive argument with $\hist_{j+1}$.

  The second case can occur in the worst case only in the $k-1$ first steps of the induction: at step $k$, $\concurState_k$ has rank $0$, which guarantees we find ourselves in the first case. This concludes the proof that $\mealy$ is positively winning from $\concurState_0$.
\end{proof}

\subsection{RRD and DDR strategies are incomparable}
\label{subsection:DDR_different_RRD}
We prove in this section that the classes RRD and DDR of finite-memory strategies are incomparable.
We have previously shown Lemma~\ref{lemma:rrd_notin_ddr}, which states that $\text{RDD}\nsubseteq\text{DDR}$ and therefore implies that $\text{DRD}\nsubseteq\text{DDR}$ and $\text{RRD}\nsubseteq\text{DDR}$.
It remains to show that $\text{DDR}\nsubseteq\text{RRD}$.

We illustrate a DDR strategy of $\simpleGame$ that has no outcome-equivalent RRD strategy in Figure~\ref{figure:DDR:not:RRD}.
For ease of analysis, we illustrate in Figure~\ref{figure:RRD:strict:RRR} a DRR strategy with fewer states that is outcome-equivalent to the Mealy machine depicted in Figure~\ref{figure:DDR:not:RRD}.
The DDR strategy of Figure~\ref{figure:DDR:not:RRD} can be obtained by applying the construction of Theorem~\ref{theorem:RRR:RDR} to the Mealy machine of Figure~\ref{figure:RRD:strict:RRR}.

\begin{figure}
  \centering
  \begin{subfigure}[b]{0.55\textwidth}
  \centering
  \begin{tikzpicture}[every state/.style={minimum size=1cm}]
      \node[state, initial below,] (m1) {$m_1$};
      \node[stochastics, node distance=0.4cm, above = of m1] (m1s) {};
      \node[stochasticc, node distance=0.6cm, above = of m1s] (m1c) {};
      \node[state, left = of m1] (m2a) {$(m_2, a)$};
      \node[stochasticc, left = of m2a] (m2c) {};
      \node[stochastics, node distance=0.5cm, above = of m2c] (m2s) {};
      \node[state, node distance=0.3cm, above = of m2s] (m2b) {$(m_2, b)$};
      \node[state, node distance=0.7cm, right = of m1c] (m3) {$m_3$};
      \node[stochastics, node distance=0.5cm, below = of m3] (m3s) {};
\path[-] (m1) edge (m1s);
      \path[-] (m1s) edgenode[right] {$b$} (m1c);
      \path[->] (m1c) edgenode[above] {$\frac{1}{4}$} (m2b);
      \path[->] (m1c) edgenode[above] {$\frac{1}{4}$} (m2a);
      \path[->] (m1c) edgenode[above] {$\frac{1}{2}$} (m3);
      \path[-] (m2s) edgenode[right] {$b$} (m2c);
      \path[->] (m2a) edge node[midway,stochastics] {} node [below right=0.1] {$a$} (m1);
      \path[-] (m2b) edge (m2s);
      \path[->] (m2c) edge[bend left]node[left] {$\frac{1}{2}$} (m2b);
      \path[->] (m2c) edgenode[below] {$\frac{1}{2}$} (m2a);
      \path[<-] (m3) edge[bend left, in=120]node[right] {$b$} (m3s);
      \path[-] (m3s) edge[bend left, out=60](m3);
  \end{tikzpicture}
  \caption{A DDR strategy witnessing $\mathrm{DDR}\nsubseteq\mathrm{RRD}$.}
  \label{figure:DDR:not:RRD}
  \end{subfigure}\hfill
  \begin{subfigure}[b]{0.4\textwidth}
    \centering
    \begin{tikzpicture}[every state/.style={minimum size=1cm}]
        \node[state, initial below,] (m1) {$m_1$};
        \node[stochastics, node distance=0.4cm, above = of m1] (m1s) {};
        \node[stochasticc, node distance=0.5cm, above = of m1s] (m1c) {};
        \node[state, left = of m1c] (m2) {$m_2$};
        \node[state, right = of m1c] (m3) {$m_3$};
        \node[stochastics, node distance=0.5cm, below = of m3] (m3s) {};
        \node[stochastics, left = of m1] (m2s) {};
\path[-] (m1) edge (m1s);
        \path[-] (m1s) edgenode[right] {$b \mid 1$} (m1c);
        \path[->] (m1c) edgenode[above] {$\frac{1}{2}$} (m2);
        \path[->] (m1c) edgenode[above] {$\frac{1}{2}$} (m3);
        \path[-] (m2) edge (m2s);
        \path[->] (m2s) edgenode[below] {$a\mid \frac{1}{2}$} (m1);
        \path[->] (m2s) edge[bend left]node[left] {$b\mid \frac{1}{2}$} (m2);
        \path[<-] (m3) edge[bend left, in=120]node[right] {$b\mid 1$} (m3s);
        \path[-] (m3s) edge[bend left, out=60](m3);
    \end{tikzpicture}
    \caption{An outcome-equivalent RRR strategy with fewer states.}
    \label{figure:RRD:strict:RRR}
  \end{subfigure}
  \caption{Outcome-equivalent strategies witnessing the non-inclusion
    $\mathrm{DDR}\nsubseteq\mathrm{RRD}$. For the sake of readability, we do not label transitions by $s$ as it is the sole state the Mealy machines can read in $\simpleGame$.
    We omit the probability of actions in Figure~\ref{figure:DDR:not:RRD} as outputs are deterministic.}
  \label{figure:strictness:2}
\end{figure}

Intuitively, these strategies have a non-zero probability of never using action $a$ after any history, while they have a positive probability of using action $a$ at any time besides the first round and right after the action $a$ occurs.
The behaviour described above cannot be reproduced by an RRD strategy.
There are two reasons to this.

First, along any play consistent with an RRD strategy, the support of the distribution over memory states cannot increase in size.
Because of deterministic updates, the probability carried by a memory state $m$ can only be transferred to at most one state, and may be lost if the used action cannot be used while in $m$.
This property does not hold for strategies that have stochastic updates, such as those of Figure~\ref{figure:strictness:2}.

Second, one can force situations in which the size of the support of the
distribution over memory states of an RRD strategy must decrease.
If after a given history $h$, the action $a$ has a positive probability of never being
used despite being assigned a positive probability at each round after $h$,
then at some point there must be some memory state of the RRD strategy that has positive
probability and that assigns (via the next-move function) probability zero to action $a$.
For instance, this is the case from the start with the RRD strategy depicted in
Figure~\ref{figure:DRD:strict:RRD}.
Intuitively, if at all times all memory states in the support of the
distribution over memory states after the current history
assign a positive probability to action $a$, the probability of using
$a$ at each round after $h$ would be bounded from below by the smallest positive
probability assigned to $a$ by the next-move function. Therefore $a$ would eventually be
played almost-surely assuming $h$ has taken place, contradicting the fact that there
was a positive probability of never using action $a$ after $h$.
By using action $a$ at a point in which some memory state in the support of the
distribution over memory states assigns probability zero to $a$, the size of the
support of the memory state distribution decreases.

By design of our DDR strategy, if one assumes the existence of an outcome-equivalent RRD strategy, then it is possible to construct a play along which the size of the support of the distribution over memory states of the RRD strategy decreases infinitely often.
Because this size cannot increase along a play, this is not possible, i.e., there is no such RRD strategy.
We formalise the sketch above in the proof of the following lemma.

\begin{lemma}\label{lemma:ddr_notin_rrd}
  There exists a DDR strategy of $\playerOne$ in $\simpleGame$ such that there is no outcome-equivalent RRD strategy.
\end{lemma}
\begin{proof}
  Consider the Mealy machine $\mealy = (M, m_1, \nextmove, \update)$
  depicted in Figure~\ref{figure:RRD:strict:RRR}. We recall that
  $\mealy$ is a DRR Mealy machine that is outcome-equivalent to the DDR
  strategy illustrated in Figure~\ref{figure:DDR:not:RRD}.
  It therefore suffices to show that
  there are no RRD strategies that are outcome-equivalent to $\mealy$ to
  end this proof.
  
  Let $\stratOne$ denote the strategy induced by $\mealy$. Intuitively,
  $\stratOne$ operates as follows. It always uses $b$ in the first round and otherwise
  has a positive probability of never using action $a$ while always having a positive
  probability of playing $a$ at any round. Whenever the action $a$ is used, the behaviour
  of the strategy resets in the following sense: witnessing action $a$
  ensures that $\mealy$ finds itself in memory state $m_1$ after the update,
  thus the strategy repeats its behaviour from the initial state of $\mealy$.

  Lemma~\ref{lemma:prelim:consistent_equivalence} ensures that we need only study plays
  consistent with $\stratOne$ for matters related to outcome-equivalence.
  The finite sequences of actions that can be suggested by this strategy can be described
  by the regular expression $(b^+a)^*b^*$.
  Therefore, we require only the definition of $\stratOne$ over histories in which the underlying sequence of actions is in this language.
  For any $w\in (\{s\}\{a,b\})^*$, let $\mealyDist{\prefHist}$ denote the distribution over memory states of $\mealy$ after $\prefHist$ has taken place.
  It can be shown by induction that for any $\prefHist\in ((\{s\}\{b\})^+\{s\}\{a\})^*$ and $k\geq 1$, we have $\mealyDist{\prefHist}(m_1) = 1$ and $\mealyDist{\prefHist(sb)^k}(m_2) = 1 - \mealyDist{\prefHist(sb)^k}(m_3) = \frac{1}{2^{k-1} + 1}$.
  It follows that for any history $h$ consistent with $\stratOne$ of the form $s$ or $\hist'a s$ and $k\geq 1$, we have $\stratOne(h)(b) = 1$ and $\stratOne(\hist(bs)^k)(a) = 1 - \stratOne(\hist(bs)^k)(b) = \frac{1}{2^{k}+2}$.

  We show that for any history $h$ consistent with $\stratOne$ in which the last used action is $a$, it holds that $\proba_s^{\stratOne}(\{h(bs)^\omega\})>0$, i.e., there is a positive probability of $a$ never being played again after any occurrence of $a$.
  Let $h$ be one such history.

  We first show that $\proba_s^{\stratOne}(\{h(bs)^\omega\}) = \proba_s^{\stratOne}(\cyl{\hist})\cdot \proba_s^{\stratOne}(\{(sb)^\omega\})$.
  We have, for any $k\in\IN$, $\stratOne(h(bs)^k)(b) = \stratOne(s(bs)^k)(b)$ by definition of $\stratOne$.
  Furthermore, the sequences $(\cyl{s(bs)^k})_{k\in\IN}$ and $(\cyl{\hist(bs)^k})_{k\in\IN}$ respectively
  decrease when taking their intersections
  to the singletons $\{(sb)^\omega\}$ and $\{h(bs)^\omega\}$. We obtain
  the following equations from the definition of $\proba_s^{\stratOne}$:
  \begin{align*}
    \proba_s^{\stratOne}(\{h(bs)^\omega\}) = &
    \lim_{k\to\infty}\proba_s^{\stratOne}(\cyl{\hist(bs)^k})\\
    = & \lim_{k\to\infty}\proba_s^{\stratOne}(\cyl{\hist}) \cdot
    \prod_{\ell=0}^{k-1}\stratOne(h(bs)^\ell)(b) \\
    = & \proba_s^{\stratOne}(\cyl{\hist})\cdot \lim_{k\to\infty}\cdot
    \prod_{\ell=0}^{k-1}\stratOne(s(bs)^\ell)(b) \\
    = & \proba_s^{\stratOne}(\cyl{\hist})\cdot \lim_{k\to\infty}
    \proba_s^{\stratOne}(\cyl{s(bs)^k}) \\
    = & \proba_s^{\stratOne}(\cyl{\hist})\cdot \proba_s^{\stratOne}(\{(sb)^\omega)\}).
  \end{align*}
  
  In light of the above, to show that $\proba_s^{\stratOne}(\{h(bs)^\omega\})>0$,
  it suffices to establish that $\proba_s^{\stratOne}(\{(sb)^\omega\})>0$
  because $h$ is assumed to be consistent with $\stratOne$.
  It can be shown that $\proba_s^{\stratOne}(\{(sb)^\omega\}) = \frac{1}{2}$ as follows:
  \begin{align*}
    \proba_s^{\stratOne}(\{(sb)^\omega\})
    & = \lim_{k\to\infty} \proba_s^{\stratOne}(\cyl{s(bs)^k}) \\
    & = \lim_{k\to\infty} 1 \cdot \prod_{j = 1}^{k-1}\frac{2^{j}+1}{2^j+2} \\
    & = \lim_{k\to\infty} \frac{1}{2^{k-1}}\cdot
      \prod_{j = 1}^{k-1}\frac{2^{j}+1}{2^{j-1}+1} \\
    & = \lim_{k\to\infty} \frac{1}{2^{k-1}} \cdot
      \frac{2^{k-1}+1}{2^{1-1}+1} = \frac{1}{2};
  \end{align*}
  the product of the probabilities of $b$ being played in each round is simplified using the fact that the denominator of a term is double the numerator of the previous one.
  This closes the proof of our claimed inequality.

  We now show that no RRD strategy is outcome-equivalent to $\stratOne$.
  Let $\mealyAlt = \mealyTupleAlt$ be an RRD Mealy machine and let $\stratAltOne$ be the strategy it induces.
  For any $\prefHist\in (\{s\}\{a, b\})^*$, let $\mealyDistAlt{\prefHist}$ denote the distribution over memory states in $\mealyStateSpaceAlt$ after $\prefHist$ has taken place under $\mealyAlt$.
  
  The remainder of the proof is structured as follows; we prove two properties of RRD strategies and use them to show that $\stratAltOne$ cannot be outcome-equivalent to $\stratOne$.
  The first claim if that for any history $h = \prefHist s$ consistent with $\stratAltOne$ and action $c\in\{a, b\}$ such that $\stratAltOne(h)(c) > 0$, we have $|\supp{\mealyDistAlt{\prefHist}}| \geq |\supp{\mealyDistAlt{\prefHist sc}}|$, i.e., the size of the support of the distribution over memory states of $\mealyAlt$ does not increase as the play progresses.
  The second claim is that for any history $h$ consistent with $\stratAltOne$, if the probability of $a$ never appearing again after $h$ is non-zero, i.e., $\proba_s^{\stratAltOne}(\{h(bs)^\omega\}) > 0$, and for any $k\in\IN$, we have $\stratAltOne(h(bs)^k)(a)> 0$, then there exists some $k_0\in\IN$ such that $|\supp{\mealyDistAlt{h(bs)^{k_0}b}}| > |\supp{\mealyDistAlt{h(bs)^{k_0+1}a}}|$.

  Let us first prove the first claim.
  It follows from a careful inspection of how the distribution over memory states is updated from one step to the next.
  Let $h = \prefHist s$ be consistent with $\stratOne$ and $c\in\{a, b\}$ such that $\stratAltOne(h)(c) > 0$.
  For any memory state $\mealyStateAlt\in \mealyStateSpaceAlt$, recall that \[\mealyDistAlt{\prefHist sc}(\mealyStateAlt) = \frac{\sum_{\mealyStateAlt'\in \mealyStateSpaceAlt} \mealyDistAlt{\prefHist}(\mealyStateAlt')\cdot \updateAlt(\mealyStateAlt', s, c)(\mealyStateAlt)\cdot \nextmoveAlt(\mealyStateAlt', s)(c)}{\sum_{m'\in M} \mealyDistAlt{\prefHist}(\mealyStateAlt')\cdot \nextmoveAlt(\mealyStateAlt', s)(c)}.\]
  Because updates are deterministic, for any given $\mealyStateAlt'\in\mealyStateSpaceAlt$, there is a unique $\mealyStateAlt\in\mealyStateSpaceAlt$ such that $\updateAlt(\mealyStateAlt', s, c)(\mealyStateAlt) = 1$.
  Therefore any element in $\supp{\mealyDistAlt{\prefHist}}$ transfers its probability to at most one memory state when deriving $\mealyDistAlt{\prefHist sc}$.
  This ends the proof of the first claim.
  We note (for the proof of the second claim) that if $\mealyStateAlt'\in \supp{\mealyDistAlt{\prefHist}}$ is such that $\nextmoveAlt(\mealyStateAlt', s)(c) = 0$, then $\mealyStateAlt'$ does not transfer its probability to any state, and in this case, we have $|\supp{\mealyDistAlt{\prefHist}}| > |\supp{\mealyDistAlt{\prefHist sc}}|$.

  We now move on to the second claim.
  Let $h$ be consistent with $\stratAltOne$ and assume that $\proba_s^{\stratAltOne}(\{h(bs)^\omega\}) > 0$, and for any $k\in\IN$, we have $\stratAltOne(h(bs)^k)(a)> 0$.
  In light of the comment above regarding the second claim, it suffices to show that for some $k_0\in\IN$, we have some $\mealyStateAlt\in \supp{\mealyDistAlt{h(bs)^{k_0}b}}$ such that $\nextmoveAlt(\mealyStateAlt, s)(a) = 0$.
  Assume towards a contradiction that this is not the case, i.e., for all $k\in\IN$ and all $\mealyStateAlt\in \supp{\mealyDistAlt{h(bs)^kb}}$, we have $\nextmoveAlt(\mealyStateAlt, s)(a) > 0$.
  Let $k\in\IN$.
  We show that the probability $\stratAltOne(h(bs)^{k+1})(a)$ is bounded below by a positive constant independent of $k$.
  This follows from the assumption that $\nextmoveAlt(\mealyStateAlt, s)(a) > 0$ for all $\mealyStateAlt\in\supp{\mealyDistAlt{h(bs)^kb}}$
  via the relations 
  \begin{align*}
    \stratAltOne(h(bs)^{k+1})(a)
    & = \sum_{\mealyStateAlt\in \mealyStateSpaceAlt}\mealyDistAlt{h(bs)^{k}b}(\mealyStateAlt)\cdot\nextmoveAlt(\mealyStateAlt, s)(a) \\
    & \geq \sum_{\mealyStateAlt\in \mealyStateSpaceAlt}\mealyDistAlt{h(bs)^{k}b}(\mealyStateAlt)\cdot \min_{\mealyStateAlt'\in \mealyStateSpaceAlt^{a > 0}}\nextmoveAlt(\mealyStateAlt', s)(a) \\
    & = \min_{\mealyStateAlt'\in \mealyStateSpaceAlt^{a>0}}\nextmoveAlt(\mealyStateAlt', s)(a) > 0,
  \end{align*}
  where $\mealyStateSpaceAlt^{a > 0} = \{\mealyStateAlt\in \mealyStateSpaceAlt\mid \nextmoveAlt(\mealyStateAlt, s)(a) > 0\}$. It follows that
  the action $a$ must be used almost-surely assuming $h$ has taken place, contradicting the fact that $\proba_s^{\stratAltOne}(\{h(bs)^\omega\}) > 0$.
  This ends the proof of the second claim.

  We now show that $\stratAltOne$ cannot be outcome-equivalent to $\stratOne$ by
  contradiction.
  Assume $\stratAltOne$ is outcome-equivalent to $\stratOne$.
  Due to the properties of $\stratOne$ shown above, we can repeatedly use the two claims above to construct a sequence of non-zero natural numbers $(k_\ell)_{\ell\in\IN}$ such that $(|\supp{\mealyDistAlt{\prefHist_\ell}}|)_{\ell\in\IN}$ is an infinite decreasing sequence of natural numbers, where $\prefHist_0 = \varepsilon$ and for all $\ell\in\IN$, $\prefHist_{\ell+1}= \prefHist_\ell(sb)^{k_\ell}sa$.
  This contradicts the well-order of $\IN$.
  This shows that there are no RRD strategies that are outcome-equivalent to $\stratOne$.  
\end{proof}

As in the previous sections, we provide a game and a specification that cannot be accomplished using an RRD strategy, but can be accomplished using a DDR strategy.
In the following example, we consider a two-player turn-based game with several reachability objectives with absorbing targets.
The goal is to construct, if it exists, a strategy that ensures given thresholds for several reachability objectives at once.

\begin{example}\label{example:multireach:sg}
  We consider the two-player turn-based game $\game = \concurTuple$ depicted in Figure~\ref{figure:multireach:sg} (ownership of vertices is distinguished by their shape), originating from~\cite{DBLP:conf/mfcs/ChenFKSW13}.
  As $\game$ is turn-based, we lighten the notation of histories and plays by only indicating the action of the player in control of the state.
  We also simplify notations for updates of Mealy machines by only taking in account the actions we keep in plays.
  We let $\arenaActionSpace = \concurActionSpaceOne\cup\concurActionSpaceTwo$ denote the set of actions.
  We consider three targets: $\target_j = \{t_j\}$ for $j\in\{1, 2, 3\}$.
  \newcommand{\belowDist}{8mm}
  \newcommand{\belowDistC}{14mm}
  \newcommand{\belowDistS}{7mm}
  \newcommand{\phantomDist}{1mm}
  \newcommand{\phantomDistC}{5mm}
  \begin{figure}[t]
    \centering
    \begin{tikzpicture}[node distance=28mm]
      \node[state] (s0) {$s_0$};
      \node[stochasticc, node distance=8mm, right = of s0] (s0c) {};
      \node[draw, square, minimum size=1cm, right = of s0c] (s1) {$s_1$};
      \node[stochasticc, right = of s1] (s1c) {};
      \node[draw, square, minimum size=1cm, right = of s1c] (s2) {$s_2$};
      \node[state, node distance={\belowDist}, below = of s0c] (s3) {$s_3$};
      \node[node distance={\belowDist}, below = of s3] (s3p) {};
      \node[state, accepting, node distance={\phantomDist}, left = of s3p] (t11) {$t_1$};
      \node[stochasticc, node distance={\phantomDistC}, right = of s3p] (s3c) {};
      \node[node distance={\belowDistC}, below = of s3c] (s3cp) {};
      \node[state, accepting, dashed, node distance={\phantomDist}, left = of s3cp] (t21) {$t_2$};
      \node[state, accepting, dotted, node distance={\phantomDist}, right = of s3cp] (t31) {$t_3$};
      \node[stochasticc, node distance={\belowDistS}, below = of s1] (s1check) {};
      \node[node distance={\belowDistC}, below = of s1check] (s1checkp) {};
      \node[state, accepting, node distance={\phantomDist}, left = of s1checkp] (t12) {$t_1$};
      \node[state, node distance={\phantomDist}, right = of s1checkp] (s4) {$s_4$};
      \node[node distance={\belowDist}, below = of s4] (s4p) {};
      \node[state, accepting, dashed, node distance={\phantomDist}, left = of s4p] (t22) {$t_2$};
      \node[state, accepting, dotted, node distance={\phantomDist}, right = of s4p] (t32) {$t_3$};
      \node[state, node distance={\belowDist}, below = of s1c] (s5) {$s_5$};
      \node[node distance={\belowDist}, below = of s5] (s5p) {};
      \node[state, accepting, dashed, node distance={\phantomDist}, left = of s5p] (t23) {$t_2$};
      \node[stochasticc, node distance={\phantomDistC}, right = of s5p] (s5c) {};
      \node[node distance={\belowDistC}, below = of s5c] (s5cp) {};
      \node[state, accepting, node distance={\phantomDist}, left = of s5cp] (t13) {$t_1$};
      \node[state, accepting, dotted, node distance={\phantomDist}, right = of s5cp] (t33) {$t_3$};
      \node[stochasticc, node distance={\belowDistS}, below = of s2] (s2check) {};
      \node[node distance={\belowDistC}, below = of s2check] (s2checkp) {};
      \node[state, accepting, dashed, node distance={\phantomDist}, left = of s2checkp] (t24) {$t_2$};
      \node[state, node distance={\phantomDist}, right = of s2checkp] (s6) {$s_6$};
      \node[node distance={\belowDist}, below = of s6] (s6p) {};
      \node[state, accepting, node distance={\phantomDist}, left = of s6p] (t14) {$t_1$};
      \node[state, accepting, dotted, node distance={\phantomDist}, right = of s6p] (t34) {$t_3$};

      \path[-] (s0) edge node[below] {$\bot$} (s0c);
      \path[->] (s2) edge[bend right=15] node[above] {$\continueAct$} (s0);
      \path[-] (s1) edge node[above] {$\continueAct$} (s1c);
      \path[-] (s1) edge node[right] {$\checkAct$} (s1check);
      \path[-] (s2) edge node[right] {$\checkAct$} (s2check);
      
      \path[->] (s0c) edge node[below] {$\frac{1}{2}$} (s1);
      \path[->] (s0c) edge node[right] {$\frac{1}{2}$} (s3);
      \path[->] (s3) edge node[above left] {$\ell$} (t11);
      \path[-] (s3) edge node[above right] {$r$} (s3c);
      \path[->] (s3c) edge node[above left] {$\frac{1}{4}$} (t21);
      \path[->] (s3c) edge node[above right] {$\frac{3}{4}$} (t31);

      \path[->] (s1check) edge node [above left] {$\frac{1}{3}$} (t12);
      \path[->] (s1check) edge node [above right] {$\frac{2}{3}$} (s4);
      \path[->] (s4) edge node [above left] {$\ell$} (t22);
      \path[->] (s4) edge node [above right] {$r$} (t32);

      \path[->] (s1c) edge node[below] {$\frac{1}{2}$} (s2);
      \path[->] (s1c) edge node[right] {$\frac{1}{2}$} (s5);
      \path[->] (s5) edge node[above left] {$\ell$} (t23);
      \path[-] (s5) edge node[above right] {$r$} (s5c);
      \path[->] (s5c) edge node[above left] {$\frac{1}{2}$} (t13);
      \path[->] (s5c) edge node[above right] {$\frac{1}{2}$} (t33);

      \path[->] (s2check) edge node [above left] {$\frac{1}{3}$} (t24);
      \path[->] (s2check) edge node [above right] {$\frac{2}{3}$} (s6);
      \path[->] (s6) edge node [above left] {$\ell$} (t14);
      \path[->] (s6) edge node [above right] {$r$} (t34);
    \end{tikzpicture}
    \caption{A turn-based stochastic game with multiple reachability objectives~\cite{DBLP:conf/mfcs/ChenFKSW13}.
      Circles and squares respectively represent states controlled by $\playerOne$ and $\playerTwo$.
      The only action enabled for players who do not control a state is $\bot$.
      States $t_1$, $t_2$ and $t_3$ are drawn repeatedly for clarity (duplicates all represent the same state).
    Actions $\continueAct$ and $\checkAct$ of $\playerTwo$ stand for proceed and check respectively.}
    \label{figure:multireach:sg}
  \end{figure}
  
  \begin{figure}
    \centering
    \begin{tikzpicture}
      \node[state, initial left] (m1) {$m_0$};
      \node[state, right = of m1] (m2) {$m_1$};
      \node[stochasticc, right = of m2] (m2c) {};
      \node[state, right = of m2c] (m3) {$m_2$};
\path[->] (m1) edge node[above] {$s_0$} (m2);
      \path[-] (m2) edge node[above] {$s_0$} (m2c);
      \path[->] (m2c) edge[bend left] node[below] {$\frac{1}{2}$} (m2);
      \path[->] (m2c) edge node[above] {$\frac{1}{2}$} (m3);
    \end{tikzpicture}
    \caption{A Mealy machine update scheme for the game of Figure~\ref{figure:multireach:sg}. Updates that do not change the memory state are not depicted.}
    \label{figure:multireach:mm}
  \end{figure}
  
  In~\cite{DBLP:conf/mfcs/ChenFKSW13}, it is shown that there is no DRD strategy $\stratOne$ of $\playerOne$ such that for all strategies $\stratTwo$ of $\playerTwo$, $\proba^{\stratOne, \stratTwo}_{s_0}(\reach{\target_j}) \geq \frac{1}{3}$ for all $j\in\{1, 2, 3\}$, despite there existing an infinite-memory one.
  We prove that (i) there is no RRD strategy that satisfies this specification and (ii) there exists a DDR strategy that does. 

  We let, for $k\in\IN$, $\hist_k = s_0 (\bot s_1\continueAct s_2\continueAct s_0)^k$.
  A description of satisfactory strategies is provided in the technical report~\cite[Appendix B]{TechRep:ChenFKSW13}.
  A strategy $\stratOne$ of $\playerOne$ ensures that all targets are each visited with probability $\frac{1}{3}$ if for all $k\in\IN$, $\stratOne(\hist_k\bot s_3)(\ell) = 1 - \frac{1}{3\cdot 2^{k-1}}$, $\stratOne(\hist_k\bot s_1\checkAct s_4)(\ell) = 1 - \frac{1}{2^{k+2}}$, $\stratOne(\hist_k\bot s_1\continueAct s_5)(\ell) = 1 - \frac{1}{3\cdot 2^{k}}$ and $\stratOne(\hist_k\bot s_1\checkAct s_6)(\ell) = 1 - \frac{1}{2^{k+2}}$, and for all $k\in\IN$, the first two equations are necessary to comply with the specification.

  Let $\mealy$ be an RRD strategy and let $\stratAltOne^\mealy$ be its induced strategy.
  We show that $\stratAltOne^\mealy$ cannot satisfy the multi-objective query by showing that the set of distributions $\{\stratAltOne^\mealy(\hist_k\bot s_3)\mid k\in\IN\}$ must be a finite set, which is incompatible with the requirements given above.

  Let $\mealyDist{\prefHist}$ denote the distribution over memory states after $\prefHist\in(\concurStateSpace\arenaActionSpace)^*$ has taken place under $\mealy$.
  For all $k\in\IN$ and $\mealyState\in\mealyStateSpace$, it holds that $\mealyDist{\hist_k\bot s_3}(\mealyState) = \sum_{\mealyState'\in\mealyStateSpace'}\mealyDistInit(\mealyState')$ for some $\mealyStateSpace'\subseteq\mealyStateSpace$ (which depends on both $k$ and $\mealyState$).
  This follows from the equations for the updates of the distributions $\mealyDist{\prefHist}$.
  In all states along $\hist_k \bot$, $\playerOne$ only has a single action.
  Furthermore, $\mealy$ has deterministic updates.
  Therefore, if $\prefHist$ and $\prefHist\concurState\arenaAction$ are prefixes of $\hist_k \bot$, for all memory states $\mealyState\in\mealyStateSpace$, we obtain $\mealyDist{\prefHist\concurState\arenaAction}(\mealyState)$ is the sum of $\mealyDist{\prefHist}(\mealyState')$ for all memory states $\mealyState'$ such that $\update(\mealyState', \concurState, \arenaAction) = \mealyState$.
  In particular, this implies that the set of distributions $\{\mealyDist{\hist_k\bot} \mid k\in\IN\}$ is finite, which shows that $\{\stratAltOne^\mealy(\hist_k\bot s_3)\mid k\in\IN\}$ is a finite set by definition of the strategy induced by a Mealy machine.

  We now describe a Mealy machine $\mealyAlt$ that induces a strategy that coincides with $\stratOne$ over $\cyl{\concurState_0}$, i.e., that ensures a probability of $\frac{1}{3}$ for all three reachability objectives.
  Once more, we provide a DRR strategy that can be transformed into an outcome-equivalent DDR strategy via the construction underlying Theorem~\ref{theorem:RRR:RDR}.
  We depict the relevant update scheme in Figure~\ref{figure:multireach:mm}; updates that do not change the current memory state are omitted from the figure.
  Let $\mealyDistAlt{\prefHist}$ denote the distribution over memory states of $\mealyAlt$ after $\prefHist\in(\concurStateSpace\arenaActionSpace)^*$ has taken place under $\mealyAlt$.
  Let $k\in\IN$.
  Below, we are interested in the distribution over memory states only for $\prefHist_k\in\{\hist_k\bot, \hist_k\bot s_1\checkAct, \hist_k\bot s_1\continueAct,\hist_k\bot s_1 \continueAct s_2\checkAct\}$: it can be shown by a straightforward induction that we have $\mealyDistAlt{\prefHist_k}(\mealyState_1) = 1 - \mealyDistAlt{\prefHist_k}(\mealyState_2) = \frac{1}{2^k}$.

  We now specify the next-move function of $\mealyAlt$ and describe the strategy $\stratOne^\mealyAlt$ induced by $\mealyAlt$.
  We let $\nextmove(\mealyState_0, s)$ be an arbitrary Dirac distribution for all states $\concurState\in\{s_3, s_4, s_5, s_6\}$ (we require Dirac distributions so our Mealy machine has an outcome-equivalent DDR strategy).
  For $s_3$, we let $\nextmove(\mealyState_1, s_3)(r) = \frac{2}{3}$ and $\nextmove(\mealyState_2, s_3)(\ell) = 1$.
  It follows that for all $k\in\IN$, we have $\stratOne^\mealyAlt(\hist_k\bot s_3)(r) = \frac{2}{3\cdot 2^{k}} = \frac{1}{3\cdot 2^{k-1}}$.
  For $s_4$, we let $\nextmove(\mealyState_1, s_4)(r) = \frac{1}{4}$ and $\nextmove(\mealyState_2, s_4)(\ell) = 1$.
  We obtain that for all $k\in\IN$, we have $\stratOne^\mealyAlt(\hist_k\bot s_2\checkAct s_4)(r) = \frac{1}{4\cdot 2^{k}} = \frac{1}{2^{k+2}}$.
  For $s_5$, we let $\nextmove(\mealyState_1, s_5)(r) = \frac{1}{3}$ and $\nextmove(\mealyState_2, s_5)(\ell) = 1$.
  For all $k\in\IN$, it holds that $\stratOne^\mealyAlt(\hist_k\bot s_2\continueAct s_5)(r) = \frac{1}{3\cdot 2^{k}}$.
  Finally, for $s_6$, we let $\nextmove(\mealyState_1, s_6)(r) = \frac{1}{4}$ and $\nextmove(\mealyState_2, s_6)(\ell) = 1$.
  We conclude that for all $k\in\IN$, $\stratOne^\mealyAlt(\hist_k\bot s_2\continueAct s_2\checkAct s_6)(r) = \frac{1}{4\cdot 2^{k}} = \frac{1}{2^{k+2}}$.
  This shows that $\stratOne^\mealyAlt$ ensures all reachability objectives are satisfied with probability at least $\frac{1}{3}$.
  \hfill $\lhd$
\end{example}

Consider a turn-based stochastic game $\game =\concurTuple$ and targets $\target_1, \ldots, \target_k\subseteq\concurStateSpace$.
The general form of the problem treated in the example above is to decide, given an initial state $\concurState_\init\in\concurStateSpace$ and a threshold vector $q\in ([0, 1]\cap\IQ)^k$ whether there exists a strategy $\stratOne$ of $\playerOne$ such that for all strategies $\stratTwo$ of $\playerTwo$, we have $\proba^{\stratOne,\stratTwo}_{\concurState_\init}(\reach{\target_j})\geq q_j$ for all $j\in\{1, \ldots, k\}$.
It is not known whether RRR strategies  of $\playerOne$ suffice to provide a positive answer whenever possible in general.
However, finite-memory strategies suffice to approximate any vector for which the problem has a positive answer.
More precisely, if $\playerOne$ can ensure $q$ from $\concurState_\init\in\concurStateSpace$, then for all $\varepsilon > 0$, $\playerOne$ has an DRD strategy such that for all strategies $\stratTwo$ of $\playerTwo$ and all $j\in\{1, \ldots, k\}$, it holds that $\proba^{\stratOne, \stratTwo}_{\concurState_\init}(\reach{\target_j})\geq q_j - \varepsilon$~\cite{DBLP:conf/mfcs/ChenFKSW13,DBLP:conf/lics/AshokCKWW20}.

\section{Extension: multiplayer games}
\label{section:multiplayer}

In the previous sections, we have only considered two-player games.
We show that the lattice of Figure~\ref{figure:lattice} extends to games with more than two players.

Let $\nPlayer\geq 1$ be a number of players.
Formally, an $\nPlayer$-player concurrent stochastic game is a tuple $\game = \nConcurTuple$ where $\concurStateSpace$ is a non-empty finite set of states, $\concurActionSpaceI$ is a finite set of actions for each player and $\concurTrans\colon\concurStateSpace\times\concurActionSpaceOne\times\ldots\times\concurActionSpaceN\to\dist{\concurStateSpace}$ is a probabilistic transition function.
We reuse the notation $\concurActionSpace = \concurActionSpaceOne\times\ldots\times\concurActionSpaceN$.
We impose the same constraints as in the two-player case regarding actions enabled in states, i.e., whether an action is available to a player is independent of the choices of others.
Plays, histories, strategies, Mealy machines and probability distributions over plays induced by strategies are defined in a similar way as in the two-player setting.

The definition of outcome-equivalence can be naturally extended to multi-player games.
Instead of quantifying universally over strategies of the other player as is done in the two-player setting, one quantifies universally over strategies of all other players in the definition of outcome-equivalence.
Formally, two strategies $\stratOne$ and $\stratAltOne$ of $\playerOne$ are outcome-equivalent if for all strategies $\stratI$ of $\playerI$ for $2\leq i\leq\nPlayer$ and all $\concurState\in\concurStateSpace$, $\proba^{\stratOne, \stratTwo, \ldots, \strat{\nPlayer}}_{\concurState} = \proba^{\stratAltOne, \stratTwo, \ldots, \strat{\nPlayer}}_{\concurState}$.

A single (fictitious) player derived from a coalition of players has access to more behaviours than the coalition, as the single player can randomise over action profiles whereas individual players can only randomise over their own set of actions.
This implies that all probability distributions over action profiles that can be induced by strategies of the players of the coalition playing individually can be simulated by the fictitious player, but the inverse is not true.
This is the crux of the argument showing that our results carry over to the multi-player setting.

\begin{theorem}\label{theorem:multi:lattice}
  The taxonomy of Figure~\ref{figure:lattice} established in two-player games extends to multiplayer games.
\end{theorem}
\begin{proof}
  All results that witness that two classes in the lattice of Figure~\ref{figure:lattice} are separated (i.e., Lemmas~\ref{lemma:rrd_notin_ddr},~\ref{lemma:behavioural:mixed:strictness},~\ref{lemma:rrd:behavioural:strictness} and~\ref{lemma:ddr_notin_rrd}) hold in one-player  games, which are a subclass of multiplayer games.

  We now prove that the inclusion results extend to this setting.
  Let $\stratClass_1$ and $\stratClass_2$ be two classes of finite-memory strategies referred to in Figure~\ref{figure:lattice} such that the lattice asserts that $\stratClass_1\subseteq\stratClass_2$.
  Let $\game = \nConcurTuple$ be an $\nPlayer$-player game.
  In the following argument, we only consider strategies of $\playerOne$ to simplify notation.
  We let $\game' = (\concurStateSpace, \concurActionSpaceOne, \prod_{2\leq i\leq\nPlayer}\concurActionSpaceI, \concurTrans)$ be the two-player (coalition) game  in which the players other than $\playerOne$ are grouped together.
  Although the sets of histories and plays of $\game$ and $\game'$ differ syntactically (due to the nature of action tuples), there is a natural bijection between these sets.
  For this reason, we identify them.
  Therefore, all strategies of $\playerOne$ in $\game$ are strategies of $\playerOne$ in $\game'$ and vice-versa.
  
  Let $\stratOne\in\stratClass_1$ be a strategy of $\playerOne$.
  Because $\stratClass_1\subseteq\stratClass_2$ holds for two-player games, there exists a strategy $\stratAltOne\in\stratClass_2$ such that $\stratOne$ and $\stratAltOne$ are outcome-equivalent in $\game'$.
  We claim that $\stratOne$ and $\stratAltOne$ are outcome-equivalent in $\game$.
  Let $\stratTwo$, \ldots, $\strat{\nPlayer}$ be strategies of players other than $\playerOne$ and $\concurState\in\concurStateSpace$ be an initial state.
  Consider the strategy $\stratAltTwo$ of the second player in $\game'$ defined by $\stratAltTwo(\hist)(\concurActionTwo, \ldots, \concurActionPl{\nPlayer}) = \prod_{2\leq i \leq \nPlayer}\stratI(\hist)(\concurActionI)$ for all $\hist\in\histSet{\game}$ and all $\concurActionI\in\concurActionSpaceI$ for $2\leq i\leq\nPlayer$.
  By definition of distributions induced by plays and outcome-equivalence of $\stratOne$ and $\stratAltOne$ in $\game'$, we obtain $\proba^{\stratOne, \stratTwo, \ldots, \strat{\nPlayer}}_{\game,\mdpState} = \proba^{\stratOne, \stratAltTwo}_{\game', \mdpState} = \proba^{\stratAltOne, \stratAltTwo}_{\game',\mdpState} = \proba^{\stratAltOne, \stratTwo, \ldots, \strat{\nPlayer}}_{\game,\mdpState}$ (where the subscript also indicates the relevant game), ending the proof.
\end{proof}

\section{Extension: imperfect information}
\label{section:imperfect:information}
This section discusses games of imperfect information, and how our results transfer to this setting.
In Section~\ref{section:imperfect:information:intro}, we introduce definitions and terminology for games of imperfect information.
We discuss finite-memory strategies in this setting in Section~\ref{section:imperfect:information:fm}.
Finally, we close with Section~\ref{section:imperfect:information:lattice}, in which we argue that the lattice of Figure~\ref{figure:lattice} transfers to games with perfect recall and provide an adaptation for games of imperfect recall.

\subsection{Games of imperfect information}\label{section:imperfect:information:intro}
We consider two-player stochastic games of imperfect information played on graphs.
Unlike games of perfect information, the players are not fully informed of the current state of the play and the actions that are used along the play.
Instead, they perceive an \textit{observation} for each state and action, and this observation may be shared between different states and actions, making them indistinguishable.
These observations are not shared between the players; each player perceives the ongoing play differently.

We formalise this game model.
A \textit{concurrent stochastic game of imperfect information} is defined as a tuple $\poGame = \poTuple$ where $\concurTuple$ is a game of perfect information,
$\obsSpaceI$ is a finite set of observations of $\playerI$ for $i\in \{1, 2\}$ and $\obsFunI\colon \poStateSpace\cup \poActionSpaceOne\cup\poActionSpaceTwo\to \obsSpaceI$ is the observation function of $\playerI$, which assigns an observation to each state and action.
We require that for any $i\in\{1, 2\}$, for any two states $\poState, \poState'\in \poStateSpace$, $\obsFunI(\poState) = \obsFunI(\poState')$ implies $\poActionSpace(\poState) = \poActionSpace(\poState')$, i.e., in two indistinguishable states, the same actions are available to both players.
We fix $\poGame$ for the remainder of the section and let $\game$ denote the underlying game of perfect information.

Plays and histories of $\poGame$ are respectively defined as plays and histories of $\game$.
We reuse the notations $\playSet{\poGame}$ and $\histSet{\poGame}$ for the sets of plays of $\poGame$ and histories of $\poGame$ respectively.
We extend the observation functions to pairs of actions and to histories. For any $\poAction = (\poActionOne,\poActionTwo)\in\poActionSpace$, we let $\obsFunI(\poAction) = (\obsFunI(\poActionOne),\obsFunI(\poActionTwo))$ and for all histories $\hist = \poState_0\poAction_0\ldots \poState_k$ of $\poGame$, we let $\obsFunI(h) = \obsFunI(\poState_0) \obsFunI(\concurAction_0)\ldots \obsFunI(\poState_k)$.
This extension is used to define strategies in games of imperfect information.

In our setting, $\playerI$ has \textit{perfect recall} if $\playerI$ can distinguish their own actions.
Formally, $\playerI$ has perfect recall if the set of actions $\poActionSpaceI$ is included in the set $\obsSpaceI$ and that for all $\poActionI\in\poActionSpaceI$ and $x\in \poStateSpace\cup\poActionSpaceOne\cup\poActionSpaceTwo$, $\obsFunI(x) = \poActionI$ if and only if $x=\poActionI$.

In $\poGame$, players can only rely on the observations they perceive to select actions.
A pure \textit{(observation-based) strategy} of $\poGame$ is a function $\stratI\colon \obsFunI(\histSet{\poGame})\to \poActionSpaceI$.
Randomised strategies can be defined as mixed strategies (i.e., distributions over pure observation-based strategies) or behavioural strategies.
Specifically, an \textit{observation-based behavioural strategy} is a function $\stratI\colon \obsFunI(\histSet{\poGame})\to \dist{\poActionSpaceI}$.
We will refer to (behavioural) strategies of the underlying game of perfect information $\game$ as \textit{history-based strategies} to distinguish them from observation-based ones.

In contrast to the perfect information setting, if we do not assume perfect recall, there need not be an equivalence between behavioural and mixed strategies.
Thankfully, randomised strategies (be they mixed or behavioural) of $\poGame$ are a subclass of history-based strategies.
This allows us to directly reuse notions previously defined for history-based strategies.
For instance, the probability distributions over plays of $\poGame$ induced by a pair of strategies from an initial state is the corresponding distribution in $\game$.
Furthermore, we avoid the need to consider mixed strategies explicitly this way.

We remark that the equivalent definition of outcome-equivalence for two strategies of $\playerOne$ formulated in Lemma~\ref{lemma:prelim:consistent_equivalence} also extends to the imperfect information setting.
On the one hand, $\playerTwo$ has access to fewer strategies, therefore the condition given in the lemma implies outcome-equivalence (the proof establishes a stronger statement).
On the other hand, the other direction requires strategies of $\playerTwo$ that are consistent with the histories considered in the proof; it suffices to consider a strategy of $\playerTwo$ that selects all available actions at random at all times for the argument to work.

\subsection{Finite-memory strategies}\label{section:imperfect:information:fm}

A strategy is \textit{finite-memory} if it is induced by a (stochastic) Mealy machine that reads observations instead of states and actions.
Formally, we define an \textit{observation-based Mealy machine} of $\playerI$ as a tuple $\mealy= \mealyTuple$ where $\mealyStateSpace$ is a finite set of memory states, $\mealyDistInit$ is an initial distribution over $\mealyStateSpace$, $\update\colon \mealyStateSpace\times \obsSpaceI^3 \to \dist{\mealyStateSpace}$ is the update function and $\nextmove\colon \mealyStateSpace\times \obsSpaceI\to \dist{\poActionSpaceI}$ is the next-move function.

An observation-based Mealy machine is a special case of a Mealy machine whose updates and outputs must coincide given inputs with the same observations.
We can thus derive a history-based strategy from an observation-based Mealy machine in the same way as in the perfect information setting.

To transfer our results on games of perfect information to games of imperfect information, we reuse the same classification of Mealy machines with three-letter acronyms for observation-based Mealy machines.
As was the case in the earlier sections, we will abusively say, e.g., $\mealy$ is an RRR observation-based strategy to mean that $\mealy$ is an observation-based Mealy machine with stochastic initialisation, outputs and updates, and avoid referring to the observation-based strategy it induces in this way.

In general, an observation-based Mealy machine may not induce a behavioural strategy of $\poGame$.
This can be illustrated with a simple RDD strategy.

\begin{example}\label{example:imperfect:mealy:not behavioural}
  We build a one-player game of imperfect information $\poGame_{a, b}$ from the game $\simpleGame$ of Figure~\ref{figure:game:mixed:behavioural}.
  We assign to $\concurState$, $a$ and $b$ a shared observation $o$.
  We consider the Mealy machine depicted in Figure~\ref{figure:DDD:strict:RDD}; note that its updates only depend on the memory state and not on the input and outputs.
  The strategy it induces, which we will denote by $\stratOne$, has a uniform probability of only playing $a$ or only playing $b$.

  No observation-based behavioural strategy is outcome-equivalent to $\stratOne$.
  Let $\stratAltOne\colon \{o\}(\{o\}^2)^*\to\dist{\{a, b\}}$ be a behavioural strategy.
  For it to be outcome-equivalent to $\stratOne$, $\stratAltOne$ has differentiate between the histories $sas$ and $sbs$ and play action $a$ and $b$ respectively following these histories.
  However, because both strategies share the same sequence of observations, $\stratAltOne$ cannot be outcome-equivalent to $\stratOne$. \hfill $\lhd$
\end{example}

We provide two sufficient conditions that ensure that observation-based Mealy machine induce a behavioural strategy.
The first one we present introduces a restriction on the games.
The second one involves no assumptions on games, but instead considers a restricted class of Mealy machines.

First, we show that all finite-memory strategies are behavioural in games with perfect recall.
Intuitively, the distribution over memory states depends heavily on the sequence of actions used by the considered player; the choice of actions conditions the distribution over memory states at each time it is updated.
The visibility of actions makes it so the distribution over memory states depends only on the observations fed to the Mealy machine.

\begin{lemma}\label{lemma:imperfect:fm:uniformity}
  Let $\mealy = \mealyTuple$ be an observation-based Mealy machine of $\playerI$.
  Assume that $\playerI$ has perfect recall in $\poGame$.
  Then the strategy induced by $\mealy$ is an observation-based behavioural strategy.
\end{lemma}

\begin{proof}
  Let $\mealyDist{\prefHist}$ denote the distribution over memory states of $\mealy$ after $\prefHist$ has taken place, for $\prefHist\in (\poStateSpace\poActionSpace)^*$.
  By definition of the strategy induced by a Mealy machine, it suffices to show the following: for all $\prefHist, v\in (\poStateSpace\poActionSpace)^*$ that are mapped to the same sequence of observations, we have $\mealyDist{\prefHist} = \mealyDist{v}$.
  
  Let $\prefHist, v\in (\poStateSpace\poActionSpace)^*$ such that $\prefHist$ and $v$ are mapped to the same sequence of observations.  
  We proceed by induction on the length of the considered sequence $\prefHist$ (which matches that of $v$).
  At the start of a play, an initial memory state is drawn following $\mealyDistInit$.
  Hence the distribution over memory states after the empty word $\varepsilon$ is $\mealyDist{\varepsilon} = \mealyDistInit$.
  In this case, there is nothing to show for the uniformity argument.

  We now assume the following by induction: for $\prefHist = \poState_0\poAction_0\ldots \poState_k\poAction_k$, the distribution $\mealyDist{\prefHist}$ over $\mealyStateSpace$ is well-defined and coincides with $\mealyDist{v}$ for $v =t_0\poActionAlt_0\ldots t_k\poActionAlt_k$ that can be mapped to the same sequence of observations as $\prefHist$.
  We consider $\prefHist' = \prefHist \poState_{k+1} \poAction_{k+1}$ and $v' = v t_{k+1}\poActionAlt_{k+1}$ that share the same sequence of observations.
  We describe $\mealyDist{\prefHist'}$, then infer that $\mealyDist{\prefHist'} = \mealyDist{v'}$.

  Due to the visibility of actions, we have $\poActionI_{k+1} = \obsFunI(\poActionI_{k+1}) = \poActionIAlt_{k+1}$.
  We distinguish two cases: $\mealyDist{\prefHist'}$ is well-defined or it is not.
  First, if for all $\mealyState\in\supp{\mealyDist{\prefHist}}$, we have $\nextmove(\mealyState, \obsFunI(s_{k+1}))(\poActionI_{k+1})=0$, then $\mealyDist{\prefHist'}$ and $\mealyDist{v'}$ are both undefined (i.e., $\prefHist'$ and $v'$ are inconsistent with $\mealy$).
  Therefore, we assume that there is $\mealyState\in\supp{\mealyDist{\prefHist}}$ such that $\nextmove(\mealyState, \obsFunI(\poState_{k+1}))(\poActionI_{k+1})>0$.
  In this case, we have, for any memory state $\mealyState\in\mealyStateSpace$,
  \[\mealyDist{\prefHist'}(\mealyState) =
    \frac{
      \sum_{\mealyState'\in\mealyStateSpace} \mealyDist{\prefHist}(\mealyState')
      \cdot
      \update(\mealyState', \obsFunI(\poState_{k+1}), \obsFunI(\poAction_{k+1}))(\mealyState)
      \cdot
      \nextmove(\mealyState', \obsFunI(\poState_{k+1}))(\poActionI_{k+1})}{\sum_{\mealyState'\in\mealyStateSpace} \mealyDist{\prefHist}(\mealyState')\cdot \nextmove(\mealyState, \obsFunI(\poState_{k+1}))(\poActionI_{k+1})}.\]
  The equation for $\mealyDist{v'}$ is the same as above, except $\poState_{k+1}$ and $\poActionIAdv_{k+1}$ are respectively replaced with $t_{k+1}$ and $\poActionIAdvAlt_{k+1}$.
  It follows immediately from $\obsFunI(\poState_{k+1}) = \obsFunI(t_{k+1})$ and $\obsFunI(\poActionIAdv_{k+1}) = \obsFunI(\poActionIAdvAlt_{k+1})$ that $\mealyDist{\prefHist'} = \mealyDist{v'}$.
  This ends the inductive argument and the proof.
\end{proof}

We have seen through Example~\ref{example:imperfect:mealy:not behavioural} that when lifting the perfect recall hypothesis, Mealy machines with randomised initialisation need not induce behavioural strategies.
A similar claim can be shown for strategies with randomised updates, e.g., by adapting the RDD example so the randomised initialisation is emulated by a stochastic memory update after the first round of the game.
On the other hand, DRD strategies always induce behavioural strategies.

\begin{lemma}\label{lemma:imperfect:fm:drd}
  Let $\mealy = \mealyTupleInSt$ be a DRD strategy of $\playerI$ in $\poGame$.
  Then the strategy induced by $\mealy$ is a behavioural strategy.
\end{lemma}
\begin{proof}
  For a DRD strategy, the distribution over memory states at any point is a Dirac distribution.
  More precisely, the memory state $\mealyState_\prefHist$ reached after $\prefHist\in(\poStateSpace\poActionSpace)^*$ is defined by induction.
  We have $\mealyState_\emptyword = \mealyStateInit$ and for $\prefHist\poState\poAction\in(\poStateSpace\poActionSpace)^+$, we have $\mealyState_{\prefHist\poState\poAction} = \update(\mealyState_\prefHist, \obsFunI(\poState), \obsFunI(\poAction))$.
  It is easy to see that $\mealyState_\prefHist$ depends only on the observations assigned to $\prefHist$, which is sufficient to end the proof.
\end{proof}

\subsection{Transferring our taxonomy to imperfect information}\label{section:imperfect:information:lattice}
We are now concerned with transferring our taxonomy of finite-memory strategies in games of perfect information to games of imperfect information.
We remark that all non-inclusions witnessed in the perfect information case hold in the imperfect information case.

On the one hand, if a player cannot perceive their own actions, some inclusions of the lattice in Figure~\ref{figure:lattice} fail.
This is already suggested by Example~\ref{example:imperfect:mealy:not behavioural} and Lemma~\ref{lemma:imperfect:fm:drd}, which imply together that RDD $\subseteq$ DRD does not hold without perfect recall.
On the other hand, it can be argued that the lattice of Figure~\ref{figure:lattice} stays unchanged in games where a player can see their actions.

\smallskip\noindent\textbf{Imperfect recall.}
We illustrate the lattice for general games of imperfect information in Figure~\ref{figure:lattice:imperfect:information}.
We first discuss the non-trivial inclusion that is preserved in this broader setting, then we explain why the others fail.

\begin{figure*}[tbh]
  \begin{center}
    \begin{tikzpicture}
      \matrix (a) [matrix of nodes, align=center, text width=4cm, column sep=0.1cm, row sep=0.7cm]{
        & {RRR = RDR}  & \\
        DRR &   & RRD \\
        DDR & DRD  & RDD  \\
        & DDD & \\
      };
      \draw (a-1-2) -- (a-2-1) node[midway, above left, align=center] {Lem.~\ref{lemma:imperfect:information:rdd_notin_drr}\\ (strictness)};
      \draw (a-1-2) -- (a-2-3);
      \draw (a-2-3) -- (a-3-3);
      \draw (a-2-1) -- (a-3-1);
      \draw (a-2-1) -- (a-3-2);
      \draw (a-2-3) -- (a-3-2);
      \draw (a-3-1) -- (a-4-2);
      \draw (a-3-2) -- (a-4-2);
      \draw (a-3-3) -- (a-4-2);
    \end{tikzpicture}
    \caption{Lattice of finite-memory strategy classes in games of imperfect information with imperfect recall. We decorate edges with relevant results introduced in this section.}\label{figure:lattice:imperfect:information}
\end{center}
\end{figure*}

Of the three non-trivial inclusions shown in Section~\ref{section:inclusions}, only RRR $\subseteq$ RDR still holds in this setting.
The idea is that Theorem~\ref{theorem:RRR:RDR}, unlike the other two inclusion theorems, does not rely on the visibility of actions in the construction of the Mealy machine.
It even provides a Mealy machine that is agnostic to a player's own actions.
Because states and the actions of the other player only serve the role of inputs (i.e., their nature does not matter), we can adapt the proof of the theorem directly to obtain its direct reformulation in games of imperfect information.

The other two non-trivial inclusions, RDD $\subseteq$ RDR and RRR $\subseteq$ DRR fail in this setting.
As explained previously, Example~\ref{example:imperfect:mealy:not behavioural} and Lemma~\ref{lemma:imperfect:fm:drd} show that the first inclusion cannot hold.
Furthermore, we obtain that DRR and RDD (and RRD) are incomparable.
To illustrate this, we prove that the strategy introduced in Example~\ref{example:imperfect:mealy:not behavioural} has no DRR equivalent.

\begin{lemma}\label{lemma:imperfect:information:rdd_notin_drr}
  Let $\poGame_{a, b}$ denote the game of imperfect information derived from $\simpleGame$ (Figure~\ref{figure:game:mixed:behavioural}) by assigning observation $o$ to everything.
  There exists an RDD strategy in $\poGame_{a, b}$ such that there is no outcome-equivalent DRR strategy.
\end{lemma}
\begin{proof}
  We consider the Mealy machine depicted in Figure~\ref{figure:DDD:strict:RDD} as in Example~\ref{example:imperfect:mealy:not behavioural} and let $\stratOne$ denote the history-based strategy it induces.
  Let $\mealy = \mealyTupleInSt$ be a DRR strategy of $\poGame_{a, b}$ and let $\stratAltOne^\mealy$ be the history-based strategy it induces.
  We assume towards a contradiction that $\stratAltOne^\mealy$ and $\stratOne$ are outcome-equivalent.

  We have $\nextmove(\mealyStateInit, o)(a) = \stratAltOne^\mealy(s)(a) = \stratOne(s)(a) = \frac{1}{2}$.
  It follows that the distributions $\mealyDist{sa}$ and $\mealyDist{sb}$ over $\mealyStateSpace$ after $sa$ and $sb$ have respectively occurred are, by definition, for all $\mealyState\in\mealyStateSpace$,
  \[\mealyDist{sa}(\mealyState) =\frac{\update(\mealyStateInit, o, o)(\mealyState)\cdot\frac{1}{2}}{\frac{1}{2}} = \mealyDist{sb}(\mealyState).\]
  We conclude that $\stratAltOne^\mealy(sas) = \stratAltOne^\mealy(sbs)$.
  However, the outcome-equivalence of $\stratOne$ and $\stratAltOne^\mealy$ implies that $\stratAltOne^\mealy(sas)(a) = \stratOne(sas)(a) = 1$ and $\stratAltOne^\mealy(sbs)(b) = \stratOne(sbs)(b) = 1$, which constitutes a contradiction.
\end{proof}

\smallskip\noindent\textbf{Perfect recall.}
We now consider games where the player we study can see their own actions.
In this case, we have the following theorem.

\begin{theorem}\label{theorem:partial:lattice}
  The taxonomy of Figure~\ref{figure:lattice} for $\playerI$ established in games of perfect information extends to games with imperfect information with perfect recall.
\end{theorem}
\begin{proof}
We have previously explained that Theorem~\ref{theorem:RRR:RDR} holds even without perfect recall.
Therefore, we need only generalise the statements of Theorems~\ref{theorem:mixed:behavioural} and~\ref{theorem:RRR:DRR} to games with imperfect information and perfect recall.
As we did with Theorem~\ref{theorem:RRR:RDR}, we briefly argument how to adapt their proofs when replacing states and actions with observations in a setting with perfect recall.

In Theorem~\ref{theorem:mixed:behavioural}, we simulate RDD strategies by means
of DRD strategies.
We keep track of a finite set of pure FM strategies and remove one whenever we perceive an action that is inconsistent with it.
The visibility of actions makes this approach viable in games of imperfect information.
Furthermore, the RDD strategy that is simulated and all of the pure FM strategies encoded in the simulating DRD strategy all use the exact same observation-based update scheme.
Therefore, any RDD strategy has an outcome-equivalent DRD counterpart in games of imperfect information with perfect recall.

Theorem~\ref{theorem:RRR:DRR} claims that any RRR strategy admits some outcome-equivalent
DRR strategy. The approach consists in adding a new initial memory state, and then
leverage stochastic updates to enter the supplied RRR strategy from the second step
of the game and proceed as though we had been using it from the start.
We designed the updates from the new initial memory state so that, from
the second step in the game, the distribution over memory states was the same in the RRR
strategy and the constructed DRR one. More precisely, the update probability
distribution from the new initial state is defined as the probability over the
memory states of the RRR strategy after one step.
The main argument of the proof of Lemma~\ref{lemma:imperfect:fm:uniformity} ensures that this distribution is robust to the passage to imperfect information and justifies that the proof approach generalises to this setting.
\end{proof}

\section{Conclusion}\label{section:conclusion}
We have provided a complete classification of randomised finite-memory strategies based on the notion of outcome-equivalence in concurrent games of perfect and imperfect information.
We have shown that all inclusions of the studied strategy classes can be witnessed by effective constructions.
Regarding the separation of strategy classes, we have provided examples on the simplest possible game and, additionally, illustrated the separation of classes on games that use specifications from the literature.

Outcome-equivalence is a specification-agnostic means of comparing strategies; two strategies that are outcome-equivalent provide the same performance against any specification no matter the strategy of the other player (or players in a multiplayer setting).
In particular, the established inclusions are universal in a sense, as they hold no matter the means of comparing the behaviour of strategies.
Nonetheless, outcome-equivalence is a very strong criterion for the comparison of strategies.
Given some specification and a strategy in a class, even if there is no outcome-equivalent strategy in another class, there may be a strategy of the second class that performs just as well, or even better with respect to the specification.
This suggests further work, where, given a family of games or specifications, we use some alternative means of comparing strategies and attempt to provide a similar taxonomy in this setting, or to attempt to understand the simplest strategies required to satisfy relevant families of specifications.

\bibliographystyle{elsarticle-num}
\bibliography{references.bib}{}

\begin{thebibliography}{10}
\expandafter\ifx\csname url\endcsname\relax
  \def\url#1{\texttt{#1}}\fi
\expandafter\ifx\csname urlprefix\endcsname\relax\def\urlprefix{URL }\fi
\expandafter\ifx\csname href\endcsname\relax
  \def\href#1#2{#2} \def\path#1{#1}\fi

\bibitem{EM79}
A.~Ehrenfeucht, J.~Mycielski, Positional strategies for mean payoff games,
  International Journal of Game Theory 8~(2) (1979) 109--113.

\bibitem{Con92}
A.~Condon, The complexity of stochastic games, Information and Computation
  96~(2) (1992) 203--224.
\newblock \href {https://doi.org/10.1016/0890-5401(92)90048-K}
  {\path{doi:10.1016/0890-5401(92)90048-K}}.

\bibitem{GZ05}
H.~Gimbert, W.~Zielonka, Games where you can play optimally without any memory,
  in: M.~Abadi, L.~{de Alfaro} (Eds.), Proceedings of the 16th International
  Conference on Concurrency Theory, {CONCUR} 2005, San Francisco, {CA}, {USA},
  August 23--26, 2005, Vol. 3653 of Lecture Notes in Computer Science,
  Springer, 2005, pp. 428--442.
\newblock \href {https://doi.org/10.1007/11539452_33}
  {\path{doi:10.1007/11539452_33}}.

\bibitem{DBLP:conf/dagstuhl/2001automata}
E.~Gr{\"{a}}del, W.~Thomas, T.~Wilke (Eds.), Automata, Logics, and Infinite
  Games: {A} Guide to Current Research [outcome of a Dagstuhl seminar, February
  2001], Vol. 2500 of Lecture Notes in Computer Science, Springer, 2002.
\newblock \href {https://doi.org/10.1007/3-540-36387-4}
  {\path{doi:10.1007/3-540-36387-4}}.

\bibitem{rECCS}
M.~Randour, Automated synthesis of reliable and efficient systems through game
  theory: A case study, in: Proc. of ECCS 2012, Springer Proceedings in
  Complexity XVII, Springer, 2013, pp. 731--738.
\newblock \href {https://doi.org/10.1007/978-3-319-00395-5\_90}
  {\path{doi:10.1007/978-3-319-00395-5\_90}}.

\bibitem{DBLP:conf/lata/BrenguierCHPRRS16}
R.~Brenguier, L.~Clemente, P.~Hunter, G.~A. P{\'{e}}rez, M.~Randour, J.~Raskin,
  O.~Sankur, M.~Sassolas, Non-zero sum games for reactive synthesis, in:
  A.~Dediu, J.~Janousek, C.~Mart{\'{\i}}n{-}Vide, B.~Truthe (Eds.), Proceedings
  of the 10th International Conference on Language and Automata Theory and
  Applications, {LATA} 2016, Prague, Czech Republic, March 14--18, 2016, Vol.
  9618 of Lecture Notes in Computer Science, Springer, 2016, pp. 3--23.
\newblock \href {https://doi.org/10.1007/978-3-319-30000-9\_1}
  {\path{doi:10.1007/978-3-319-30000-9\_1}}.

\bibitem{DBLP:reference/mc/BloemCJ18}
R.~Bloem, K.~Chatterjee, B.~Jobstmann, Graph games and reactive synthesis, in:
  E.~M. Clarke, T.~A. Henzinger, H.~Veith, R.~Bloem (Eds.), Handbook of Model
  Checking, Springer, 2018, pp. 921--962.
\newblock \href {https://doi.org/10.1007/978-3-319-10575-8\_27}
  {\path{doi:10.1007/978-3-319-10575-8\_27}}.

\bibitem{DBLP:conf/focs/EmersonJ88}
E.~A. Emerson, C.~S. Jutla, The complexity of tree automata and logics of
  programs (extended abstract), in: Proceedings of the 29th Annual Symposium on
  Foundations of Computer Science, {FOCS} 1988, White Plains, New York, {USA},
  October 24--26, 1988, {IEEE} Computer Society, 1988, pp. 328--337.
\newblock \href {https://doi.org/10.1109/SFCS.1988.21949}
  {\path{doi:10.1109/SFCS.1988.21949}}.

\bibitem{DBLP:journals/tcs/Zielonka98}
W.~Zielonka, Infinite games on finitely coloured graphs with applications to
  automata on infinite trees, Theoretical Computer Science 200~(1-2) (1998)
  135--183.

\bibitem{DBLP:journals/corr/BruyereHR16}
V.~Bruy{\`{e}}re, Q.~Hautem, M.~Randour, Window parity games: an alternative
  approach toward parity games with time bounds, in: D.~Cantone, G.~Delzanno
  (Eds.), Proceedings of the 7th International Symposium on Games, Automata,
  Logics, and Formal Verification, {GandALF} 2016, Catania, Italy, September
  14--16, 2016, Vol. 226 of {EPTCS}, 2016, pp. 135--148.
\newblock \href {https://doi.org/10.4204/EPTCS.226.10}
  {\path{doi:10.4204/EPTCS.226.10}}.

\bibitem{DBLP:conf/concur/BrihayeDOR19}
T.~Brihaye, F.~Delgrange, Y.~Oualhadj, M.~Randour, Life is random, time is not:
  {M}arkov decision processes with window objectives, in: Fokkink and van
  Glabbeek  \cite{DBLP:conf/concur/2019}, pp. 8:1--8:18.
\newblock \href {https://doi.org/10.4230/LIPIcs.CONCUR.2019.8}
  {\path{doi:10.4230/LIPIcs.CONCUR.2019.8}}.

\bibitem{DBLP:journals/acta/BouyerMRLL18}
P.~Bouyer, N.~Markey, M.~Randour, K.~G. Larsen, S.~Laursen, Average-energy
  games, Acta Informatica 55~(2) (2018) 91--127.
\newblock \href {https://doi.org/10.1007/s00236-016-0274-1}
  {\path{doi:10.1007/s00236-016-0274-1}}.

\bibitem{DBLP:conf/concur/BruyereHRR19}
V.~Bruy{\`{e}}re, Q.~Hautem, M.~Randour, J.~Raskin, Energy mean-payoff games,
  in: Fokkink and van Glabbeek  \cite{DBLP:conf/concur/2019}, pp. 21:1--21:17.
\newblock \href {https://doi.org/10.4230/LIPIcs.CONCUR.2019.21}
  {\path{doi:10.4230/LIPIcs.CONCUR.2019.21}}.

\bibitem{DBLP:journals/acta/ChatterjeeRR14}
K.~Chatterjee, M.~Randour, J.~Raskin, Strategy synthesis for multi-dimensional
  quantitative objectives, Acta Informatica 51~(3-4) (2014) 129--163.
\newblock \href {https://doi.org/10.1007/s00236-013-0182-6}
  {\path{doi:10.1007/s00236-013-0182-6}}.

\bibitem{DBLP:journals/iandc/VelnerC0HRR15}
Y.~Velner, K.~Chatterjee, L.~Doyen, T.~A. Henzinger, A.~M. Rabinovich,
  J.~Raskin, The complexity of multi-mean-payoff and multi-energy games,
  Information and Computation 241 (2015) 177--196.
\newblock \href {https://doi.org/10.1016/j.ic.2015.03.001}
  {\path{doi:10.1016/j.ic.2015.03.001}}.

\bibitem{DBLP:journals/fmsd/RandourRS17}
M.~Randour, J.~Raskin, O.~Sankur, Percentile queries in multi-dimensional
  {Markov} decision processes, Formal methods in system design 50~(2-3) (2017)
  207--248.
\newblock \href {https://doi.org/10.1007/s10703-016-0262-7}
  {\path{doi:10.1007/s10703-016-0262-7}}.

\bibitem{DBLP:conf/tacas/DelgrangeKQR20}
F.~Delgrange, J.~Katoen, T.~Quatmann, M.~Randour, Simple strategies in
  multi-objective {MDPs}, in: A.~Biere, D.~Parker (Eds.), Proceedings (Part
  {I}) of the 26th International Conference on Tools and Algorithms for the
  Construction and Analysis of Systems, {TACAS} 2020, Held as Part of {ETAPS}
  2020, Dublin, Ireland, April 25--30, 2020, Vol. 12078 of Lecture Notes in
  Computer Science, Springer, 2020, pp. 346--364.
\newblock \href {https://doi.org/10.1007/978-3-030-45190-5\_19}
  {\path{doi:10.1007/978-3-030-45190-5\_19}}.

\bibitem{DBLP:conf/lics/ChatterjeeD12}
K.~Chatterjee, L.~Doyen, Partial-observation stochastic games: How to win when
  belief fails, in: Proceedings of the 27th Annual {IEEE} Symposium on Logic in
  Computer Science, {LICS} 2012, Dubrovnik, Croatia, June 25--28, 2012, {IEEE}
  Computer Society, 2012, pp. 175--184.
\newblock \href {https://doi.org/10.1109/LICS.2012.28}
  {\path{doi:10.1109/LICS.2012.28}}.

\bibitem{DBLP:conf/icalp/BerthonRR17}
R.~Berthon, M.~Randour, J.~Raskin, Threshold constraints with guarantees for
  parity objectives in {M}arkov decision processes, in: I.~Chatzigiannakis,
  P.~Indyk, F.~Kuhn, A.~Muscholl (Eds.), Proceedings of the 44th International
  Colloquium on Automata, Languages, and Programming, {ICALP} 2017, Warsaw,
  Poland, July 10--14, 2017, Vol.~80 of LIPIcs, Schloss Dagstuhl -
  Leibniz-Zentrum f{\"{u}}r Informatik, 2017, pp. 121:1--121:15.
\newblock \href {https://doi.org/10.4230/LIPIcs.ICALP.2017.121}
  {\path{doi:10.4230/LIPIcs.ICALP.2017.121}}.

\bibitem{DBLP:journals/corr/abs-1006-1404}
J.~Cristau, C.~David, F.~Horn, How do we remember the past in randomised
  strategies?, in: A.~Montanari, M.~Napoli, M.~Parente (Eds.), Proceedings of
  the First Symposium on Games, Automata, Logic, and Formal Verification,
  {GANDALF} 2010, Minori (Amalfi Coast), Italy, June 17--18, 2010, Vol.~25 of
  {EPTCS}, 2010, pp. 30--39.
\newblock \href {https://doi.org/10.4204/EPTCS.25.7}
  {\path{doi:10.4204/EPTCS.25.7}}.

\bibitem{OR94}
M.~J. Osborne, A.~Rubinstein, A course in game theory, The MIT Press, 1994.

\bibitem{Aumann64}
R.~J.~. Aumann, Mixed and behavior strategies in infinite extensive games, in:
  M.~Dresher, L.~S. Shapley, A.~W. Tucker (Eds.), Advances in Game Theory.
  (AM-52), Volume 52, Princeton University Press, 1964, pp. 627--650.
\newblock \href {https://doi.org/10.1515/9781400882014-029}
  {\path{doi:10.1515/9781400882014-029}}.

\bibitem{DBLP:journals/iandc/BruyereFRR17}
V.~Bruy{\`{e}}re, E.~Filiot, M.~Randour, J.~Raskin, Meet your expectations with
  guarantees: Beyond worst-case synthesis in quantitative games, Information
  and Computation 254 (2017) 259--295.
\newblock \href {https://doi.org/10.1016/j.ic.2016.10.011}
  {\path{doi:10.1016/j.ic.2016.10.011}}.

\bibitem{DBLP:journals/lmcs/BouyerRORV22}
P.~Bouyer, S.~{Le Roux}, Y.~Oualhadj, M.~Randour, P.~Vandenhove, Games where
  you can play optimally with arena-independent finite memory, Logical Methods
  in Computer Science 18~(1) (2022).
\newblock \href {https://doi.org/10.46298/lmcs-18(1:11)2022}
  {\path{doi:10.46298/lmcs-18(1:11)2022}}.

\bibitem{DBLP:conf/concur/BouyerORV21}
P.~Bouyer, Y.~Oualhadj, M.~Randour, P.~Vandenhove, Arena-independent
  finite-memory determinacy in stochastic games, in: S.~Haddad, D.~Varacca
  (Eds.), Proceedings of the 32nd International Conference on Concurrency
  Theory, {CONCUR} 2021, Virtual Conference, August 24--27, 2021, Vol. 203 of
  LIPIcs, Schloss Dagstuhl -- Leibniz-Zentrum f{\"{u}}r Informatik, 2021, pp.
  26:1--26:18.
\newblock \href {https://doi.org/10.4230/LIPIcs.CONCUR.2021.26}
  {\path{doi:10.4230/LIPIcs.CONCUR.2021.26}}.

\bibitem{DBLP:journals/tcs/AlfaroHK07}
L.~de~Alfaro, T.~A. Henzinger, O.~Kupferman, Concurrent reachability games,
  Theoretical Computer Science 386~(3) (2007) 188--217.
\newblock \href {https://doi.org/10.1016/j.tcs.2007.07.008}
  {\path{doi:10.1016/j.tcs.2007.07.008}}.

\bibitem{Sha53}
L.~S. Shapley, \href{https://www.pnas.org/content/39/10/1095}{Stochastic
  games}, Proceedings of the National Academy of Sciences 39~(10) (1953)
  1095--1100.
\newblock \href
  {http://arxiv.org/abs/https://www.pnas.org/content/39/10/1095.full.pdf}
  {\path{arXiv:https://www.pnas.org/content/39/10/1095.full.pdf}}, \href
  {https://doi.org/10.1073/pnas.39.10.1095}
  {\path{doi:10.1073/pnas.39.10.1095}}.
\newline\urlprefix\url{https://www.pnas.org/content/39/10/1095}

\bibitem{MaitraSudderth}
A.~Maitra, W.~Sudderth, Stochastic games with {Borel} payoffs, in: A.~Neyman,
  S.~Sorin (Eds.), Stochastic Games and Applications, Springer Netherlands,
  Dordrecht, 2003, pp. 367--373.

\bibitem{DBLP:journals/iandc/Chatterjee0GH15}
K.~Chatterjee, L.~Doyen, H.~Gimbert, T.~A. Henzinger, Randomness for free,
  Information and Computation 245 (2015) 3--16.
\newblock \href {https://doi.org/10.1016/j.ic.2015.06.003}
  {\path{doi:10.1016/j.ic.2015.06.003}}.

\bibitem{DBLP:conf/fsttcs/0001PR18}
S.~{Le Roux}, A.~Pauly, M.~Randour, Extending finite-memory determinacy by
  {B}oolean combination of winning conditions, in: S.~Ganguly, P.~K. Pandya
  (Eds.), Proceedings of the 38th {IARCS} Annual Conference on Foundations of
  Software Technology and Theoretical Computer Science, {FSTTCS} 2018,
  Ahmedabad, India, December 11--13, 2018, Vol. 122 of LIPIcs, Schloss Dagstuhl
  - Leibniz-Zentrum fuer Informatik, 2018, pp. 38:1--38:20.
\newblock \href {https://doi.org/10.4230/LIPIcs.FSTTCS.2018.38}
  {\path{doi:10.4230/LIPIcs.FSTTCS.2018.38}}.

\bibitem{DBLP:conf/stacs/BouyerRV22}
P.~Bouyer, M.~Randour, P.~Vandenhove, Characterizing omega-regularity through
  finite-memory determinacy of games on infinite graphs, in: P.~Berenbrink,
  B.~Monmege (Eds.), Proceedings of the 39th International Symposium on
  Theoretical Aspects of Computer Science, {STACS} 2022, Marseille, France,
  March 15--18, 2022, Vol. 219 of LIPIcs, Schloss Dagstuhl - Leibniz-Zentrum
  f{\"{u}}r Informatik, 2022, pp. 16:1--16:16.
\newblock \href {https://doi.org/10.4230/LIPIcs.STACS.2022.16}
  {\path{doi:10.4230/LIPIcs.STACS.2022.16}}.

\bibitem{KopThesis}
E.~Kopczy{\'n}ski, Half-positional determinacy of infinite games, Ph.D. thesis,
  Warsaw University (2008).

\bibitem{DBLP:journals/lmcs/0001HPW19}
J.~Gutierrez, P.~Harrenstein, G.~Perelli, M.~J. Wooldridge,
  \href{https://doi.org/10.23638/LMCS-15(3:32)2019}{Nash equilibrium and
  bisimulation invariance}, Logical Methods in Computer Science 15~(3) (2019).
\newblock \href {https://doi.org/10.23638/LMCS-15(3:32)2019}
  {\path{doi:10.23638/LMCS-15(3:32)2019}}.
\newline\urlprefix\url{https://doi.org/10.23638/LMCS-15(3:32)2019}

\bibitem{DBLP:conf/icalp/CasaresO23}
A.~Casares, P.~Ohlmann,
  \href{https://doi.org/10.4230/LIPIcs.ICALP.2023.122}{Characterising memory in
  infinite games}, in: K.~Etessami, U.~Feige, G.~Puppis (Eds.), Proceedings of
  the 50th International Colloquium on Automata, Languages, and Programming,
  {ICALP} 2023, July 10--14, 2023, Paderborn, Germany, Vol. 261 of LIPIcs,
  Schloss Dagstuhl - Leibniz-Zentrum f{\"{u}}r Informatik, 2023, pp.
  122:1--122:18.
\newblock \href {https://doi.org/10.4230/LIPICS.ICALP.2023.122}
  {\path{doi:10.4230/LIPICS.ICALP.2023.122}}.
\newline\urlprefix\url{https://doi.org/10.4230/LIPIcs.ICALP.2023.122}

\bibitem{DBLP:journals/jacm/BertrandGG17}
N.~Bertrand, B.~Genest, H.~Gimbert, Qualitative determinacy and decidability of
  stochastic games with signals, Journal of the ACM 64~(5) (2017) 33:1--33:48.
\newblock \href {https://doi.org/10.1145/3107926} {\path{doi:10.1145/3107926}}.

\bibitem{DBLP:conf/qest/ChatterjeeAH04}
K.~Chatterjee, L.~{de Alfaro}, T.~A. Henzinger, Trading memory for randomness,
  in: Proceedings of the 1st International Conference on Quantitative
  Evaluation of Systems, {QEST} 2004, Enschede, The Netherlands, 27--30
  September 2004, {IEEE} Computer Society, 2004, pp. 206--217.
\newblock \href {https://doi.org/10.1109/QEST.2004.1348035}
  {\path{doi:10.1109/QEST.2004.1348035}}.

\bibitem{DBLP:conf/hybrid/ChatterjeeHP08}
K.~Chatterjee, T.~A. Henzinger, V.~S. Prabhu, Trading infinite memory for
  uniform randomness in timed games, in: M.~Egerstedt, B.~Mishra (Eds.),
  Proceedings of the 11th International Workshop on Hybrid Systems: Computation
  and Control, {HSCC} 2008, St. Louis, MO, USA, April 22--24, 2008, Vol. 4981
  of Lecture Notes in Computer Science, Springer, 2008, pp. 87--100.
\newblock \href {https://doi.org/10.1007/978-3-540-78929-1_7}
  {\path{doi:10.1007/978-3-540-78929-1_7}}.

\bibitem{DBLP:conf/stacs/Horn09}
F.~Horn, Random fruits on the {Zielonka} tree, in: S.~Albers, J.~Marion (Eds.),
  Proceedings of the 26th International Symposium on Theoretical Aspects of
  Computer Science, {STACS} 2009, Freiburg, Germany, February 26--28, 2009,
  Vol.~3 of LIPIcs, Schloss Dagstuhl --Leibniz-Zentrum f{\"{u}}r Informatik,
  Germany, 2009, pp. 541--552.
\newblock \href {https://doi.org/10.4230/LIPIcs.STACS.2009.1848}
  {\path{doi:10.4230/LIPIcs.STACS.2009.1848}}.

\bibitem{DBLP:conf/concur/MonmegePR20}
B.~Monmege, J.~Parreaux, P.~Reynier, Reaching your goal optimally by playing at
  random with no memory, in: I.~Konnov, L.~Kov{\'{a}}cs (Eds.), Proceedings of
  the 31st International Conference on Concurrency Theory, {CONCUR} 2020,
  Vienna, Austria, September 1--4, 2020, Vol. 171 of LIPIcs, Schloss Dagstuhl
  --Leibniz-Zentrum f{\"{u}}r Informatik, 2020, pp. 26:1--26:21.
\newblock \href {https://doi.org/10.4230/LIPIcs.CONCUR.2020.26}
  {\path{doi:10.4230/LIPIcs.CONCUR.2020.26}}.

\bibitem{DBLP:conf/concur/MainR22}
J.~C.~A. Main, M.~Randour, Different strokes in randomised strategies:
  Revisiting {Kuhn's} theorem under finite-memory assumptions, in: B.~Klin,
  S.~Lasota, A.~Muscholl (Eds.), Proceedings of the 33rd International
  Conference on Concurrency Theory, {CONCUR} 2022, Warsaw, Poland, September
  12--16, 2022, Vol. 243 of LIPIcs, Schloss Dagstuhl --Leibniz-Zentrum
  f{\"{u}}r Informatik, 2022, pp. 22:1--22:18.
\newblock \href {https://doi.org/10.4230/LIPIcs.CONCUR.2022.22}
  {\path{doi:10.4230/LIPIcs.CONCUR.2022.22}}.

\bibitem{Dur19}
R.~Durrett, Probability: Theory and Examples, 5th Edition, Cambridge Series in
  Statistical and Probabilistic Mathematics, Cambridge University Press, 2019.
\newblock \href {https://doi.org/10.1017/9781108591034}
  {\path{doi:10.1017/9781108591034}}.

\bibitem{DBLP:journals/corr/abs-1104-3489}
T.~Br{\'{a}}zdil, V.~Brozek, K.~Chatterjee, V.~Forejt, A.~Kucera, {Markov}
  decision processes with multiple long-run average objectives, Logical Methods
  in Computer Science 10~(1) (2014).
\newblock \href {https://doi.org/10.2168/LMCS-10(1:13)2014}
  {\path{doi:10.2168/LMCS-10(1:13)2014}}.

\bibitem{DBLP:journals/lmcs/ChatterjeeKK17}
K.~Chatterjee, Z.~Kret{\'{\i}}nsk{\'{a}}, J.~Kret{\'{\i}}nsk{\'{y}}, Unifying
  two views on multiple mean-payoff objectives in {Markov} decision processes,
  Logical Methods in Computer Science 13~(2) (2017).
\newblock \href {https://doi.org/10.23638/LMCS-13(2:15)2017}
  {\path{doi:10.23638/LMCS-13(2:15)2017}}.

\bibitem{DBLP:journals/lmcs/EtessamiKVY08}
K.~Etessami, M.~Z. Kwiatkowska, M.~Y. Vardi, M.~Yannakakis, Multi-objective
  model checking of {Markov} decision processes, Logical Methods in Computer
  Science 4~(4) (2008).
\newblock \href {https://doi.org/10.2168/LMCS-4(4:8)2008}
  {\path{doi:10.2168/LMCS-4(4:8)2008}}.

\bibitem{gog23}
N.~Fijalkow, N.~Bertrand, P.~Bouyer{-}Decitre, R.~Brenguier, A.~Carayol,
  J.~Fearnley, H.~Gimbert, F.~Horn, R.~Ibsen{-}Jensen, N.~Markey, B.~Monmege,
  P.~Novotn{\'{y}}, M.~Randour, O.~Sankur, S.~Schmitz, O.~Serre, M.~Skomra,
  Games on graphs, CoRR abs/2305.10546 (2023).
\newblock \href {https://doi.org/10.48550/arXiv.2305.10546}
  {\path{doi:10.48550/arXiv.2305.10546}}.

\bibitem{DBLP:conf/mfcs/ChenFKSW13}
T.~Chen, V.~Forejt, M.~Z. Kwiatkowska, A.~Simaitis, C.~Wiltsche, On stochastic
  games with multiple objectives, in: K.~Chatterjee, J.~Sgall (Eds.),
  Proceedings of the 38th International Symposium on Mathematical Foundations
  of Computer Science, {MFCS} 2013, Klosterneuburg, Austria, August 26--30,
  2013, Vol. 8087 of Lecture Notes in Computer Science, Springer, 2013, pp.
  266--277.
\newblock \href {https://doi.org/10.1007/978-3-642-40313-2\_25}
  {\path{doi:10.1007/978-3-642-40313-2\_25}}.

\bibitem{TechRep:ChenFKSW13}
T.~Chen, V.~Forejt, M.~Kwiatkowska, A.~Simaitis, C.~Wiltsche, On stochastic
  games with multiple objectives, Tech. Rep. RR-13-06, University of Oxford,
  Department of Computer Science (2013).

\bibitem{DBLP:conf/lics/AshokCKWW20}
P.~Ashok, K.~Chatterjee, J.~Kret{\'{\i}}nsk{\'{y}}, M.~Weininger, T.~Winkler,
  Approximating values of generalized-reachability stochastic games, in:
  H.~Hermanns, L.~Zhang, N.~Kobayashi, D.~Miller (Eds.), Proceedings of the
  35th Annual {ACM/IEEE} Symposium on Logic in Computer Science, {LICS} 2020,
  Saarbr{\"{u}}cken, Germany, July 8--11, 2020, {ACM}, 2020, pp. 102--115.
\newblock \href {https://doi.org/10.1145/3373718.3394761}
  {\path{doi:10.1145/3373718.3394761}}.

\bibitem{DBLP:conf/concur/2019}
W.~J. Fokkink, R.~van Glabbeek (Eds.), Proceedings of the 30th International
  Conference on Concurrency Theory, {CONCUR} 2019, Amsterdam, the Netherlands,
  August 27--30, 2019, Vol. 140 of LIPIcs, Schloss Dagstuhl - Leibniz-Zentrum
  f{\"{u}}r Informatik, 2019.
\newblock \href {https://doi.org/10.4230/LIPIcs.CONCUR.2019}
  {\path{doi:10.4230/LIPIcs.CONCUR.2019}}.

\end{thebibliography}

\newpage
\appendix
\section{Probability over memory states in stochastic-update Mealy machines}\label{appendix:probability_memory}

\subsection{Inductive relation for the distribution over memory states}

We fix a game $\game = \concurTuple$.
In this section, we derive the formula for updates of the distribution over memory states of a Mealy machine after a word in $(\concurStateSpace\concurActionSpace)^*$ takes place under its induced strategy.
We build our reasoning on conditional probabilities.
We show the equations for a Mealy machine of $\playerOne$; the reasoning for $\playerTwo$ is analogous.
We fix a Mealy machine $\mealy =\mealyTuple$ of $\playerOne$.

Let $\prefHist= \concurState_0\concurAction_0\concurState_1\concurAction_1\ldots\concurState_k\concurAction_k\in(\concurStateSpace\concurActionSpace)^*$.
We study the distribution over $\mealyStateSpace$ after $\prefHist$ takes place.
This distribution is well-defined only under specific assumptions on $\prefHist$.
The probability of being in some memory state $\mealyState$ after $\prefHist$ is formally the conditional probability of being in $\mealyState$ at step $k+1$ given $\prefHist$.
We must therefore require that $\prefHist$ is of positive probability under $\mealy$ and (at least) one strategy of $\playerTwo$, i.e., $\prefHist$ must be consistent with $\mealy$.

We reuse the notation $\mealyDist{\prefHist}$ introduced in Section~\ref{section:preliminaries}.
The main goal of this section is to prove the inductive relation for $\mealyDist{\prefHist}$ recalled below.
Assume $\prefHist$ is not the empty word and let $\prefHist' = \concurState_0\concurAction_0\concurState_1\concurAction_1\ldots \concurState_{k-1}\concurAction_{k-1}$.
We prove the equation:
\begin{equation}\label{equation:update}
  \mealyDist{\prefHist}(\mealyState) = \frac{\sum_{\mealyState'\in\mealyStateSpace} \mealyDist{\prefHist'}(\mealyState')\cdot \update(\mealyState', \concurState_k, \concurAction_k)(m)\cdot \nextmove(\mealyState', \concurState_k)(\concurActionOne_k)}{\sum_{\mealyState'\in\mealyStateSpace} \mealyDist{\prefHist'}(\mealyState')\cdot \nextmove(\mealyState', \concurState_k)(\concurActionOne_k)}.
\end{equation}

We derive this equation by studying the Markov chain induced in $\game$ by $\mealy$ and a strategy of $\playerTwo$ from an initial state.
We fix a strategy $\stratTwo$ of $\playerTwo$ and an initial state $\concurState_\init\in\concurStateSpace$.
In the sequel, we prove that the equations above hold for any $\prefHist\in (\concurStateSpace\concurActionSpace)^*$ that starts in $\concurState_\init$ and is consistent with $\mealy$ and $\stratTwo$.
As indicated by Equation~\eqref{equation:update}, the choice of $\stratTwo$ has no impact on $\mealyDist{\prefHist}$ (this strategy is required so the Markov chain is well-defined).

\subsection{Description of the Markov chain}

First, we describe the Markov chain induced by playing $\mealy$ and $\stratTwo$ from $\concurState_\init$ in $\game$.
Formally, it is an infinite Markov chain where states are non-empty sequences $(\concurState_0, \mealyState_0, \concurAction_0)\ldots (\concurState_k, \mealyState_k, \concurAction_k)$ in $(\concurStateSpace\times\mealyStateSpace\times\concurActionSpace)^*$ where $\concurState_0\concurAction_0\ldots \concurAction_{k-1}\concurState_k$ is a history of $\game$ and $\concurAction_k\in\concurActionSpace(\concurState_k)$.
The initial probability of a state $(\concurState_\init, \mealyState, \concurAction)$ is given as the product
$\mu_\init(\mealyState)\cdot \nextmove(\mealyState, \concurState_\init)(\concurActionOne)\cdot \stratTwo(\concurState_\init)(\concurActionTwo);$
we multiply the probability that $\mealyState$ is drawn as the initial memory state, that $\concurActionOne$ is selected in memory state $\mealyState$ and that $\concurActionTwo$ is selected by $\stratTwo$.
The initial distribution assigns $0$ to any other state of the Markov chain.

Let $t = (\concurState_0, \mealyState_0, \concurAction_0)\ldots (\concurState_k, \mealyState_k, \concurAction_k)$ and $t' = t (\concurState_{k+1}, \mealyState_{k+1}, \concurAction_{k+1})$ be two states of the Markov chain.
The transition probability from $t$ to $t'$ is defined by the product
\begin{equation*}
  \concurTrans(\concurState_{k},\concurAction_{k})(\concurState_{k+1})\cdot
  \update(\mealyState_{k}, \concurState_{k},\concurAction_{k})(\mealyState_{k+1}) \cdot
  \nextmove(\mealyState_{k+1},\concurState_{k+1})(\concurActionOne_{k+1})\cdot
  \stratTwo(\concurState_0\concurAction_0\ldots\concurAction_{k}\concurState_{k+1})(\concurActionTwo_{k+1}).
\end{equation*}

We define a probability measure over infinite sequences of states of the Markov chain described above in the standard way, using cylinders.
Initial infinite sequences of this Markov chain belong in $((\concurStateSpace\times\mealyStateSpace\times\concurActionSpace)^*)^\omega$ and are of the form $t_0 (t_0t_1)(t_0t_1t_2)\ldots$  where $t_k\in \concurStateSpace\times\mealyStateSpace\times\concurActionSpace$.
We identify these infinite initial sequences to elements of $(\concurStateSpace\times\mealyStateSpace\times\concurActionSpace)^\omega$.
We will write $\proba$ for the probability distribution over $(\concurStateSpace\times\mealyStateSpace\times\concurActionSpace)^\omega$ obtained this way.

In the sequel, we use random variables defined over $(\concurStateSpace\times\mealyStateSpace\times\concurActionSpace)^\omega$ to refer to components or parts of these sequences and derive Equation~\eqref{equation:update}.
We introduce some notation.
Let $B$ denote a set. For any random variable $X\colon (\concurStateSpace\times\mealyStateSpace\times\concurActionSpace)^\omega\to B$ and $b\in B$, we write $\{X = b\}$ for $X^{-1}(\{b\})$ and omit the braces when evaluating $\proba$ over such sets, e.g., we write $\proba(X = b)$ for $\proba(\{X = b\})$.

We use the following random variables.
We denote by $\concurStateSpace_k$ (resp.~$\mealyStateSpace_k$, $\concurActionSpace_k = (\concurActionSpaceOne_k, \concurActionSpaceTwo_k)$) the random variable that describes the state of the game (resp.~memory state, pair of actions) at position $k$ of a sequence in $(\concurStateSpace\times\mealyStateSpace\times\concurActionSpace)^\omega$.
We write $W_k$  for the random variable describing the sequence $W_k = \concurStateSpace_0\concurActionSpace_0\concurStateSpace_1\concurActionSpace_1\ldots \concurStateSpace_{k-1}\concurActionSpace_{k-1}$ which is the sequence read by $\mealy$ prior to step $k$.
Similarly, we write $H_k$ (resp.~$\overline{\mealyStateSpace_k}$) for the random variable $H_k = W_k\concurStateSpace_k$ (resp.~$\overline{\mealyStateSpace_k} = \mealyStateSpace_0\mealyStateSpace_1\ldots \mealyStateSpace_{k}$) that describes the history at step $k$ (resp.~the sequence of memory states up to step $k$).

We now list properties of these random variables we refer to in the proof of Equation~\eqref{equation:update}.
We will be concerned with conditional probabilities, and therefore all upcoming equations will assume that some event has a positive probability.
We mainly rely on the properties listed below.

First, memory updates only depend on the latest memory state, game state and pair of actions.
Formally, let us take a non-empty sequence $\prefHist=\concurState_0\concurAction_0\ldots \concurState_{k-1}\concurAction_{k-1}\in (\concurStateSpace\concurActionSpace)^+$ such that $\proba(W_k = \prefHist)>0$.
For any sequence of memory states $\overline{\mealyState} = \mealyState_0\mealyState_1\ldots \mealyState_{k-1}\in \mealyStateSpace^{k}$ such that $\proba(\overline{\mealyStateSpace_{k-1}}=\overline{\mealyState}\mid W_k = \prefHist)>0$, we have, for every state $\mealyState\in\mealyStateSpace$,
\begin{align*}
  \proba(\mealyStateSpace_k = \mealyState \mid & W_k = w\land \overline{\mealyStateSpace_{k-1}} = \overline{\mealyState}) \\
  &=  \proba(\mealyStateSpace_k = \mealyState \mid \concurStateSpace_{k-1} = \concurState_{k-1}\land \mealyStateSpace_{k-1} = \mealyState_{k-1}\land \concurActionSpace_{k-1} = \concurAction_{k-1}) \\
  & = \update(\mealyState_{k-1}, \concurState_{k-1}, \concurAction_{k-1})(m).
\end{align*}

Next, memory updates are independent from game state updates.
In particular, for any history $\hist=\concurState_0\concurAction_0\ldots \concurAction_{k-1}\concurState_k\in\histSet{\game}$ such that $\proba(H_k = \hist)>0$, we have for any  $\mealyState\in\mealyStateSpace$,
\[\proba(\mealyStateSpace_k = \mealyState \mid H_k = h) = \proba(\mealyStateSpace_k = \mealyState \mid W_k = \prefHist),\]
where $\prefHist$ denotes $\concurState_0\concurAction_0\ldots \concurState_{k-1}\concurAction_{k-1}$.

The last three properties are related to the probability of actions following a history.
To formulate these properties, we fix $\hist=\concurState_0\concurAction_0\ldots \concurState_k\in\histSet{\game}$ such that $\proba(H_k = \hist)>0$ and a sequence of memory states $\overline{\mealyState} = \mealyState_0\mealyState_1\ldots \mealyState_k\in M^{k+1}$ that is compatible with $\hist$, i.e., such that $\proba(\overline{\mealyStateSpace_k}=\overline{\mealyState}\mid H_k = \hist)>0$.
First, we note that the actions choices of the players are independent given $\hist$ and $\overline{\mealyState}$, i.e., for all $\concurAction\in\concurActionSpace(\concurState_k)$, we have
\begin{align*}
  \proba(\concurActionSpace_k = \concurAction\mid
  & H_k = \hist\land \overline{\mealyStateSpace_k} = \overline{\mealyState}) \\ 
  =\, &  \proba(\concurActionSpaceOne_k = \concurActionOne\mid H_k = \hist\land \overline{\mealyStateSpace_k} = \overline{\mealyState})\cdot \proba(\concurActionSpaceTwo_k = \concurActionTwo\mid H_k = \hist\land \overline{\mealyStateSpace_k} = \overline{\mealyState}).
\end{align*}

Second, we remark that the next action of $\playerOne$ at any point depends only on the last state of the history and the last memory state.
In other words, for any sequence of memory states $\overline{\mealyState} = \mealyState_0\mealyState_1\ldots \mealyState_k\in \mealyStateSpace^{k+1}$ such that $\proba(\overline{\mealyStateSpace_k}=\overline{\mealyState}\mid H_k = \hist)>0$ (i.e., any sequence of memory states likely to occur by processing $\hist$)  and action $\concurActionOne\in \concurActionSpaceOne(\concurState_k)$,
\begin{align*}
  \proba(\concurActionSpaceOne_k = \concurActionOne\mid
  H_k = h\land \overline{\mealyStateSpace_k} = \overline{\mealyState}) & =
  \proba(\concurActionSpaceOne_k = \concurActionOne \mid
  \concurStateSpace_k = \concurState_k\land \mealyStateSpace_k = \mealyState_k) \\
  & = \nextmove(\mealyState_k, \concurState_k)(\concurActionOne).
\end{align*}
Finally, we remark that the probability of the next action of $\playerTwo$ is given by $\stratTwo(\hist)$ and is independent of the sequence of memory states $\overline{\mealyState}$.
Formally, we have,
\[\proba(\concurActionSpaceTwo_k = \concurActionTwo\mid
  H_k = \hist\land \overline{\mealyStateSpace_k} = \overline{\mealyState}) =
  \proba(\concurActionSpaceTwo_k =\concurActionTwo \mid H_k = \hist) =
  \stratTwo(\hist)(\concurActionTwo).
\]

\subsection{Proving Equation~\eqref{equation:update}}
Let $\prefHist = \concurState_0\concurAction_0\concurState_1\concurAction_1\ldots \concurState_{k}\concurAction_{k}\in (\concurStateSpace\concurActionSpace)^*$ such that $\proba(W_{k+1}=\prefHist) > 0$.
For any $\mealyState\in\mealyStateSpace$, the probability $\mealyDist{\prefHist}(\mealyState)$ is formalised by the conditional probability $\proba(\mealyStateSpace_{k+1}=m\mid W_{k+1} = \prefHist)$.
Henceforth, we assume that $\prefHist$ is non-empty.
Let $\prefHist' = \concurState_0\concurAction_0\ldots \concurState_{k-1}\concurAction_{k-1}$ be the prefix of $\prefHist$ without its last state and last action.
To prove Equation~\eqref{equation:update}, we must express $\mealyDist{\prefHist}$ as a function of $\mealyDist{\prefHist'}$.

We fix $\mealyState\in\mealyStateSpace$ for the remainder of the section.
The first step in our approach is to consider all possible paths in $\mealy$ that reach $\mealyState$ and have a positive probability of occurring whenever $\prefHist$ is processed by $\mealy$.
Considering these paths will allow us to exhibit the term in which $\update$ appears within Equation~\eqref{equation:update}.
We use the notation $\paths_{\prefHist}^{\mealyState}$ for the set of sequences $\mealyState_0\mealyState_1\ldots \mealyState_k$ such that the path $\mealyState_0\mealyState_1\ldots \mealyState_k\mealyState$ in $\mealy$ is compatible with $\prefHist$, i.e., we let
\[\paths_{\prefHist}^{\mealyState} = \{\mealyState_0\mealyState_1\ldots \mealyState_{k}\in \mealyStateSpace^{k+1}\mid \proba( \overline{\mealyStateSpace_{k+1}}= \mealyState_0\ldots \mealyState_{k}\mealyState\mid W_{k+1} = \prefHist)>0\}.\]
We define, for any memory state $\mealyState'\in\mealyStateSpace$, the set $\paths_{\prefHist'}^{\mealyState'}$ as a subset of $\mealyStateSpace^{k}$ in the same fashion.
 It follows from the law of total probability (formulated for conditional probabilities), that
\begin{align*}
  \mealyDist{\prefHist}(\mealyState)
  & =\proba(\mealyStateSpace_{k+1}=\mealyState \mid  W_{k+1} = \prefHist) \\
  &= \sum_{\overline{\mealyState}\mealyState'\in \paths_{\prefHist}^{\mealyState}}
    \proba(\mealyStateSpace_{k+1}=\mealyState\mid W_{k+1} = \prefHist\land \overline{\mealyStateSpace_k}=\overline{\mealyState}\mealyState')\cdot
    \proba( \overline{\mealyStateSpace_k}=\overline{\mealyState}\mealyState'\mid W_{k+1} = \prefHist) \\
  &= \sum_{\overline{\mealyState}\mealyState'\in \paths_{\prefHist}^{\mealyState}}
    \update(\mealyState', \concurState_k, \concurAction_k)(m) \cdot
    \proba( \overline{\mealyStateSpace_k}=\overline{\mealyState}\mealyState'\mid W_{k+1} = \prefHist) \\
  &= \sum_{\mealyState'\in\mealyStateSpace}\sum_{\overline{\mealyState}\in \paths_{\prefHist'}^{\mealyState'}}
    \update(\mealyState', \concurState_k, \concurAction_k)(m) \cdot
    \proba( \overline{\mealyStateSpace_k}=\overline{\mealyState}\mealyState'\mid W_{k+1} = \prefHist) \\
  &= \sum_{\mealyState'\in\mealyStateSpace}
    \left(
    \update(\mealyState', \concurState_k, \concurAction_k)(m)\cdot
    \sum_{\overline{\mealyState}\in\paths_{\prefHist'}^{\mealyState'}}
    \proba(\overline{\mealyStateSpace_k}=\overline{\mealyState}\mealyState'\mid W_{k+1} = \prefHist)
    \right).
\end{align*}

We now shift our focus to the inner sum.
Let us fix $\mealyState'\in\mealyStateSpace$.
This sum is indexed by all paths in $\mealy$ that reach $\mealyState'$ and have positive probability.
Therefore, it follows from the law of total probability that
\[\sum_{\overline{\mealyState}\in\paths_{\prefHist'}^{\mealyState'}}
  \proba(\overline{\mealyStateSpace_k}=\overline{\mealyState}\mealyState'\mid
  W_{k+1} = \prefHist) =
  \proba(\mealyStateSpace_k=\mealyState'\mid
  W_{k+1} = \prefHist).
\]

We underscore that this probability is not $\mealyDist{\prefHist'}(\mealyState') = \proba(\mealyStateSpace_k=\mealyState'\mid W_k= \prefHist')$.
Up to this point, we have established that
\begin{equation}\label{equation:conditional:3}
  \mealyDist{\prefHist}(\mealyState) =
  \sum_{\mealyState'\in\mealyStateSpace}
  \update(\mealyState', \concurState_k, \concurAction_k)(\mealyState)\cdot
  \proba(\mealyStateSpace_k=\mealyState'\mid W_{k+1} = \prefHist).
\end{equation}

Using Bayes' theorem, we can show a relation between the probability $\proba(\mealyStateSpace_k=\mealyState'\mid W_{k+1} = \prefHist)$ and $\mealyDist{\prefHist'}(\mealyState')$.
Let us write $\hist'$ in the following for the history $\prefHist'\concurState_k$ given by $\prefHist$ without its last action.
We note that $\{W_{k+1} = \prefHist\}$ and $\{H_{k} = \hist'\}\cap \{\concurActionSpace_{k} = \concurAction_k\}$ both denote the same set.
We have the following chain of equations:
\begin{align*}
  \proba(\mealyStateSpace_k=\mealyState'\mid
  & W_{k+1} = \prefHist) \\
  &= \proba(\mealyStateSpace_k=\mealyState'\land H_k = \hist'\mid
    W_{k+1} = \prefHist)\\
  &= \frac{\proba(W_{k+1} = \prefHist \mid
    \mealyStateSpace_k=\mealyState'\land H_k = \hist')\cdot
    \proba(\mealyStateSpace_k=\mealyState'\land H_k = \hist')}{\proba(W_{k+1} = \prefHist)} \\
  &= \frac{\proba(\concurActionSpace_k = \concurAction_k \mid
    \mealyStateSpace_k=\mealyState'\land H_k = \hist')\cdot
    \proba(\mealyStateSpace_k=\mealyState'\mid H_k = \hist')}{\proba(\concurActionSpace_k = \concurAction_k \mid H_k = \hist')}.
\end{align*}

The first equality is a consequence of $W_{k+1} = \prefHist$ implying $H_k=h'$. Bayes' theorem is used between lines two and three. To go from the third to the fourth line, both the numerator and denominator of the fraction have been multiplied by $\proba(H_k = h')$ and the definition of conditional probabilities has been used to rewrite the denominator and the rightmost factor of the numerator.

We now analyse the three terms of the fraction above.
The probability $\proba(\mealyStateSpace_k=\mealyState'\mid H_k = \hist')$ is equal to the probability $\proba(\mealyStateSpace_k=\mealyState'\mid W_k = \prefHist')$.
This is exactly $\mealyDist{\prefHist'}(\mealyState')$.

Next, we obtain from the independence of the action choices of both players and  how the action probabilities are computed that
\[\proba(\concurActionSpace_k = \concurAction_k \mid \mealyStateSpace_k=\mealyState'\land H_k = \hist') = \nextmove(\mealyState', \concurState_k)(\concurActionOne)\cdot\stratTwo(\hist')(\concurActionTwo).\]

We use the previous equation to analyse the final term of the quotient.
We have
\begin{align*}
  \proba(\concurActionSpace_k & = \concurAction_k \mid H_k  = h')\\
  &= \sum_{\substack{\mealyState''\in M\\\proba(\mealyStateSpace_k=\mealyState''\mid H_k=h')>0}} \proba(\concurActionSpace_k = \concurAction_k \mid \mealyStateSpace_k = \mealyState'' \land H_k = h')\cdot \proba(\mealyStateSpace_k = \mealyState''\mid H_k=h')  \\
  &=  \stratTwo(\hist')(\concurActionTwo_k)\cdot \sum_{\mealyState''\in M} \nextmove(\mealyState'', \concurState_k)(\concurActionOne_k)\cdot \mealyDist{\prefHist'}(\mealyState'').
\end{align*}

By injecting the above in Equation~\eqref{equation:conditional:3}, we directly obtain Equation~\eqref{equation:update} (note that any term appearing in a denominator is non-zero by the assumption $\proba(W_{k+1} = \prefHist)>0$).
This concludes the explanation on how to derive the formula to update the distribution over memory states following a state transition and the choice of a pair of actions.
 
\let\simpleGame\undefined
\let\checkAct\undefined
\let\continueAct\undefined

\end{document}